\documentclass[11pt]{article}

\usepackage[margin=1in]{geometry}
\usepackage{fullpage}
\usepackage{tgtermes}
\usepackage[T1]{fontenc}
\usepackage[colorlinks,citecolor=blue,linkcolor=blue,urlcolor=red]{hyperref}
\usepackage{graphicx}
\usepackage{amsfonts,amsmath,amsthm,amssymb,dsfont,mathtools,multirow}
\mathtoolsset{centercolon}
\usepackage{xfrac,nicefrac}
\usepackage{mathdots}
\usepackage{bm,bbm}
\usepackage{url}
\usepackage{caption}
\usepackage{paralist}
\usepackage{enumerate}
\usepackage[normalem]{ulem}
\usepackage{enumitem}
\usepackage{xspace}
\xspaceaddexceptions{]\}}
\usepackage[capitalise]{cleveref}
\crefname{algocf}{Algorithm}{Algorithms}
\Crefname{algocf}{Algorithm}{Algorithms}
\crefname{claim}{Claim}{Claims}
\usepackage{tabu}
\usepackage{framed}
\usepackage{float,wrapfig}
\usepackage{subfigure}
\usepackage{soul}
\usepackage[usenames,dvipsnames]{pstricks}
\usepackage{thmtools, thm-restate}
\usepackage[linesnumbered,boxed,ruled,vlined]{algorithm2e}
\usepackage{algpseudocode}
\usepackage{tikz}
\usepackage{authblk}
\usetikzlibrary{decorations.pathreplacing}
\usetikzlibrary{calc}
\usepackage{regexpatch}
\usetikzlibrary{positioning}
\usetikzlibrary{arrows.meta}

\usepackage{color}

\theoremstyle{plain}

\newtheorem{theorem}{Theorem}[section]
\newtheorem{lemma}[theorem]{Lemma}

\newtheorem{observation}[theorem]{Observation}

\newtheorem{corollary}[theorem]{Corollary}

\theoremstyle{definition}
\newtheorem{definition}[theorem]{Definition}

\newcommand{\too}[1]{\tilde{O}({#1})}
\newcommand{\abs}[1]{\vert{#1}\vert}
\newcommand{\new}[1]{\pi_{G'}(#1)}

\newcommand{\set}[1]{\{ #1 \}}

\newcommand{\og}[3]{\pi_{G-#3}\left(#1,#2\right)}
\newcommand{\hg}[3]{\pi_{H-#3}\left(#1,#2\right)}
\newcommand{\nng}[3]{\pi_{G'-#3}\left(#1,#2\right)}
\newcommand{\odg}[3]{\varpi_{G-#3}\left(#1,#2\right)}
\newcommand{\ndg}[3]{\varpi_{G'-#3}\left(#1,#2\right)}
\newcommand{\pp}[1]{T_{G',R}\left(#1\right)}

\newcommand{\zdd}[1]{{#1}^{l}}
\newcommand{\ydd}[1]{{#1}^{r}}

\begin{document}


\title{Undirected 3-Fault Replacement Path in Nearly Cubic Time}



\author[1]{Shucheng Chi \thanks{chisc21@mails.tsinghua.edu.cn}}
\author[1]{Ran Duan \thanks{duanran@mail.tsinghua.edu.cn}}
\author[2]{Benyu Wang \thanks{benyuw@umich.edu}}
\author[1]{Tianle Xie \thanks{xtl21@mails.tsinghua.edu.cn}}

\affil[1]{Institute for Interdisciplinary Information Sciences, Tsinghua University}
\affil[2]{University of Michigan, Ann Arbor}



\maketitle
\begin{abstract}
Given a graph $G=(V,E)$ and two vertices $s,t\in V$, the $f$-fault replacement path ($f$FRP) problem computes for every set of edges $F$ where $|F|\leq f$, the distance from $s$ to $t$ when edges in $F$ fail. A recent result shows that 2FRP in directed graphs can be solved in $\tilde{O}(n^3)$ time [Vassilevska Williams, Woldeghebriel, Xu 2022]. In this paper, we show a 3FRP algorithm in deterministic $\tilde{O}(n^3)$ time for undirected weighted graphs, which almost matches the size of the output. This implies that $f$FRP in undirected graphs can be solved in almost optimal $\tilde{O}(n^f)$ time for all $f\geq 3$.

To construct our 3FRP algorithm, we introduce an incremental distance sensitivity oracle (DSO) with $\tilde{O}(n^2)$ worst-case update time, while preprocessing time, space, and query time are still $\tilde{O}(n^3)$, $\tilde{O}(n^2)$ and $\tilde{O}(1)$, respectively, which match the static DSO [Bernstein and Karger 2009]. Here in a DSO, we can preprocess a graph so that the distance between any pair of vertices given any failed edge can be answered efficiently. From the recent result in [Peng and Rubinstein 2023], we can obtain an offline dynamic DSO from the incremental worst-case DSO, which makes the construction of our 3FRP algorithm more convenient. By the offline dynamic DSO, we can also construct a 2-fault single-source replacement path (2-fault SSRP) algorithm in $\tilde{O}(n^3)$ time, that is, from a given vertex $s$, we want to find the distance to any vertex $t$ when any pair of edges fail. Thus the $\tilde{O}(n^3)$ time complexity for 2-fault SSRP is also almost optimal.

Now we know that in undirected graphs 1FRP can be solved in $\tilde{O}(m)$ time [Nardelli, Proietti, Widmayer 2001], and 2FRP and 3FRP in undirected graphs can be solved in $\tilde{O}(n^3)$ time. In this paper, we also show that a truly subcubic algorithm for 2FRP in undirected graphs does not exist under APSP-hardness conjecture.
\end{abstract}

\thispagestyle{empty}
\setcounter{page}{1}
\pagestyle{plain}

\section{Introduction}


The shortest path problem is one of the most fundamental problems in computer science, and the single-source shortest path problem is known to have an almost linear $\tilde{O}(m)$ time algorithm~\cite{Dij59,DMSY23}. In a failure-prone graph $G=(V,E)$ ($n=|V|$, $m=|E|$), some edges may fail on the shortest path, so we want to find a \emph{replacement path}. In the classical replacement path (RP) problem, given source $s$ and destination $t$, we want to find the $s$-$t$ shortest path when edge $e$ fails for every edge $e\in E$. Of course, if $e$ is not on the original $s$-$t$ shortest path in $G$, the replacement path is still the original shortest path, so in fact we just need to find an $s$-$t$ replacement path for every edge $e$ on the $s$-$t$ shortest path in $G$, which counts for $O(n)$ replacement paths.

We can generalize the RP problem to any number of failed edges $f$, that is, finding $s$-$t$ shortest paths for every failed edge set $F$ which has $|F|\leq f$. This problem is called \emph{$f$-fault replacement path} ($f$FRP)~\cite{WWX22}, therefore 1FRP is just the original RP problem. As before, given the first $(f-1)$ failed edges and the corresponding $s$-$t$ replacement shortest path, we still only need to consider the cases when the $f$-th failed edge is on it. Because there are $O(n)$ edges on each shortest path, the total number of replacement paths we need to find for $f$FRP problem is $O(n^f)$. The current well-known results of $f$FRP algorithms for \emph{real-weighted} directed and undirected graphs are summarized as follows.

\begin{itemize}
    \item RP problem in directed graphs can be solved in $O(mn+n^2\log\log n)=O(n^3)$ time~\cite{GL09}.
    \item Under APSP-hardness conjecture\footnote{APSP is an abbreviation for all-pair shortest path problem, and APSP-hardness conjecture suggests that APSP cannot be solved in truly subcubic $O(n^{3-\epsilon})$ time for any constant $\epsilon>0$.}, RP in directed graphs does not have $O(n^{3-\epsilon})$ time algorithm for any constant $\epsilon>0$~\cite{WW18}. So the $O(n^3)$-time algorithm is (conditionally) nearly optimal.
    \item RP problem in undirected graphs can be solved in $\tilde{O}(m)$ time\footnote{Here $\tilde{O}(\cdot)$ hides polylogarithmic factors.}~\cite{NPW01}, thus also almost optimal.
    \item Recently, Vassilevska Williams, Woldeghebriel and Xu~\cite{WWX22} gave a 2FRP algorithm with $\tilde{O}(n^3)$ running time for directed graphs, almost matching the 1FRP case. (Although not formulated, it also works for undirected graphs with slight modifications.)
\end{itemize}

Thus, 1FRP and 2FRP problems both have $\tilde{O}(n^3)$-time algorithms, so it is natural to ask whether 3FRP still has a $\tilde{O}(n^3)$-time algorithm. In this paper, we give such a 3FRP algorithm for undirected graphs: (Note that all algorithms in this paper are deterministic.)

\begin{theorem}\label{thm:main1}
    The 3FRP problem in undirected real-weighted graphs can be solved in $\tilde{O}(n^3)$ time.
\end{theorem}

Denote the shortest path between $s$ and $t$ in graph $G$ by $\pi_G(s,t)$. Then in $\tilde{O}(n^3)$ time, for every edge $d_1\in \pi_G(s,t)$, we can find all edges of $\pi_{G-\{d_1\}}(s,t)$, then for every edge $d_2\in \pi_{G-\{d_1\}}(s,t)$, we can find all edges of $\pi_{G-\{d_1,d_2\}}(s,t)$, then for every edge $d_3\in \pi_{G-\{d_1,d_2\}}(s,t)$, we can find $|\pi_{G-\{d_1,d_2,d_3\}}(s,t)|$, so the algorithm outputs $O(n^3)$ distances in total. One can see the difficulty of this problem since every answer only takes $\tilde{O}(1)$ time on average.

For any $f\geq 3$, by Theorem~\ref{thm:main1}, the $f$FRP problem can be solved in $\tilde{O}(n^f)$ time in undirected graphs (see the reduction in~\cite{WWX22}). Since the size of the output of $f$FRP can be $\Theta(n^f)$, the running time is almost optimal. As in~\cite{WWX22}, we also take the 1-failure distance sensitivity oracle as an important subroutine, since it only takes $\tilde{O}(1)$ query time. Here an $f$-failure distance sensitivity oracle (DSO) is a data structure that supports distance queries between any pair of $u,v\in V$ given any $f$ failed edges. It is widely known that 1-failure DSO takes $\tilde{O}(n^2)$ space, $O(1)$ query time, and $\tilde{O}(mn)$ construction time~\cite{Demetrescu2008,2009A}. Since current 2-failure DSOs are still hard to construct efficiently~\cite{duan2009dual}, we instead construct an incremental 1-failure DSO: (In the following DSO means 1-failure DSO for convenience.)

\begin{theorem}
    For a given undirected graph $G$, there is an incremental DSO that can be constructed in $\tilde{O}(n^3)$ time, so that when we insert an edge $e$ into $G$, the DSO can be maintained in worst-case $\tilde{O}(n^2)$ time. The DSO also has $\tilde{O}(n^2)$ space and $\tilde{O}(1)$ query time.
\end{theorem}

Recently Peng and Rubinstein~\cite{peng2023fully} gave a reduction from offline fully dynamic structure to worst-case incremental structure, where ``offline'' means all updates are known ahead. So we can obtain an offline dynamic efficient DSO. (We also give a simple proof of the reduction we need by a different method as~\cite{peng2023fully}, see Theorem~\ref{thm:offline}.)

\begin{theorem} (\cite{peng2023fully})
    Let $T \geq 1$ be the total number of updates. If there exists an incremental dynamic algorithm with query time $\Gamma_q$ and worst-case update time $\Gamma_u$, then there is an offline dynamic algorithm with query time $\Gamma_q$ and worst-case update time $O(\Gamma_u \cdot \log^2(T))$.
\end{theorem}

\begin{corollary}\label{cor:dynamic}
    Given an undirected graph $G$ and a sequence of $O(n)$ edge insertions and deletions, there is an offline dynamic DSO which can be constructed in $\tilde{O}(n^3)$ time, and the total update time is also $\tilde{O}(n^3)$. It can answer the distance between any pair of vertices when any edge fails in $\tilde{O}(1)$ time.
\end{corollary}

We have known 1FRP is as hard as APSP in directed graphs~\cite{WW18}, but in undirected graphs, 1FRP can be solved in $\tilde{O}(m)$ time~\cite{NPW01}, and 2FRP in undirected graphs can be solved in $\tilde{O}(n^3)$ time. One may wonder whether 2FRP in undirected graphs has a truly subcubic time algorithm. However, we show that it is not possible under the APSP-hardness conjecture.

\begin{theorem}
    Assuming the APSP-hardness conjecture that APSP cannot be solved in truly subcubic $O(n^{3-\epsilon})$ time for any constant $\epsilon>0$, then 2FRP problem in undirected weighted graphs cannot be solved in truly subcubic time.
\end{theorem}

We can also apply the offline dynamic DSO to get a 2-fault single-source replacement path (2-fault SSRP) algorithm for undirected graphs, that is, given a source $s$, for every other vertex $t\in V$, we can find $\pi_{G-\{d_1,d_2\}}(s,t)$ for all edges $d_1,d_2$. The reduction is very simple: we can remove an edge $d_1$ from the shortest path tree from $s$ each time, and then put it back. By the offline dynamic DSO, we can answer the distance between $s$ and $t$ avoiding $d_2$ in the current graph without $d_1$.

\begin{theorem}\label{thm:SSRP1}
    The 2-fault single-source replacement path problem can be solved in $\tilde{O}(n^3)$ time for undirected graphs.
\end{theorem}

Note that Theorem~\ref{thm:SSRP1} can be extended to undirected $f$-fault SSRP algorithm in $\tilde{O}(n^{f+1})$ time for all $f\geq 2$, which is almost optimal. Previously there are 1-fault single-source replacement path algorithms of running time $\tilde{O}(m\sqrt{n}+n^2)$ for unweighted undirected graphs~\cite{CC19,DG22} and unweighted directed graphs~\cite{CM20}.

Note that the 2FRP, 3FRP, and single-source 2FRP algorithms in this paper obtain every distance from DSOs, which are essentially similar to the one in~\cite{2009A}, or the APSP table. By the properties of the DSOs and the construction of our algorithms, in these algorithms in fact we can obtain an oracle of size $\tilde{O}(n^3)$ in which we can retrieve a shortest path under failures in $O(1)$ time per edge.

\paragraph{Other related work.}

The current fastest running time for the directed all-pair shortest path (APSP) problem is $n^3/2^{\Omega(\sqrt{\log n})}$ \cite{Williams14}. Finding truly subcubic time algorithms for APSP with arbitrary weights is considered a major open problem in algorithm research, and its hardness is one of the major assumptions in proving conditional lower bounds.

For replacement path (RP) and single-source replacement path (SSRP) problems, there are also many subcubic time results for graphs with small integer edge weights. In unweighted directed graphs, the RP problem can be solved in $\tilde{O}(m\sqrt{n})$ time~\cite{ACC19}. If the edge weights are integers in $[-M,M]$, Vassilevska Williams gave a $\tilde{O}(Mn^{\omega})$ time RP algorithm~\cite{Williams14}, where $\omega<2.371339$ is the exponent of the complexity of matrix multiplication~\cite{alman2024,VXXZ24,DWZ23}. Moreover, \cite{WWX22} gave a 2FRP algorithm in small edge weight directed graphs in $\tilde{O}(M^{2/3}n^{2.9153})$ time. For SSRP problem, there is also a $\tilde{O}(Mn^{\omega})$ time algorithm for graphs with integer weights in $[1,M]$ and algorithms with running time $\tilde{O}(M^{0.7519}n^{2.5286})$ and $\tilde{O}(M^{0.8043}n^{2.4957})$ for graphs with integer weights in $[-M,M]$.~\cite{GV20,GV12,GPVX21}.

After the breakthrough result of efficient 1-failure DSO by Demetrescu, Thorup, Chowdhury and Ramachandran~\cite{Demetrescu2008}, there are many efforts to improve the preprocessing time~\cite{BK08,2009A,GV12,GV20,CC19,REN22,GR21}, space~\cite{DZ17}, and extend to more failures. However, keeping query time $\tilde{O}(1)$ and space $\tilde{O}(n^2)$ is difficult when generalized to any constant number of $f$ failures. The 2-failure DSO of $\tilde{O}(n^2)$ space and $\tilde{O}(1)$ query time~\cite{duan2009dual} is much more complicated than the 1-failure case, so it seems impossible to generalize it to $f$-failure DSO in arbitrary graphs. For undirected graphs, recently Dey and Gupta gave an $f$-failure DSO with $O(f^4n^2\log^2(nM))$ space and $O((f\log(nM))^{O(f^2)})$ query time~\cite{DG24_stoc}, improving the result in~\cite{DR22}, but their construction time is still large, thus still not suitable for efficient $f$FRP algorithms. 

\paragraph{Organization.}

In Section~\ref{sec:prelim} we introduce basic notations, assumptions and concepts. In Section~\ref{sec:overview} we give a brief description of the 2FRP algorithm for undirected graphs, which is a little simpler than the one for directed graphs in~\cite{WWX22}, then give an overview on how to extend it to 3FRP algorithm. In Section~\ref{sec-inc} the incremental DSO is discussed, which is the crucial part of our 3FRP algorithm. The algorithms for the cases that 1 edge, 2 edges, and 3 edges of the set of failed edges is/are on the original shortest path $\pi_G(s,t)$ be given in Section~\ref{sec:1-edge},~\ref{sec:2-edge},~\ref{sec4}, respectively. In Section~\ref{sec:hardness} we prove the hardness of undirected 2FRP under APSP-hardness conjecture, and in Section~\ref{sec:2ssrp} the 2-fault single-source replacement path algorithm is discussed.
\section{Preliminaries}\label{sec:prelim}

\subsection{Notations and Assumptions}

We consider an undirected weighted graph $G=(V,E)$ and $n=|V|, m=|E|$. For any two vertices $u,v \in V$, we use $\pi_G(u,v)$ to denote the shortest path from $u$ to $v$ in $G$. If the graph is clear in the context, then we use $uv$ to be short for $\pi_G(u,v)$. For a path $P$, we use $|P|, \| P \|$ to denote the length of $P$ and the number of edges in $P$, respectively.

If the context is clear that a vertex $z$ is on the shortest path $uv$ from $u$ to $v$ in a graph $G$, we use $z \oplus i$ to represent the vertex that is $i$ vertices after $z$ in the direction from $u$ to $v$. Similarly, we define $z \ominus i$ to represent the vertex that is $i$ vertices before $z$. If $x$ and $y$ are two vertices on the shortest path from $u$ to $v$ in graph $G$, we use the notation $x>y$ to say that the number of edges between $u$ and $x$ is larger than that between $u$ and $y$, and the vertices $u,v$ should be clear in the context. Similarly $x<y$ means $\|ux\|<\|uy\|$, $x\le y$ means $\|ux\|\le\|uy\|$, $x\ge y$ means $\|ux\|\ge\|uy\|$. 

Also when $x,y\in \pi_G(u,v)$ we say $\pi_G(x,y)$ is a subpath of $\pi_G(u,v)$. If the context is clear we use the interval $[x,y]$ for $x\leq y$ to denote the subpath $xy$ of $uv$ in graph $G$. 
Thus the interval between $i$-th vertex and $j$-th vertex after $u$ on $uv$ is denoted as $R=[u \oplus i, u\oplus j]$.

Let $P_1,P_2$ be two paths. If they share a same endpoint, we use $P_1\circ P_2$ to denote the concatenation of them. If they do not share any edge, we say $P_1$ avoids $P_2.$ Sometimes for brevity, we use $\min\{P_1,P_2\}$ to denote the shorter path of $P_1,P_2$, which means if $\abs{P_1}\le \abs{P_2}$, $\min\{P_1,P_2\}=P_1$  and otherwise $\min\{P_1,P_2\}=P_2$.

In this paper, for a path $P$ in $G$, $G-P$ denotes the graph obtained by removing \textbf{edges} of $P$ from $G$. (We do not need to remove any vertex in this paper.)

As in many distance oracle papers, we also use the \textbf{unique shortest path assumption}, which means the shortest path between any two vertices in any subgraph of $G$ is unique. For the incremental DSO, we also make the unique shortest path assumption at any time. If this is not the case, we can break the ties by adding small random variables to the edge weights. (See~\cite{DI04}.) So for example, we can check whether an edge $e=(x,y)$ is on $st$ or not by checking whether $|sx\circ (x,y)\circ yt|=|st|$ or $|sy\circ (y,x)\circ xt|=|st|$.


\subsection{Replacement Paths}

We first consider the divergence and convergence points of two paths: (As before, although our graph is undirected, we may assume a path $P$ starts from a vertex $u$ and ends at $v$.)

\begin{itemize}
    \item For two paths $P,Q$ both starting from $u$, we say they diverge at vertex $a \in P\cap Q$ if $Q$ first gets into an edge out of $P$ from $a$, that is, the subpaths of $P$ and $Q$ before $a$ coincide, and we denote this divergence point $a$ by $\Delta(P,Q)$. If $P,Q$ are both shortest paths, their divergence point is unique.
    
    \item Symmetrically, for $P,Q$ both ends at $v$, we say they converges at vertex $b \in P\cap Q$ if $Q$ gets into an edge in $P$ from $b$, and the subpaths of $P$ and $Q$ coincide after $b$. The convergence point $b$ is denoted by $\nabla(P,Q)$. 

    \item For two shortest paths $P,Q$ which do not share endpoints, if they intersect, the intersection part is also a shortest path, that is, they must first converge and then diverge.
\end{itemize}

The following theorem is well known for replacement paths. 

\begin{theorem}(\cite{Demetrescu2008})\label{thm4-2}
    Let $l$ be a 1-fault replacement path in $G$, i.e. $l=\og{u}{v}{d}$ for an edge $d$. Let $R_1,R_2$ be two subpaths in $uv$ with no intersection, and $R_3$ be the subpath between them. If $l$ avoids both $R_1$ and $R_2$, then it avoids $R_3.$
\end{theorem}
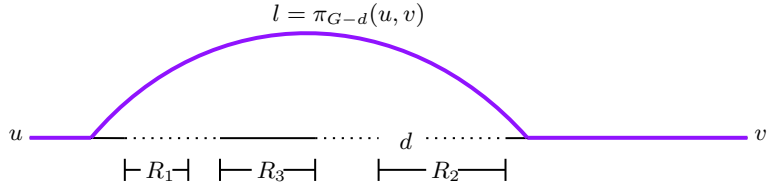
\begin{center}
    \tikzset{every picture/.style={line width=0.75pt}} 

\begin{tikzpicture}[x=0.75pt,y=0.75pt,yscale=-1,xscale=1]

\draw    (96,223.82) -- (144,224) ;
\draw    (336,224) -- (458,223.82) ;
\draw [color={rgb, 255:red, 144; green, 19; blue, 254 }  ,draw opacity=0.5 ][line width=1.5]    (127,223.82) .. controls (187.2,152.62) and (285.2,154.22) .. (347,223.82) ;
\draw [color={rgb, 255:red, 144; green, 19; blue, 254 }  ,draw opacity=0.5 ][line width=1.5]    (96,223.82) -- (127,223.82) ;
\draw [color={rgb, 255:red, 144; green, 19; blue, 254 }  ,draw opacity=0.5 ][line width=1.5]    (347,223.82) -- (458,223.82) ;
\draw    (192,224) -- (240,224) ;
\draw  [dash pattern={on 0.84pt off 2.51pt}]  (296,224) -- (336,224) ;
\draw  [dash pattern={on 0.84pt off 2.51pt}]  (144,224) -- (192,224) ;
\draw  [dash pattern={on 0.84pt off 2.51pt}]  (240,224) -- (272,224) ;
\draw    (152,240) -- (144,240) ;
\draw [shift={(144,240)}, rotate = 360] [color={rgb, 255:red, 0; green, 0; blue, 0 }  ][line width=0.75]    (0,5.59) -- (0,-5.59)   ;
\draw    (168,240) -- (176,240) ;
\draw [shift={(176,240)}, rotate = 180] [color={rgb, 255:red, 0; green, 0; blue, 0 }  ][line width=0.75]    (0,5.59) -- (0,-5.59)   ;
\draw    (296,240) -- (272,240) ;
\draw [shift={(272,240)}, rotate = 360] [color={rgb, 255:red, 0; green, 0; blue, 0 }  ][line width=0.75]    (0,5.59) -- (0,-5.59)   ;
\draw    (312,240) -- (336,240) ;
\draw [shift={(336,240)}, rotate = 180] [color={rgb, 255:red, 0; green, 0; blue, 0 }  ][line width=0.75]    (0,5.59) -- (0,-5.59)   ;
\draw    (208,240) -- (192,240) ;
\draw [shift={(192,240)}, rotate = 360] [color={rgb, 255:red, 0; green, 0; blue, 0 }  ][line width=0.75]    (0,5.59) -- (0,-5.59)   ;
\draw    (224,240) -- (240,240) ;
\draw [shift={(240,240)}, rotate = 180] [color={rgb, 255:red, 0; green, 0; blue, 0 }  ][line width=0.75]    (0,5.59) -- (0,-5.59)   ;

\draw (281,219.4) node [anchor=north west][inner sep=0.75pt]  [font=\footnotesize]  {$d$};
\draw (460,218.4) node [anchor=north west][inner sep=0.75pt]  [font=\footnotesize]  {$v$};
\draw (84,218.4) node [anchor=north west][inner sep=0.75pt]  [font=\footnotesize]  {$u$};
\draw (217,154.4) node [anchor=north west][inner sep=0.75pt]  [font=\footnotesize]  {$l=\pi _{G-d}( u,v)$};
\draw (153,234.4) node [anchor=north west][inner sep=0.75pt]  [font=\footnotesize]  {$R_{1}$};
\draw (297,234.4) node [anchor=north west][inner sep=0.75pt]  [font=\footnotesize]  {$R_{2}$};
\draw (209,234.4) node [anchor=north west][inner sep=0.75pt]  [font=\footnotesize]  {$R_{3}$};

\end{tikzpicture}
\end{center}

From this theorem, we can see for any $f$-fault replacement path $\og{u}{v}{F}$ where only one failed edge in $F$ is on the original shortest path $uv$, then $\og{u}{v}{F}$ only diverges and converges once on $uv$. And if two failed edges $d_1,d_2$ in $F$ are on $uv$, they will cut $uv$ into three intervals $D_1,D_2,D_3$. Because $D_1,D_2,D_3$ are all shortest paths in $G-F$, $\og{u}{v}{F}$ maybe:
\begin{itemize}
    \item diverges at $D_1$ and converges at $D_3$
    \item diverges at $D_1$, converges at some point in $D_2$, diverges from $D_2$, then converges at $D_3$
\end{itemize}


In undirected weighted graphs, a notable theorem on the structure of replacement paths says that a 1-fault replacement path $\og{u}{v}{f}$ is a concatenation of a shortest path $ux$, an edge $e=(x,y)$, and a shortest path $yv$, for some vertices $x$ and $y.$ (Any part can degenerate to a point.) More generally,

\begin{theorem} (\cite{2001Restoration}) \label{ReplacementPath}
For any undirected weighted graph $G$ and any set of failed edges $F$ with $|F|=f$, an $f$-fault replacement path can be partitioned into a concatenation interleaving at most $f+1$ subpaths and $f$ edges, where each subpath is a shortest path in the graph $G$. 
\end{theorem}

From this theorem, we have:
\begin{lemma}(\cite{BodwinGPW17})\label{thm:n-edges}
    The union of all 1-fault $u$-$v$ replacement paths has $O(n)$ edges in total in undirected graphs.
\end{lemma}
\begin{proof}
    From Theorem~\ref{ReplacementPath}, a 1-fault $u$-$v$ replacement path in undirected graphs is composed of a shortest path $ux$, an edge $(x,y)$, and a shortest path $yv$. The total number of the middle edges $(x,y)$ is bounded by $O(n)$ since there are $O(n)$ 1-fault replacement paths, and other edges are on the shortest path trees from $u$ and $v$, thus there are also at most $O(n)$ edges.
    
\end{proof}

\subsection{DSOs}\label{sec:intro-dso}

The paper uses the 1-fault DSO given by Bernstein and Karger (\cite{2009A}) to initialize the incremental DSO. While \cite{2009A} gives a 1-fault DSO in directed weighted graphs, it also works for undirected weighted graphs by a simple reduction. Suppose $G$ is an undirected weighted graph, we introduce a directed graph $G'$ by replacing each edge in $G$ by a pair of directed edges with the same weight. Suppose that in graph $G$ edge $(x,y)$ on the shortest path $st$ is removed ($x$ is closer to $s$), then we remove the directed edge $(x,y)$ on the shortest path from $s$ to $t$ in $G'.$ We can see that the 1-fault replacement path $\pi_{G'-(x,y)}(s,t)$ cannot go through edge $(y,x).$ This is because that no edge on $\pi_{G'}(s,x)$ is removed, so $\pi_{G'-(x,y)}(s,x)=\pi_{G'}(s,x)$ does not go through $(y,x).$ Therefore, $\pi_{G-(x,y)}(s,t)=\pi_{G'-(x,y)}(s,t).$


There is a theorem in \cite{peng2023fully} to view incremental structures as offline fully dynamic structures.

\begin{theorem} (\cite{peng2023fully})\label{thm1-2}
Let $T \geq 1$ be the total number of updates. Suppose there exists an incremental dynamic algorithm with query time $\Gamma_q$ and worst-case update time $\Gamma_u$, then there is a dynamic algorithm for deletions-look-ahead setting with query time $\Gamma_q$ and worst-case update time $O(\Gamma_u \cdot \log^2(T))$.
\end{theorem}

This theorem can be used to utilize our incremental DSO in Section~\ref{sec-inc} as an offline fully dynamic DSO, and here we also prove the reduction we need in a self-contained way:

\begin{theorem}\label{thm:offline}
    Starting from a graph of $O(n)$ vertices, if we have an incremental DSO of $\tilde{O}(n^2)$ space with preprocessing time $\tilde{O}(n^3)$, worst case update time $\tilde{O}(n^2)$ and query time $\tilde{O}(1)$, then it can be transformed into an offline fully dynamic DSO, such that if the total number of updates is $T=O(n)$, the total preprocessing and update time is $\tilde{O}(n^3)$ and the query time is still $\tilde{O}(1)$.
\end{theorem}
\begin{proof}
    In the offline fully dynamic DSO, let the initial graph be $G_0$ and the graph after the $i$-th update be $G_i$. We first list all graphs $G_0,G_1,\cdots,G_T$. For an interval of integers $[i,j]$, define the graph $G_{[i,j]}=G_i\cap G_{i+1}\cap\cdots\cap G_j$. Then for an interval $[i',j']\subseteq [i,j]$, we know that $G_{[i,j]}\subseteq G_{[i',j']}$ and $|G_{[i',j']}-G_{[i,j]}|\leq j-i$.

    Then we build a binary range tree on $[0,T]$, in which the root contains the graph $G_{[0,T]}$, the children of the root contain the graphs $G_{[0,\lfloor T/2\rfloor]}$ and $G_{[\lfloor T/2\rfloor,T]}$, respectively, and each leaf contains a graph $G_i$. So we can build the DSOs on this tree from the root using the incremental structure, since $G_R\subseteq G_{R'}$ if $R'$ is a child of $R$. Then we can see the total construction time is $\tilde{O}(n^3)$, and we can use the DSOs on leaves to answer queries.
\end{proof}


\section{Technical Overview}\label{sec:overview}

In this section, we first give a brief description of 1FRP and 2FRP algorithms for undirected graphs, which can list all edges of each path in $\tilde{O}(n^3)$ time, so that we can find the third failed edge in our 3FRP algorithm. Then the high-level ideas for the 3FRP algorithm and incremental DSO are discussed. 

\subsection{2FRP algorithm between $s,t$}\label{sec:2FRP}

All 1-fault replacement paths are easy to find in $\tilde{O}(mn)$ time since we can delete each edge in $st$ and then run Dijkstra algorithm~\cite{Dij59} from $s$. Then for edges $d_1\in st$ and $d_2\in \og{s}{t}{d_1}$, consider whether $d_2$ is on $st$ or not.

\subsubsection{Only $d_1$ is on $st$}\label{sec:one-2FRP}

By the methods in~\cite{WWX22}, we create a graph $H$ from $G-st$ by adding new vertices as follows. Let $N$ be a number that is large enough. (Namely, $N$ can be the sum of all edge weights in $G$.)
\begin{itemize}
    \item First let $H$ be a copy of $G-st$, that is, remove all edges on $\pi_G(s,t)$ from $G$.
    \item For every edge $d=(x,y)\in st$ and $x$ is before $y$ on $st$, introduce two new vertices $d^-$ and $d^+$. 
    \begin{itemize}
        \item For every vertex $x'$ on $sx$, the node $d^-$ is connected with $x'$ by an edge with weight $|sx'|+N$.
        \item For every vertex $y'$ on $yt$, the node $d^+$ is connected with $y'$ by an edge with weight $|y't|+N$.
    \end{itemize}
\end{itemize}

The number of vertices in $H$ is still $O(n)$, so we construct a 1-failure DSO on $H$ with $\tilde{O}(n^3)$ construction time, $\tilde{O}(n^2)$ space and $O(1)$ query time by~\cite{2009A}. We say the edge $(d^-,x')$ corresponds to the path $sx'$ since their length differs by a fixed number $N$, similarly $(y',d^+)$ corresponds to the path $y't$. (Here $N$ guarantees that any path between new vertices will not travel through other new vertices.) Since $d_2$ is not on $st$, we can obtain $\og{s}{t}{\set{d_1,d_2}}$ by querying the DSO. We have:
\begin{lemma} \label{lemma:1edge}
    Given $G$ and $d_1\in st$, $d_2\notin st$, in the graph $H$ we defined,
    $$ |\og{s}{t}{\set{d_1, d_2}}| = |\hg{d_1^-}{d_1^+}{d_2}| - 2N$$
\end{lemma}

\begin{proof}
The proof is easy to see since the first edge $(d_1^-,x')$ and last edge $(y',d_1^+)$ on the path $\hg{d_1^-}{d_1^+}{d_2}$ have:
$$|(d_1^-,x')|=|sx'|+N, \; |(y',d_1^+)|=|y't|+N$$
Also $x'$ is before $d_1$ and $y'$ is after $d_1$, so on $\hg{d_1^-}{d_1^+}{d_2}$ the first edge and last edge can be replaced by the original shortest paths $s'x$ and $y't$, respectively, so it is a path in $G-\{d_1,d_2\}$, and $|\og{s}{t}{\set{d_1, d_2}}| \leq |\hg{d_1^-}{d_1^+}{d_2}| - 2N$. Similarly the subpaths of $\og{s}{t}{\set{d_1, d_2}}$ before divergence point $x'$ and after convergence point $y'$ on $st$ can be replaced by edges $(d_1^-,x')$ and $(y',d_1^+)$ in $H$, respectively (with value $2N$ added), so $|\og{s}{t}{\set{d_1, d_2}}| \geq |\hg{d_1^-}{d_1^+}{d_2}| - 2N$.
\end{proof}
Since the 1-failure DSO supports path retrieving~\cite{Demetrescu2008, 2009A} in $O(1)$ time per edge, we can restore all paths $\og{s}{t}{\set{d_1, d_2}}$ in $\tilde{O}(n^3)$ time.

\subsubsection{Only $d_1,d_2$ are on $st$}\label{sec:two-2FRP}

Our structure in this case is also similar to~\cite{WWX22}, but is simpler since $G$ is undirected here. If we use the graph $H$ before, the distance $\pi_H(d_1^-,d_2^+)$ cannot capture $\og{s}{t}{\set{d_1, d_2}}$ since it only formulates the subpaths before divergence point and after convergence point on $st$, but $\og{s}{t}{\set{d_1, d_2}}$ may go through some edges between $d_1$ and $d_2$ on $st$, which are not included in $H$. 

W.l.o.g. assume $d_1=(x_1,y_1)$ is before $d_2=(x_2,y_2)$ on $st$, so $s\leq x_1<y_1\leq x_2<y_2\leq t$ on the shortest path $st$. For all pairs of $d_1,d_2$, we want to find the paths: ($w,a,b,z$ are all on $st$.)
$$P(d_1,d_2)=\min_{w\leq x_1, y_1\leq a\leq b\leq x_2, z\geq y_2} \{sw\circ \pi_{G-st}(w,a)\circ ab\circ \pi_{G-st}(b,z)\circ zt\}$$
and also:
$$P'(d_1,d_2)=\min_{w\leq x_1, y_1\leq b\leq a\leq x_2, z\geq y_2} \{sw\circ \pi_{G-st}(w,a)\circ ab\circ \pi_{G-st}(b,z)\circ zt\}$$

First we run the APSP algorithm on the graph $G-st$, then find the paths $U(d_1,a)=\min_{w\leq x_1}\{sw\circ \pi_{G-st}(w,a)\}$ for all edge $d_1=(x_1,y_1)$ and vertex $a\geq y_1$ on $st$ in $O(n^3)$ time, since we need $O(n)$ time to enumerate $w$ for every pair of $d_1$ and $a$. Symmetrically also find $U'(d_2,b)=\min_{z\geq y_2}\{\pi_{G-st}(b,z)\circ zt\}$ for all edge $d_2=(x_2,y_2)$ and vertex $b\leq x_2$ on $st$ in $O(n^3)$ time. Then for every pair of $d_1=(x_1,y_1),d_2=(x_2,y_2)$ and $y_1\leq b\leq x_2$, we can find:
$$W(d_1,d_2,b)=\min_{w\leq x_1, b\leq a\leq x_2} \{sw\circ \pi_{G-st}(w,a)\circ ab\}$$

That is, the path diverges before $d_1$, converges with $st$ at some point $a$ between $b$ and $x_2$, then goes through subpath of $st$ to $b$. If $b=x_2$, it is just $U(d_1,x_2)$. Then we can compute $W(d_1,d_2,b)$ for $b$ from $x_2$ backward to $y_1$. If $(b,b')$ with $b<b'$ is an edge on $st$, the path $W(d_1,d_2,b)$ can either first go through $W(d_1,d_2,b')$ or directly go to $b$, so $W(d_1,d_2,b)=\min\{U(d_1,b), W(d_1,d_2,b')\circ (b,b')\}$. So the total time to get all of $W(d_1,d_2,b)$ is still $O(n^3)$.

Thus, given $d_1,d_2$, we can get $P'(d_1,d_2)$ from $W(d_1,d_2,b)\circ U'(d_2,b)$ by enumerating all possible $b$ between $d_1,d_2$ in $O(n)$ time. $P(d_1,d_2)$ can be obtained similarly. So the total time is $O(n^3)$.

\subsection{Ideas of 3FRP algorithm}

As the 2FRP algorithm in Section~\ref{sec:2FRP}, we also consider how many of $d_1,d_2,d_3$ are on the original shortest path $st$, and use different methods to solve them.


\paragraph{Only $d_1$ is on $st$.} (Although there is 2-failure DSO of $\tilde{O}(n^2)$ space and $\tilde{O}(1)$ query time~\cite{duan2009dual}, it is hard to utilize it here since its construction time is too large. Here we build an incremental 1-failure DSO instead.) To solve the case that only $d_1$ is on $st$ but $d_2,d_3$ are not, note that if $d_2\in \pi_{G-d_1}(s,t)$, from Lemma~\ref{thm:n-edges}, the number of possible ``second failure'' $d_2$ is $O(n)$. If we have a dynamic 1-failure DSO with $\tilde{O}(n^2)$ update time when updating an edge, then for every possible $d_2$ we can delete it from $H$ temporarily and query $|\pi_{(H-d_2)-d_3}(d_1^-, d_1^+)|$ for every $d_1$ and $d_3$, which equals $|\og{s}{t}{\set{d_1, d_2,d_3}}|+2N$. After that add $d_2$ back to $H$. Although fully dynamic DSO may be difficult, in our algorithm we only need an offline dynamic version, and by the reduction from offline dynamic structure to worst-case incremental structure in~\cite{peng2023fully} and Theorem~\ref{thm:offline}, we construct an incremental DSO with worst-case $\tilde{O}(n^2)$ update time instead.

\paragraph{Only $d_1,d_2$ are on $st$.} Even with the incremental DSO, it is not very easy to solve the case that $d_1,d_2$ are on $st$ but $d_3$ is not. 
So we apply the incremental DSO on the binary partition structure as in the 2-failure DSO of~\cite{duan2009dual}, that is, we build a binary range tree on the path $st$, and define two graphs $H_{i,0}$ and $H_{i,1}$ (like the graph $H$ in Section~\ref{sec:one-2FRP}) on each level $i$ of the range tree, where in $H_{i,0}$ we remove all odd ranges on level $i$ and in $H_{i,1}$ we remove all even ranges. Given $d_1$ and $d_2$, we find the first level on which $d_1$ and $d_2$ are in different ranges, then $d_1$ is in an even range and $d_2$ is in an adjacent odd range, and let $m$ be the point separating these two ranges. Consider the following cases on the intersection of $\og{s}{t}{\set{d_1, d_2,d_3}}$ and the edges between $d_1$ and $d_2$ on $st$.
\begin{itemize}
    \item If $\og{s}{t}{\set{d_1, d_2,d_3}}$ does not travel through edges between $d_1$ and $d_2$, we can use DSO on $H$ to give the answer.
    \item If $\og{s}{t}{\set{d_1, d_2,d_3}}$ only goes through edges between $d_1$ and $m$, we can use $H_{i,0}$ which does not contain edges between $m$ and $d_2$. So the distance between $d_1^-$ and $d_2^+$ in $H_{i,0}-d_1$ avoiding $d_3$ can give the answer. (Note that edges between $s$ and $d_1$ in $H_{i,0}$ will not affect the answer, similarly for edges between $d_2$ and $t$.)
    \item If $\og{s}{t}{\set{d_1, d_2,d_3}}$ only goes through edges between $m$ and $d_2$, symmetrically we can use $H_{i,1}-d_2$.
    \item The only case left is that $\og{s}{t}{\set{d_1, d_2,d_3}}$ travels through $m$, we can also use the offline dynamic DSO on $H_{i,0}-d_1$ and $H_{i,1}-d_2$ to answer $|\og{s}{m}{\set{d_1, d_2,d_3}}|+|\og{m}{t}{\set{d_1, d_2,d_3}}|$.
\end{itemize}

So we need to build offline dynamic DSOs on $H_{i,0}$ and $H_{i,1}$ in each level $i$. Since there are $O(\log n)$ levels, the total time is still $\tilde{O}(n^3)$.

\paragraph{All of $d_1,d_2,d_3$ are on $st$.} We assume the order of them on $st$ is just $d_1,d_2,d_3$, then they will cut $st$ into four ranges $D_1,D_2,D_3,D_4$ in order. The difficulty will come from the fact that the path $\og{s}{t}{\set{d_1, d_2,d_3}}$ can go through both of $D_2$ and $D_3$, in any order. To solve this case, we first build a structure in $\tilde{O}(n^3)$ time which can answer the following query in $\tilde{O}(1)$ time: 
\begin{itemize}
    \item Given two vertex-disjoint ranges $R_1,R_2$ of consecutive edges on the original shortest path $st$ ($R_1, R_2$ can be single vertices), answer the shortest path that starts from the leftmost (rightmost) vertex of $R_1$, diverges from $R_1,$ converges in $R_2$, and finally ends at the leftmost (rightmost) vertex of $R_2$.
    \item Given an edge $d_1$ on $st$ and two vertex-disjoint ranges $R_1,R_2$ of consecutive edges on $st$ after $d_1$, answer the shortest path that starts from $s$, diverges before $d_1$, converges in $R_1$, and diverges from $R_1,$ then converges in $R_2,$ and finally ends at the leftmost (rightmost) point of $R_2$. 
\end{itemize}


Note that in the structure we store the paths for all ranges $R_1,R_2$ in the binary range tree, then given any $R_1,R_2$ in the query, they can be split into $O(\log n)$ ranges in the binary range tree.

Similar to \cref{sec:two-2FRP}, by this structure, when given $d_1,d_2,d_3$ and a vertex $b$ on $st$, it takes $\tilde{O}(1)$ time to find the shortest path goes through $b$ avoiding $d_1,d_2,d_3$. So it will take $\tilde{O}(n)$ time to find the shortest path goes through both of $D_2$ and $D_3$ avoiding $d_1,d_2,d_3$, but we need the time bound to be $\tilde{O}(1)$ on average. For every pair of $d_1$ and $d_3$, the main ideas are as follows.

\begin{itemize}
    \item[-] First find the middle edge $e_1$ on the range between $d_1$ and $d_3$, and use $\tilde{O}(n)$ time to find the shortest path $\og{s}{t}{\set{d_1,d_3,e_1}}$.
    \item[-] For other edges $d_2$ between $d_1$ and $d_3$, if $d_2$ is not on $\og{s}{t}{\set{d_1,d_3,e_1}}$, then either $\og{s}{t}{\set{d_1, d_2,d_3}}$ is equal to $\og{s}{t}{\set{d_1,d_3,e_1}}$ or it goes through $e_1$, so this case can be solved in $\tilde{O}(1)$ time.
    \item[-] The intersection of $\og{s}{t}{\set{d_1,d_3,e_1}}$ and $st$ between $d_1$ and $d_3$ consists of at most two ranges, then find the middle edge $e_2$ of the larger range, and find the shortest path avoiding $d_1,d_3,e_1,e_2$ and goes through two ranges between them. This can also be done in $\tilde{O}(n)$ time.
    \item[-] As before, for $d_2$ not on the new path we can solve it by querying the path goes through $e_1$ and $e_2$, so next we only need to consider the edges in the intersection of all paths so far, and we can argue that the size of the intersection between $d_1$ and $d_3$ will shrink by a constant factor in every two iterations, so only $O(\log n)$ iterations are needed. Thus the time for every pair of $d_1$ and $d_3$ is only $\tilde{O}(n)$. 
\end{itemize}

\subsection{Ideas of incremental DSO}

To maintain APSP incrementally, when inserting an edge $e=(x,y)$, for every pair of vertices $u,v$ we just need to check whether $e$ can make their distance shorter, that is, the new distance will be $\min\{|uv|,|ux|+|e|+|yv|,|uy|+|e|+|xv|\}$, thus update time is $O(n^2)$. Similar to previous works~\cite{Demetrescu2008}, we want to make a DSO structure that maintains the following for all pairs of vertices $u,v$.
\begin{itemize}
    \item Shortest path tree from every vertex.
    \item $\og{u}{v}{(u \oplus i)(v \ominus j)}$ for all $i, j\in\mathbb{N}$ that are 0 or integer powers of 2 such that $i+j < \|uv\|$.
\end{itemize}

(Remind that $\|uv\|$ stands for the number of edges in $uv$.) Note that it is a little different from previous DSOs~\cite{Demetrescu2008,2009A} since we want to query $\pi_{G-ab}(u,v)$ for any subpath $ab$ in $uv$, which is needed in computing new paths after inserting an edge $e$ into $G$.

As in~\cite{2009A}, we can also use $\max_{f\in (u \oplus i)(v \ominus j)}\og{u}{v}{f}$ to replace $\og{u}{v}{(u \oplus i)(v \ominus j)}$. When inserting new edges into the graph it is very difficult to maintain the path avoiding the whole range or the maximum shortest detour of the range.
However, in $\max_{f\in (u \oplus i)(v \ominus j)}\og{u}{v}{f}$, we only need the path $\og{u}{v}{f}$ for some $f$ which avoids the whole range between $(u \oplus i)$ and $(v \ominus j)$. When such an $f$ does not exist, that is, even the maximum one still intersects with the range between $(u \oplus i)$ and $(v \ominus j)$, we do not need to care about it. This makes the incremental DSO possible.
\section{Incremental DSO}\label{sec-inc}

In this section, we give the construction of the incremental distance sensitive oracle. That is, given an undirected weighted graph $G$,  we can construct an oracle of space $\tilde{O}(n^2)$ in $\too{n^3}$ preprocessing time, then for any pair of vertices $u,v$, and an edge $e$, we can query the shortest $u$-$v$ path in $G$ that avoids $e$ in time $\too{1}$. Further, the oracle supports incremental update: when a weighted edge is added to the graph, we can maintain the oracle in worst-case update time $\too{n^2}$. By Theorem~\ref{thm:offline}, it can be transformed into an offline fully dynamic DSO, such that if the total number of updates is $O(n)$, the total preprocessing and update time is $\tilde{O}(n^3)$ and the query time is still $\tilde{O}(1)$.


\subsection{Notations}

Below we say a subpath $ab \subseteq uv$ is \textbf{weak} under $uv$ if there exists an edge $f \in ab$, $\og{u}{v}{f} = \og{u}{v}{ab}$, and $f$ is called a \textbf{weak point} of $ab$ (under $uv$). (Recall that $G-ab$ means removing edges of $ab$ from $G$.) In other words, if $ab$ is weak under $uv$, the shortest $u-v$ path avoiding the weak point also avoids the whole subpath $ab.$ By definition, on $uv$ all weak points of $ab$ correspond to the same replacement path $\og{u}{v}{ab}$. 

\begin{definition}
    If a path $P$ between $u$ and $v$ in a graph $G$ is represented by the concatenation of a shortest path $ux$, an edge $(x,y)$ and a shortest path $yv$ in $G$, then we say $P$ is in the \textbf{proper form} in $G$. Here the edge $(x,y)$ and shortest path $yv$ can be empty.
\end{definition}

By ~\cite{2001Restoration} and Theorem~\ref{ReplacementPath} in this paper, when $ab$ is weak under $uv$, $\og{u}{v}{ab}$ can be represented in the proper form $ux\circ (x,y)\circ yv$.
We further define $\odg{u}{v}{ab}$ to be a null path, or of the proper form, such that,


\begin{itemize}
    \item  If $ab$ is weak, $\odg{u}{v}{ab}=\og{u}{v}{ab}$
    \item If $ab$ is not weak, $\odg{u}{v}{ab}$ is any path of the proper form $ux\circ (x,y)\circ yv$ avoiding $ab$, or is a null path with weight $+\infty.$ 
\end{itemize}

Note that in this structure we generally can hardly say whether a given $ab$ is weak or not under $uv$, even if we know the value of $\odg{u}{v}{ab}$. We also point out that even if $ab$ is weak under $uv$ we do not record the position of the weak point.

In our incremental structure, the graph before introducing edge $e$ is called $G$, and the graph $G\cup \{ e \}$ is called $G'.$  
For example, $\new{u, v}$ is the shortest $u-v$ path in $G'$, which is possibly different from $uv$, the shortest $u-v$ path in $G$, as it may go through $e$.

We sometimes need to check whether a path is in proper form, and if so, transform it to its proper form representation. From Lemma \ref{lemma2-2} and \ref{lemma2-3}, by constructing a least common ancestor (LCA) structure with $O(n)$ preprocessing time and $O(1)$ query time~\cite{BF00} on the shortest path tree from every vertex, in our incremental DSO structure we can transform every path to its proper form if possible in $\too{1}$ time.

\begin{definition}
Let $P$ be a path from $u$ to $v$ in $G$ and let $R$ be a subpath of $uv$ in $G$. We define a transform $T_{G,R}(P)$ on $P$ to be
\begin{itemize}
    \item the path $P$ of the proper form, if $P$ can be transformed to the proper form and it avoids edges of $R.$
    \item a null path with weight $+\infty$, otherwise.
\end{itemize}
\end{definition}

When this transform is used in this section, $P$ is given by a concatenation of several shortest paths and edges in $G$, and we want to transform it into the proper form in $G$ or the new graph $G'$. (Then $R$ is also a shortest path in $G'$.) By Lemma \ref{lemma2-2} and \ref{lemma2-3}, such transform $T_{G,R}$ can be done in $\tilde{O}(1)$ time.

\begin{lemma}\label{lemma2-2}
Given a path $P$ explicitly or implicitly, suppose the list of vertices on $P$ is $v_0,v_1,\cdots, v_r$, such that when given an index $i$, we can find $v_i$ and the length of the subpath of $P$ from $v_0$ to $v_i$ in $\tilde{O}(1)$ time.
Then we can transform $P$ into the proper form if it can, or otherwise point out that it cannot, in $\too{1}$ time.
\end{lemma}

\begin{proof}

Use a binary search to find the largest $j$ such that the subpath $v_0v_1\cdots v_j$ is a shortest path in the graph, then check whether $v_{j+1}v_{j+2}\cdots v_r$ or $v_jv_{j+1}\cdots v_r$ is a shortest path by the shortest path tree structure. So $P$ is in the proper form if one of them is a shortest path.
\end{proof}

\begin{lemma}\label{lemma2-3}
Let $P$ be a path in the proper form $ux\circ e\circ yv$ in $G$, and let $R$ be a subpath in the shortest path between $u$ and $v$ in $G$, then we can check whether $P$ and $R$ share any edge in $O(1)$ time.
\end{lemma}

\begin{proof}
Note that both $ux$ and $R$ are paths in the shortest path tree from $u.$ By calculating the least common ancestor~\cite{BF00} we can check whether they intersect in $O(1)$ time. Similarly for $yv$ and $R.$ Checking $e\in R$ or not is also in $O(1)$ time. 
\end{proof}

\subsection{Preprocessing}
In our structure for all pairs of vertices $u,v$, we maintain
\begin{itemize}
    \item Shortest path tree from every vertex, and the LCA structure~\cite{BF00} on them.
    \item $\odg{u}{v}{(u \oplus i)(v \ominus j)}$ for all $i, j\in\{0\}\cup\{2^k\}_{k\in\mathbb{N}}$ such that $i+j < ||uv||$.
\end{itemize}

(Recall that $\|uv\|$ stands for the number of edges in $uv$.) Note that all the detours in this structure are maintained in the proper form. We can see the space of this structure is $\tilde{O}(n^2)$.
Although our structure is a little different from the DSO in~\cite{2009A}, we can use their DSO to construct our structure:

\begin{lemma}
    We can construct this initial structure for any graph $G$ in $\too{n^3}$ time.
\end{lemma}

\begin{proof}
The shortest path tree structure and LCA structure can be constructed in $O(n^3)$ time. From \cite{2009A} and the reduction in \cref{sec:intro-dso} we can obtain a DSO for undirected graph $G$ in $\too{n^3}$ time, then for all pair of vertices $u,v$ and interval $(u \oplus i)(v \ominus j)$ where $i, j$ are 0 or integer powers of 2, run the following.

\begin{itemize}
    \item Find $f \in (u \oplus i)(v \ominus j)$ which maximizes $|\og{u}{v}{f}|$.
    \item After finding such $f$, retrieve the path $\og{u}{v}{f}$ and check whether it intersects with the edges in interval $(u \oplus i)(v \ominus j)$. 
    \item If it intersects with $(u \oplus i)(v \ominus j)$, set $\odg{u}{v}{(u \oplus i)(v \ominus j)}$ to be a null path.
    \item If it does not intersect with $(u \oplus i)(v \ominus j)$, use Lemma~\ref{lemma2-2} to transform the path $\og{u}{v}{f}$ into the proper form then store it as $\odg{u}{v}{(u \oplus i)(v \ominus j)}$.
\end{itemize}

We can see the time needed to compute every $\odg{u}{v}{(u \oplus i)(v \ominus j)}$ is $O(n)$, thus the total time is $\tilde{O}(n^3)$.


 
\end{proof}

Using our structure we can answer the shortest path between $u,v$ avoiding any subpath $ab$ of $uv$ if $ab$ is weak under $uv$.

\begin{theorem}\label{thm4-3}
    Let $ab$ be a subpath of $uv.$ By the information we maintain in the DSO, we can calculate the value of $\odg{u}{v}{ab}$ in $\too{1}$ time. W.l.o.g., suppose $a$ is closer to $u$ than $b$ is. Let $i$ be the largest integer power of 2 s.t. $u\oplus i$ falls in $ua.$ Similarly $j$ is the largest integer power of 2 s.t. $v\ominus j$ falls in $bv.$ Let $a'=a\ominus i, u'=u\oplus i, v'=v\ominus j, b'=b\oplus j.$ (Note that if $a=u$, $a'=u'=u$ and if $b=v$, $b'=v'=v$.) Then
    
\begin{equation}\label{eq-1}   
    \begin{aligned}
    \odg{u}{v}{ab}=\min\Big\{&T_{G,ab}\left(ua'\circ \odg{a'}{b'}{ab}\circ b'v\right), ~\odg{u}{v}{u'v'}, ~T_{G,ab}\left(ua'\circ \odg{a'}{v}{av'}\right), \\
    &~T_{G,ab}\left(\odg{u}{b'}{u'b}\circ b'v\right)\Big\}.\\
    \end{aligned}
\end{equation}

\end{theorem}      

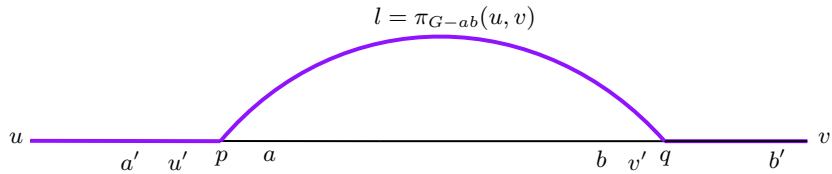
\begin{center}
    \tikzset{every picture/.style={line width=0.75pt}} 

\begin{tikzpicture}[x=0.75pt,y=0.75pt,yscale=-1,xscale=1]

\draw    (96,183.82) -- (272,184) ;
\draw [color={rgb, 255:red, 144; green, 19; blue, 254 }  ,draw opacity=0.5 ][line width=1.5]    (96,183.82) -- (192,184) ;
\draw [color={rgb, 255:red, 144; green, 19; blue, 254 }  ,draw opacity=0.5 ][line width=1.5]    (416,184) -- (488,184) ;
\draw    (272,184) -- (488,184) ;
\draw [color={rgb, 255:red, 144; green, 19; blue, 254 }  ,draw opacity=0.5 ][line width=1.5]    (192,184) .. controls (252.2,112.8) and (354.2,114.4) .. (416,184) ;

\draw (492,178.4) node [anchor=north west][inner sep=0.75pt]  [font=\footnotesize]  {$v$};
\draw (84,178.4) node [anchor=north west][inner sep=0.75pt]  [font=\footnotesize]  {$u$};
\draw (188,187.4) node [anchor=north west][inner sep=0.75pt]  [font=\footnotesize]  {$p$};
\draw (412,187.4) node [anchor=north west][inner sep=0.75pt]  [font=\footnotesize]  {$q$};
\draw (212,187.4) node [anchor=north west][inner sep=0.75pt]  [font=\footnotesize]  {$a$};
\draw (380,187.4) node [anchor=north west][inner sep=0.75pt]  [font=\footnotesize]  {$b$};
\draw (164,187.4) node [anchor=north west][inner sep=0.75pt]  [font=\footnotesize]  {$u'$};
\draw (140,187.4) node [anchor=north west][inner sep=0.75pt]  [font=\footnotesize]  {$a'$};
\draw (396,187.4) node [anchor=north west][inner sep=0.75pt]  [font=\footnotesize]  {$v'$};
\draw (467,186.4) node [anchor=north west][inner sep=0.75pt]  [font=\footnotesize]  {$b'$};
\draw (268,114.4) node [anchor=north west][inner sep=0.75pt]  [font=\footnotesize]  {$l=\pi _{G-ab}( u,v)$};

\end{tikzpicture}
\end{center}

\begin{proof}

First, we suppose that $ab$ is weak under $uv$, with a weak point $f$. Let $l=\og{u}{v}{ab}=\og{u}{v}{f}.$ Suppose $l$ diverges from $uv$ at $p$, and converges at $q.$ 
\begin{itemize}
    \item[-] If $p\ge a'$ and $q\le b'$, $l$ is a concatenation of $ua', \og{a'}{b'}{f}$, and $b'v$, i.e. $l=ua'\circ \og{a'}{b'}{f}\circ b'v.$ Since $l$ avoids $ab$, we know $\og{a'}{b'}{f}$ avoids $ab$, so $ab$ is also weak under $a'b'$, so $\og{a'}{b'}{ab}=\odg{a'}{b'}{ab}.$ Note that the numbers of edges in $aa'$ and $bb'$ are integer powers of 2 (or zero), so $\odg{a'}{b'}{ab}$ is a value we maintained in the DSO structure.
    \item[-] If $p<a'$ and $q>b'$, we have $p<u'$ and $q>v'.$ In this case $l=\og{u}{v}{u'v'}.$ Also note that $\og{u}{v}{u'v'}=\og{u}{v}{ab}=\og{u}{v}{f}$, which means that $u'v'$ is also weak under $uv$ with weak point $f$, so $\og{u}{v}{u'v'}=\odg{u}{v}{u'v'}.$ The numbers of edges in $uu$ and $vv'$ are integer powers of 2 (or zero), so $\odg{u}{v}{u'v'}$ is a value we maintained in the DSO structure.
    \item[-] If $p\ge a'$ and $q>b'$, similarly we can show that $av'$ is weak under $a'v$ and $l=ua'\circ \odg{a'}{v}{av'}.$
    \item[-] If $p<a'$ and $q\le b'$, $u'b$ is weak under $ub'$ and $l=\odg{u}{b'}{u'b}\circ b'v.$
\end{itemize}
 
So when $ab$ is weak under $uv$, $\og{u}{v}{ab}$ must be in one of these 4 cases which can be transformed to the proper form by $T_{G,ab}$, so $LHS \geq RHS$ in Equation~(\ref{eq-1}). Also, any path that can pass the transform $T_{G,ab}$ is a path of proper form avoiding $ab$, so $LHS \leq RHS$ in Equation~(\ref{eq-1}).


Of course, if $ab$ is not weak under $uv$, by our definition any proper form $u-v$ path avoiding $ab$ or the null path is legal for $\odg{u}{v}{ab}$, so the algorithm also works. 
\end{proof}

We remark that when $(a,b)$ is an edge, the path $ab$ is always weak by definition, so we can get $\og{u}{v}{ab}$ by Theorem~\ref{thm4-3}. Therefore, we can answer any DSO query $\og{u}{v}{f}$ in $\too{1}$ time.

\subsection{Incremental Update}\label{sec4-3}

Consider the current DSO structure for $G$ and an edge $e=(x,y)$ to be added. Recall that we denote $G'=G\cup \{ e \}$, so we want to construct the DSO structure for $G'$ using current structure.
\begin{itemize}
    \item For every pair of vertices $u,v$, we first maintain the new information of their shortest path as:
          \[\new{u, v}=\min\{uv,ux\circ e\circ yv,~uy\circ e\circ xv\}.\]
    
          Since if the new shortest path does not go through $e$, it will be the original $uv$ in $G$. Otherwise, there are two cases that it goes through $e$ from $x$ to $y$ or from $y$ to $x$, which correspond to the other two cases. We also construct the shortest path tree from every vertex $u$ and the LCA structure on it.~\cite{BF00}

    
        
        
    \item For the detours on intervals $\odg{u}{v}{(u \oplus i)(v \ominus j)}$, where $i,j$ are zero or integer powers of 2, we consider two cases separately as below: either $\new{u, v} \neq uv$ (see \ref{sec3-3-1}), or $\new{u, v} = uv$ (see \ref{sec3-3-2}).
\end{itemize}

In the remaining part of this section, when showing the correctness of the incremental algorithm, by default we assume that $R=\new{u \oplus i,v \ominus j}$ is weak under $\new{u,v}$. This is because we only care about the value $\nng{u}{v}{R}$ when $R$ is weak. If $R$ is not weak, because of the transform of $T_{G',R}$, we can always make sure that the path given by the algorithm is a proper form for the $\varpi$ function (or a null path).

\subsubsection{$\pi_G(u,v)\neq\pi_{G'}(u,v)$}\label{sec3-3-1}

We know in this case $e=(x,y)$ is on $\new{u, v}$. W.l.o.g., suppose that $x$ is nearer to $u$ than $y$. We can see that $uv$ is a shortest $u-v$ path in $G=G'-e$, i.e. $uv=\nng{u}{v}{e}$. Let $p=\Delta(uv,\new{u, v})$ be the divergence point of $uv$ and $\new{u, v}$, and $q=\nabla(uv,\new{u, v})$ be the convergence point. Let $a=u \oplus i, b=v \ominus j$, then $R= \new{a,b}$ and we want to find $\ndg{u}{v}{R}$. So the points $a,b,x,y$ and $p,q$ are all on the new shortest path $\new{u, v}.$  

There are six possible relative positions of $R,p,q,e$:
\begin{itemize}
    \item \textbf{CASE 1:} $p,q\not\in R$, $e\in R$
    \item \textbf{CASE 2:} $p,q\not\in R$, $e\not\in R,$ $pq\cap R=\emptyset$
    \item \textbf{CASE 3:} $p,q\not\in R$, $e\not\in R,$ $pq\cap R\ne\emptyset$
    \item \textbf{CASE 4:} One of $p,q$ is in $R$, $e\in R$
    \item \textbf{CASE 5:} One of $p,q$ is in $R$, $e\not\in R$
    \item \textbf{CASE 6:} $p,q\in R$, $e\in R$
\end{itemize}

Note that it is impossible to have both $p,q\in R$ but $e\not\in R.$ This is because $e\in\new{u,v}$ but $e\not\in \pi_G(u,v)$, meaning that $e\in \new{p,q}.$ Therefore, the relative positions of $p,q,x,y,a,b$ fall into one of the following six cases or their symmetric cases. (Recall that we only consider the case that $R$ is weak under $\new{u,v}$.)

~\\
\noindent\textbf{CASE 1:} $p\le a\le x<y\le b\le q$,

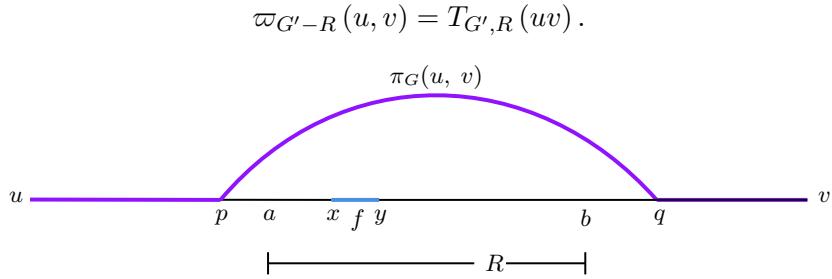
$$\ndg{u}{v}{R}=\pp{uv}.$$

\begin{center}
    \tikzset{every picture/.style={line width=0.75pt}} 

\begin{tikzpicture}[x=0.75pt,y=0.75pt,yscale=-1,xscale=1]

\draw    (96,183.82) -- (248,184) ;
\draw [color={rgb, 255:red, 144; green, 19; blue, 254 }  ,draw opacity=0.5 ][line width=1.5]    (96,183.82) -- (192,184) ;
\draw [color={rgb, 255:red, 74; green, 144; blue, 226 }  ,draw opacity=0.5 ][line width=1.5]    (248,184) -- (272,184) ;
\draw    (320,216) -- (216,216) ;
\draw [shift={(216,216)}, rotate = 360] [color={rgb, 255:red, 0; green, 0; blue, 0 }  ][line width=0.75]    (0,5.59) -- (0,-5.59)   ;
\draw    (336,216) -- (376,216) ;
\draw [shift={(376,216)}, rotate = 180] [color={rgb, 255:red, 0; green, 0; blue, 0 }  ][line width=0.75]    (0,5.59) -- (0,-5.59)   ;
\draw [color={rgb, 255:red, 144; green, 19; blue, 254 }  ,draw opacity=0.5 ][line width=1.5]    (192,184) .. controls (252.2,112.8) and (350.2,114.4) .. (412,184) ;
\draw [color={rgb, 255:red, 144; green, 19; blue, 254 }  ,draw opacity=0.5 ][line width=1.5]    (412,184) -- (488,184) ;
\draw    (272,184) -- (488,184) ;

\draw (84,178.4) node [anchor=north west][inner sep=0.75pt]  [font=\footnotesize]  {$u$};
\draw (276,114.4) node [anchor=north west][inner sep=0.75pt]  [font=\footnotesize]  {$\pi _{G}( u,\ v)$};
\draw (244,187.4) node [anchor=north west][inner sep=0.75pt]  [font=\footnotesize]  {$x$};
\draw (188,187.4) node [anchor=north west][inner sep=0.75pt]  [font=\footnotesize]  {$p$};
\draw (268,187.4) node [anchor=north west][inner sep=0.75pt]  [font=\footnotesize]  {$y$};
\draw (212,187.4) node [anchor=north west][inner sep=0.75pt]  [font=\footnotesize]  {$a$};
\draw (324,210.4) node [anchor=north west][inner sep=0.75pt]  [font=\footnotesize]  {$R$};
\draw (492,178.4) node [anchor=north west][inner sep=0.75pt]  [font=\footnotesize]  {$v$};
\draw (409,187.4) node [anchor=north west][inner sep=0.75pt]  [font=\footnotesize]  {$q$};
\draw (372,187.4) node [anchor=north west][inner sep=0.75pt]  [font=\footnotesize]  {$b$};
\draw (256,187.4) node [anchor=north west][inner sep=0.75pt]  [font=\footnotesize]  {$f$};

\end{tikzpicture}
\end{center}

In this case, since $uv=\nng{u}{v}{e}$, $e\in R$, $R= \new{a, b}$, and $uv$ avoids $R$, so $R$ is weak with a weak point $f=e$, and $\nng{u}{v}{e}=uv$, which must pass the transform $T_{G',R}$. 

~\\
\noindent\textbf{CASE 2:} $a\le b\le p\le x<y\le q$,
$$\ndg{u}{v}{R}=
    \min\{\pp{\odg{u}{v}{R}}, \pp{\odg{u}{x}{R}\circ e\circ yv}\}.
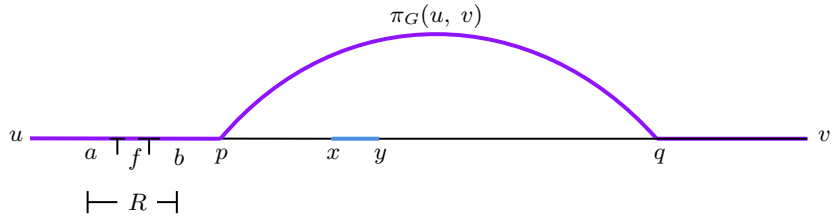
$$

\begin{center}
    \tikzset{every picture/.style={line width=0.75pt}} 

\begin{tikzpicture}[x=0.75pt,y=0.75pt,yscale=-1,xscale=1]

\draw    (96,183.82) -- (248,184) ;
\draw [color={rgb, 255:red, 144; green, 19; blue, 254 }  ,draw opacity=0.5 ][line width=1.5]    (96,183.82) -- (192,184) ;
\draw [color={rgb, 255:red, 74; green, 144; blue, 226 }  ,draw opacity=0.5 ][line width=1.5]    (248,184) -- (272,184) ;
\draw    (140,216) -- (125,216) ;
\draw [shift={(125,216)}, rotate = 360] [color={rgb, 255:red, 0; green, 0; blue, 0 }  ][line width=0.75]    (0,5.59) -- (0,-5.59)   ;
\draw    (160,216) -- (170,216) ;
\draw [shift={(170,216)}, rotate = 180] [color={rgb, 255:red, 0; green, 0; blue, 0 }  ][line width=0.75]    (0,5.59) -- (0,-5.59)   ;
\draw [color={rgb, 255:red, 144; green, 19; blue, 254 }  ,draw opacity=0.5 ][line width=1.5]    (192,184) .. controls (252.2,112.8) and (350.2,114.4) .. (412,184) ;
\draw [color={rgb, 255:red, 144; green, 19; blue, 254 }  ,draw opacity=0.5 ][line width=1.5]    (412,184) -- (488,184) ;
\draw    (272,184) -- (488,184) ;
\draw    (140,192) -- (140,184) ;
\draw [shift={(140,184)}, rotate = 90] [color={rgb, 255:red, 0; green, 0; blue, 0 }  ][line width=0.75]    (0,3.91) -- (0,-3.91)   ;
\draw    (156,192) -- (156,184) ;
\draw [shift={(156,184)}, rotate = 90] [color={rgb, 255:red, 0; green, 0; blue, 0 }  ][line width=0.75]    (0,5.59) -- (0,-5.59)   ;

\draw (84,178.4) node [anchor=north west][inner sep=0.75pt]  [font=\footnotesize]  {$u$};
\draw (276,114.4) node [anchor=north west][inner sep=0.75pt]  [font=\footnotesize]  {$\pi _{G}( u,\ v)$};
\draw (244,187.4) node [anchor=north west][inner sep=0.75pt]  [font=\footnotesize]  {$x$};
\draw (188,187.4) node [anchor=north west][inner sep=0.75pt]  [font=\footnotesize]  {$p$};
\draw (268,187.4) node [anchor=north west][inner sep=0.75pt]  [font=\footnotesize]  {$y$};
\draw (122,187.4) node [anchor=north west][inner sep=0.75pt]  [font=\footnotesize]  {$a$};
\draw (144,210.4) node [anchor=north west][inner sep=0.75pt]  [font=\footnotesize]  {$R$};
\draw (492,178.4) node [anchor=north west][inner sep=0.75pt]  [font=\footnotesize]  {$v$};
\draw (409,187.4) node [anchor=north west][inner sep=0.75pt]  [font=\footnotesize]  {$q$};
\draw (167,187.4) node [anchor=north west][inner sep=0.75pt]  [font=\footnotesize]  {$b$};
\draw (144,187.4) node [anchor=north west][inner sep=0.75pt]  [font=\footnotesize]  {$f$};

\end{tikzpicture}
\end{center}

\begin{itemize}
    \item[-] If $\nng{u}{v}{R}$ does not go through $e$,
    
\vspace{5pt}
 similar as \textbf{CASE 1}, it is in $G'-e$, which equals $G.$ Since $R$ is weak under $\new{u,v}$, it is also weak under $uv$, 
 so $\nng{u}{v}{R}=\odg{u}{v}{R}$.

\vspace{5pt}
    \item[-] If $\nng{u}{v}{R}$ goes through $e$,
    
\vspace{5pt}
by Theorem~\ref{ReplacementPath} and~\ref{thm4-2}, it takes a shortest $u-x$ path that avoids $R$, then goes along $\new{x,v}$, i.e. it equals $\og{u}{x}{R}\circ e\circ yv.$ Suppose $f\in ab$ is a weak point of $R$ under $\new{u,v}$, $\og{u}{x}{f}\circ e\circ yv=\nng{u}{v}{f}=\nng{u}{v}{R}=\og{u}{x}{R}\circ e\circ yv$, so $R$ is also weak under $ux$ with a weak point $f$, so $\odg{u}{x}{R}=\og{u}{x}{R}.$
\vspace{5pt}
    
\end{itemize}

This means if $\odg{u}{x}{R}\circ e\circ yv$ is not a proper form, $\nng{u}{v}{R}$ does not go through $e$, so $\ndg{u}{v}{R}=uv.$ Otherwise we should choose the minimum of both. Note that if $\odg{u}{x}{R}$ intersects with $yv$, it cannot be the shortest one and also cannot pass the transform $T_{G',R}$.

~\\
\noindent\textbf{CASE 3:} $p\le x<y\le a\le b\le q$,

$$\ndg{u}{v}{R}=
    \min\{\pp{uv}, \pp{ux\circ e\circ\odg{y}{v}{R}}\}.
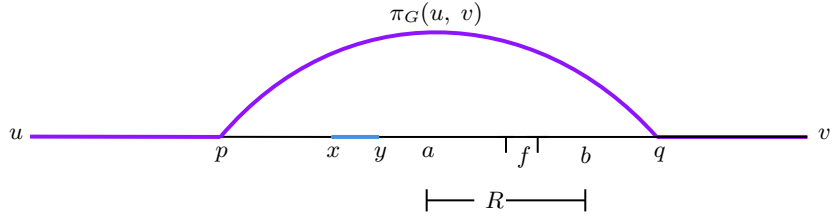
$$

\begin{center}
    \tikzset{every picture/.style={line width=0.75pt}} 

\begin{tikzpicture}[x=0.75pt,y=0.75pt,yscale=-1,xscale=1]

\draw    (96,183.82) -- (248,184) ;
\draw [color={rgb, 255:red, 144; green, 19; blue, 254 }  ,draw opacity=0.5 ][line width=1.5]    (96,183.82) -- (192,184) ;
\draw [color={rgb, 255:red, 74; green, 144; blue, 226 }  ,draw opacity=0.5 ][line width=1.5]    (248,184) -- (272,184) ;
\draw    (320,216) -- (296,216) ;
\draw [shift={(296,216)}, rotate = 360] [color={rgb, 255:red, 0; green, 0; blue, 0 }  ][line width=0.75]    (0,5.59) -- (0,-5.59)   ;
\draw    (336,216) -- (376,216) ;
\draw [shift={(376,216)}, rotate = 180] [color={rgb, 255:red, 0; green, 0; blue, 0 }  ][line width=0.75]    (0,5.59) -- (0,-5.59)   ;
\draw [color={rgb, 255:red, 144; green, 19; blue, 254 }  ,draw opacity=0.5 ][line width=1.5]    (192,184) .. controls (252.2,112.8) and (350.2,114.4) .. (412,184) ;
\draw [color={rgb, 255:red, 144; green, 19; blue, 254 }  ,draw opacity=0.5 ][line width=1.5]    (412,184) -- (488,184) ;
\draw    (272,184) -- (488,184) ;
\draw    (336,192) -- (336,184) ;
\draw [shift={(336,184)}, rotate = 90] [color={rgb, 255:red, 0; green, 0; blue, 0 }  ][line width=0.75]    (0,3.91) -- (0,-3.91)   ;
\draw    (352,192) -- (352,184) ;
\draw [shift={(352,184)}, rotate = 90] [color={rgb, 255:red, 0; green, 0; blue, 0 }  ][line width=0.75]    (0,5.59) -- (0,-5.59)   ;

\draw (84,178.4) node [anchor=north west][inner sep=0.75pt]  [font=\footnotesize]  {$u$};
\draw (276,114.4) node [anchor=north west][inner sep=0.75pt]  [font=\footnotesize]  {$\pi _{G}( u,\ v)$};
\draw (244,187.4) node [anchor=north west][inner sep=0.75pt]  [font=\footnotesize]  {$x$};
\draw (188,187.4) node [anchor=north west][inner sep=0.75pt]  [font=\footnotesize]  {$p$};
\draw (268,187.4) node [anchor=north west][inner sep=0.75pt]  [font=\footnotesize]  {$y$};
\draw (292,187.4) node [anchor=north west][inner sep=0.75pt]  [font=\footnotesize]  {$a$};
\draw (324,210.4) node [anchor=north west][inner sep=0.75pt]  [font=\footnotesize]  {$R$};
\draw (492,178.4) node [anchor=north west][inner sep=0.75pt]  [font=\footnotesize]  {$v$};
\draw (409,187.4) node [anchor=north west][inner sep=0.75pt]  [font=\footnotesize]  {$q$};
\draw (372,187.4) node [anchor=north west][inner sep=0.75pt]  [font=\footnotesize]  {$b$};
\draw (340,187.4) node [anchor=north west][inner sep=0.75pt]  [font=\footnotesize]  {$f$};

\end{tikzpicture}
\end{center}

The discussion is very similar to \textbf{CASE 2}. If $\nng{u}{v}{R}$ does not go through $e$, then $\nng{u}{v}{R}$ is in $G'-e$, which equals $G,$ so we have $\nng{u}{v}{R}=uv$. If $\nng{u}{v}{R}$ goes through $e$, it is the same as \textbf{CASE 2}, so $\nng{u}{v}{R}=ux\circ e\circ\odg{y}{v}{R}$.

~\\
Now we introduce Lemma \ref{thm4-1} to be used in the following cases.

\begin{lemma}\label{thm4-1}
    Under the settings of Section \ref{sec3-3-1}. suppose that $e\in R=\new{a,b}$, and $R$ is weak under $\new{u,v}$ with weak point $f.$ Assuming that $uv\cap R\neq\emptyset$, then:\\
    (1) $f\in uv$,\\
    (2) $\og{u}{v}{f}=\nng{u}{v}{f}=\nng{u}{v}{R}.$
\end{lemma}

\begin{proof}
If $f$ is not in $uv$, meaning that $uv$ is a $u-v$ path avoiding $f.$ Since $R$ is weak and $f$ is a weak point, we know $\nng{u}{v}{f}$ avoids the whole interval $R$, so it avoids $e$. As a result, $\nng{u}{v}{f}$ is a path in $G$, i.e. $\nng{u}{v}{f}=\og{u}{v}{f}.$ However, $f$ is not in $uv$, so $\nng{u}{v}{f}=uv$. Since $uv$ does not avoid the whole interval $R$, this is contradictory with the fact that $f$ is a weak point of $R$, so (1) is proved.

Let $f\in uv$ be a weak point of $R.$ We know $\nng{u}{v}{f}$ avoids $R$, and $e\in R$, so it avoids $e.$ So $\nng{u}{v}{f}$ is a path in $G$, i.e. $\nng{u}{v}{f}=\og{u}{v}{f}$, so (2) is proved.

\end{proof}

\noindent\textbf{CASE 4:} $p\le a\le x<y\le q\le b$,

$$\ndg{u}{v}{R}=\pp{\odg{u}{v}{qb}} $$


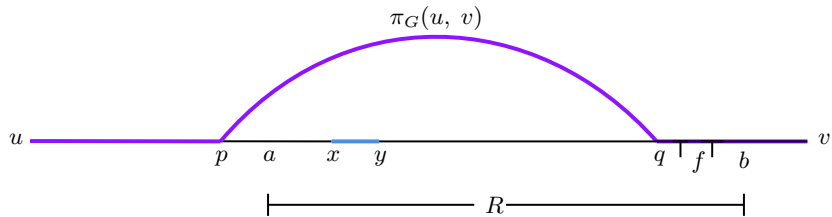
\begin{center}
    \tikzset{every picture/.style={line width=0.75pt}} 

\begin{tikzpicture}[x=0.75pt,y=0.75pt,yscale=-1,xscale=1]

\draw    (96,183.82) -- (248,184) ;
\draw [color={rgb, 255:red, 144; green, 19; blue, 254 }  ,draw opacity=0.5 ][line width=1.5]    (96,183.82) -- (192,184) ;
\draw [color={rgb, 255:red, 74; green, 144; blue, 226 }  ,draw opacity=0.5 ][line width=1.5]    (248,184) -- (272,184) ;
\draw    (320,216) -- (216,216) ;
\draw [shift={(216,216)}, rotate = 360] [color={rgb, 255:red, 0; green, 0; blue, 0 }  ][line width=0.75]    (0,5.59) -- (0,-5.59)   ;
\draw    (336,216) -- (456,216) ;
\draw [shift={(456,216)}, rotate = 180] [color={rgb, 255:red, 0; green, 0; blue, 0 }  ][line width=0.75]    (0,5.59) -- (0,-5.59)   ;
\draw [color={rgb, 255:red, 144; green, 19; blue, 254 }  ,draw opacity=0.5 ][line width=1.5]    (192,184) .. controls (252.2,112.8) and (350.2,114.4) .. (412,184) ;
\draw [color={rgb, 255:red, 144; green, 19; blue, 254 }  ,draw opacity=0.5 ][line width=1.5]    (412,184) -- (488,184) ;
\draw    (272,184) -- (488,184) ;
\draw    (424,192) -- (424,184) ;
\draw [shift={(424,184)}, rotate = 90] [color={rgb, 255:red, 0; green, 0; blue, 0 }  ][line width=0.75]    (0,3.91) -- (0,-3.91)   ;
\draw    (440,192) -- (440,184) ;
\draw [shift={(440,184)}, rotate = 90] [color={rgb, 255:red, 0; green, 0; blue, 0 }  ][line width=0.75]    (0,5.59) -- (0,-5.59)   ;

\draw (84,178.4) node [anchor=north west][inner sep=0.75pt]  [font=\footnotesize]  {$u$};
\draw (276,114.4) node [anchor=north west][inner sep=0.75pt]  [font=\footnotesize]  {$\pi _{G}( u,\ v)$};
\draw (244,187.4) node [anchor=north west][inner sep=0.75pt]  [font=\footnotesize]  {$x$};
\draw (188,187.4) node [anchor=north west][inner sep=0.75pt]  [font=\footnotesize]  {$p$};
\draw (268,187.4) node [anchor=north west][inner sep=0.75pt]  [font=\footnotesize]  {$y$};
\draw (212,187.4) node [anchor=north west][inner sep=0.75pt]  [font=\footnotesize]  {$a$};
\draw (324,210.4) node [anchor=north west][inner sep=0.75pt]  [font=\footnotesize]  {$R$};
\draw (492,178.4) node [anchor=north west][inner sep=0.75pt]  [font=\footnotesize]  {$v$};
\draw (409,187.4) node [anchor=north west][inner sep=0.75pt]  [font=\footnotesize]  {$q$};
\draw (452,187.4) node [anchor=north west][inner sep=0.75pt]  [font=\footnotesize]  {$b$};
\draw (428,187.4) node [anchor=north west][inner sep=0.75pt]  [font=\footnotesize]  {$f$};

\end{tikzpicture}
\end{center}

By Lemma~\ref{thm4-1}, if $R$ is weak under $\new{u, v}$, the weak point $f$ is in $qb$, and $\og{u}{v}{f}=\nng{u}{v}{f}=\nng{u}{v}{R}$. By the unique shortest path assumption, this means $\og{u}{v}{f}$ also avoids $R$, so it also avoids $qb$  thus $qb$ is weak under $uv$ and $f$ is a weak point of it. Now $\nng{u}{v}{R}=\og{u}{v}{f}=\odg{u}{v}{qb}$.

Also, if $\odg{u}{v}{qb}$ does not avoid $R$, $R$ is not weak under $\new{u, v}$, so any path passing the transform $T_{G',R}$ or the null path is legal for $\ndg{u}{v}{R}$. (Note that it is possible that $\odg{u}{v}{qb}$ intersects with $ax$ in $R$.)

~\\

\noindent\textbf{CASE 5:} $p\le x<y\le a\le q\le b$,

$$\ndg{u}{v}{R}=
    \min\{\pp{\odg{u}{v}{qb}}, \pp{ux\circ e\circ \odg{y}{v}{ab}}\}.
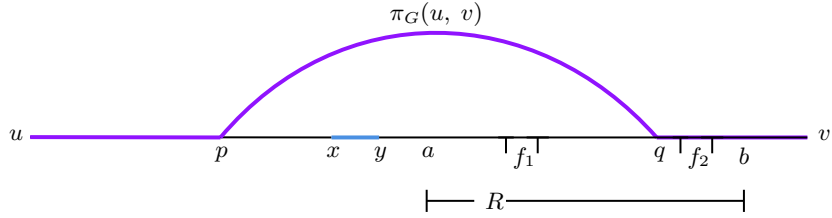
$$

\begin{center}
    \tikzset{every picture/.style={line width=0.75pt}} 

\begin{tikzpicture}[x=0.75pt,y=0.75pt,yscale=-1,xscale=1]

\draw    (96,183.82) -- (248,184) ;
\draw [color={rgb, 255:red, 144; green, 19; blue, 254 }  ,draw opacity=0.5 ][line width=1.5]    (96,183.82) -- (192,184) ;
\draw [color={rgb, 255:red, 74; green, 144; blue, 226 }  ,draw opacity=0.5 ][line width=1.5]    (248,184) -- (272,184) ;
\draw    (320,216) -- (296,216) ;
\draw [shift={(296,216)}, rotate = 360] [color={rgb, 255:red, 0; green, 0; blue, 0 }  ][line width=0.75]    (0,5.59) -- (0,-5.59)   ;
\draw    (336,216) -- (456,216) ;
\draw [shift={(456,216)}, rotate = 180] [color={rgb, 255:red, 0; green, 0; blue, 0 }  ][line width=0.75]    (0,5.59) -- (0,-5.59)   ;
\draw [color={rgb, 255:red, 144; green, 19; blue, 254 }  ,draw opacity=0.5 ][line width=1.5]    (192,184) .. controls (252.2,112.8) and (350.2,114.4) .. (412,184) ;
\draw [color={rgb, 255:red, 144; green, 19; blue, 254 }  ,draw opacity=0.5 ][line width=1.5]    (412,184) -- (488,184) ;
\draw    (272,184) -- (488,184) ;
\draw    (336,192) -- (336,184) ;
\draw [shift={(336,184)}, rotate = 90] [color={rgb, 255:red, 0; green, 0; blue, 0 }  ][line width=0.75]    (0,3.91) -- (0,-3.91)   ;
\draw    (352,192) -- (352,184) ;
\draw [shift={(352,184)}, rotate = 90] [color={rgb, 255:red, 0; green, 0; blue, 0 }  ][line width=0.75]    (0,5.59) -- (0,-5.59)   ;
\draw    (424,192) -- (424,184) ;
\draw [shift={(424,184)}, rotate = 90] [color={rgb, 255:red, 0; green, 0; blue, 0 }  ][line width=0.75]    (0,3.91) -- (0,-3.91)   ;
\draw    (440,192) -- (440,184) ;
\draw [shift={(440,184)}, rotate = 90] [color={rgb, 255:red, 0; green, 0; blue, 0 }  ][line width=0.75]    (0,5.59) -- (0,-5.59)   ;

\draw (84,178.4) node [anchor=north west][inner sep=0.75pt]  [font=\footnotesize]  {$u$};
\draw (276,114.4) node [anchor=north west][inner sep=0.75pt]  [font=\footnotesize]  {$\pi _{G}( u,\ v)$};
\draw (244,187.4) node [anchor=north west][inner sep=0.75pt]  [font=\footnotesize]  {$x$};
\draw (188,187.4) node [anchor=north west][inner sep=0.75pt]  [font=\footnotesize]  {$p$};
\draw (268,187.4) node [anchor=north west][inner sep=0.75pt]  [font=\footnotesize]  {$y$};
\draw (292,187.4) node [anchor=north west][inner sep=0.75pt]  [font=\footnotesize]  {$a$};
\draw (324,210.4) node [anchor=north west][inner sep=0.75pt]  [font=\footnotesize]  {$R$};
\draw (492,178.4) node [anchor=north west][inner sep=0.75pt]  [font=\footnotesize]  {$v$};
\draw (409,187.4) node [anchor=north west][inner sep=0.75pt]  [font=\footnotesize]  {$q$};
\draw (452,187.4) node [anchor=north west][inner sep=0.75pt]  [font=\footnotesize]  {$b$};
\draw (338,187.4) node [anchor=north west][inner sep=0.75pt]  [font=\footnotesize]  {$f_{1}$};
\draw (426,187.4) node [anchor=north west][inner sep=0.75pt]  [font=\footnotesize]  {$f_{2}$};

\end{tikzpicture}
\end{center}

Assume $R$ is weak and has a weak point $f$, so $\nng{u}{v}{f}=\nng{u}{v}{R}$.

\begin{itemize}
    \item If $\nng{u}{v}{f}$ goes through $e$, similar as the discussion in \textbf{CASE 2} we know it equals $ux\circ e\circ\og{y}{v}{f}$. Since $f$ is a weak point of $R$, this path avoids $R$, so $\og{y}{v}{f}=\og{y}{v}{ab}$, thus $ab$ is weak under $yv$ with weak point $f$. So $\nng{u}{v}{f}=ux\circ e\circ\odg{y}{v}{ab}$.
    \item If $\nng{u}{v}{f}$ does not go through $e$, so $\nng{u}{v}{f}=\og{u}{v}{f}$ and they avoids $R$.
    \begin{itemize}
        \item If $f\in aq$, then $f$ is not in $uv$, $\og{u}{v}{f}=uv$ which intersects with $R$, making a contradiction.
        \item Thus $f\in qb$. Since $\og{u}{v}{f}$ avoids $R$, it avoids $qb$, so $\og{u}{v}{qb}=\og{u}{v}{f}$, and $qb$ is weak under $uv$ with weak point $f$. So in this case $\ndg{u}{v}{R}=\odg{u}{v}{qb}$.
    \end{itemize}
\end{itemize}

Thus, if $R$ is weak, the minimum one of $\odg{u}{v}{qb}$ and $ux\circ e\circ \odg{y}{v}{ab}$ which pass the transform $T_{G',R}$ must be $\ndg{u}{v}{R}$. If $R$ is not weak, any path passing the transform or the null path is legal.

~\\
\noindent\textbf{CASE 6:} $a\le p\le x<y\le q\le b$,

$$\ndg{u}{v}{R}=\pp{\odg{u}{v}{ab}}.$$

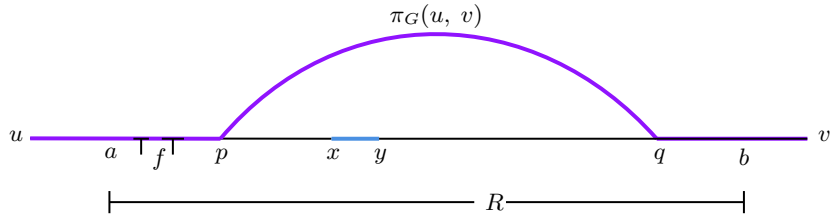
\begin{center}
    \tikzset{every picture/.style={line width=0.75pt}} 

\begin{tikzpicture}[x=0.75pt,y=0.75pt,yscale=-1,xscale=1]

\draw    (96,183.82) -- (248,184) ;
\draw [color={rgb, 255:red, 144; green, 19; blue, 254 }  ,draw opacity=0.5 ][line width=1.5]    (192,184) .. controls (252.2,112.8) and (350.2,114.4) .. (412,184) ;
\draw [color={rgb, 255:red, 144; green, 19; blue, 254 }  ,draw opacity=0.5 ][line width=1.5]    (96,183.82) -- (192,184) ;
\draw [color={rgb, 255:red, 144; green, 19; blue, 254 }  ,draw opacity=0.5 ][line width=1.5]    (412,184) -- (488,184) ;
\draw [color={rgb, 255:red, 74; green, 144; blue, 226 }  ,draw opacity=0.5 ][line width=1.5]    (248,184) -- (272,184) ;
\draw    (272,184) -- (488,184) ;
\draw    (320,216) -- (136,216) ;
\draw [shift={(136,216)}, rotate = 360] [color={rgb, 255:red, 0; green, 0; blue, 0 }  ][line width=0.75]    (0,5.59) -- (0,-5.59)   ;
\draw    (336,216) -- (456,216) ;
\draw [shift={(456,216)}, rotate = 180] [color={rgb, 255:red, 0; green, 0; blue, 0 }  ][line width=0.75]    (0,5.59) -- (0,-5.59)   ;
\draw    (152,192) -- (152,184) ;
\draw [shift={(152,184)}, rotate = 90] [color={rgb, 255:red, 0; green, 0; blue, 0 }  ][line width=0.75]    (0,3.91) -- (0,-3.91)   ;
\draw    (168,192) -- (168,184) ;
\draw [shift={(168,184)}, rotate = 90] [color={rgb, 255:red, 0; green, 0; blue, 0 }  ][line width=0.75]    (0,5.59) -- (0,-5.59)   ;

\draw (492,178.4) node [anchor=north west][inner sep=0.75pt]  [font=\footnotesize]  {$v$};
\draw (84,178.4) node [anchor=north west][inner sep=0.75pt]  [font=\footnotesize]  {$u$};
\draw (276,114.4) node [anchor=north west][inner sep=0.75pt]  [font=\footnotesize]  {$\pi _{G}( u,\ v)$};
\draw (244,187.4) node [anchor=north west][inner sep=0.75pt]  [font=\footnotesize]  {$x$};
\draw (188,187.4) node [anchor=north west][inner sep=0.75pt]  [font=\footnotesize]  {$p$};
\draw (409,187.4) node [anchor=north west][inner sep=0.75pt]  [font=\footnotesize]  {$q$};
\draw (268,187.4) node [anchor=north west][inner sep=0.75pt]  [font=\footnotesize]  {$y$};
\draw (132,187.31) node [anchor=north west][inner sep=0.75pt]  [font=\footnotesize]  {$a$};
\draw (452,187.4) node [anchor=north west][inner sep=0.75pt]  [font=\footnotesize]  {$b$};
\draw (324,210.4) node [anchor=north west][inner sep=0.75pt]  [font=\footnotesize]  {$R$};
\draw (156,187.4) node [anchor=north west][inner sep=0.75pt]  [font=\footnotesize]  {$f$};

\end{tikzpicture}
\end{center}

As we have assumed, $R$ is weak. Suppose $f\in R$ is a weak point of $R.$ By Lemma \ref{thm4-1}, $f\in ap\cup qb$, and $\og{u}{v}{f}=\nng{u}{v}{f}=\nng{u}{v}{R}$. By the unique shortest path assumption, $\og{u}{v}{f}$ also avoids $R$, so it avoids $ap$ and $qb$. If $\og{u}{v}{f}$ intersects $ab=\pi_G(a,b)$, it can only intersects $pq=\pi_G(p,q).$ By Theorem \ref{thm4-2}, it is impossible that $\og{u}{v}{f}$ intersects $pq$ but avoids $ap$ and $qb$, so $\og{u}{v}{f}$ avoids $ab$, and $ab$ is weak under $uv$ with weak point $f$, thus $\ndg{u}{v}{R}=\og{u}{v}{f}=\odg{u}{v}{ab}$.

Also, if $\odg{u}{v}{ab}$ does not pass the transform $T_{G',R}$, then $R$ is not weak under $\new{u, v}$, so any path passing the transform $T_{G',R}$ or the null path is legal for $\ndg{u}{v}{R}$.

\subsubsection{$\pi_G(u,v)=\pi_{G'}(u,v)$}\label{sec3-3-2}

As before, let $a = u \oplus i$, $b = v \ominus j$,  $G'=G\cup\{e\}$, $R=\new{a, b}$, and we want to find $\ndg{u}{v}{R}$. Also suppose that $R$ is weak under $\pi_{G'}(u,v)$ in graph $G'$, and $f$ is a weak point of it. The algorithm is as follows.



\begin{itemize}
    \item Set $\ndg{u}{v}{R}\leftarrow \pp{\odg{u}{v}{R}}$
    \item Execute the following algorithm, then swap $x,y$ and repeat again
    \item Let $p=\Delta(ux,uv)$ be the last point that $\pi_G(u,x)$ and $\pi_G(u,v)$ shares, and similarly $q=\nabla(yv,uv)$ be the first point that $\pi_G(y,v)$ and $\pi_G(u,v)$ shares, consider the following cases
    \begin{itemize}
        \item If $p\le a<b\le q$, set $\ndg{u}{v}{R}\leftarrow\min\{\ndg{u}{v}{R},\pp{ux\circ e\circ yv}\}$
        \item If $a<p\le b\le q$ , set $\ndg{u}{v}{R}\leftarrow\min\{\ndg{u}{v}{R},\pp{\odg{u}{x}{ap}\circ e\circ yv}\}$
        \item If $p\le a\le q< b$, set $\ndg{u}{v}{R}\leftarrow\min\{\ndg{u}{v}{R},\pp{ux\circ e\circ \odg{y}{v}{qb}}\}$
        \item If $a<p\le q<b$, do nothing
        \item If $a<b\le p<q$, set $\ndg{u}{v}{R}\leftarrow\min\{\ndg{u}{v}{R},\pp{\odg{u}{x}{R}\circ e\circ yv}\}$
        \item if $p<q\le a<b$, set $\ndg{u}{v}{R}\leftarrow\min\{\ndg{u}{v}{R},\pp{ux\circ e\circ\odg{y}{v}{R}}\}$
    \end{itemize}
\end{itemize}

We first prove some lemmas we need.
\begin{lemma}\label{lemma4-5}
    For any edge $f$, if $\nng{u}{v}{f}$ does not go through $e=(x,y)$, then $\nng{u}{v}{f}=\og{u}{v}{f}$.
\end{lemma}

\begin{proof}
    $\nng{u}{v}{f}$ not going through $e=(x,y)$ means that $\nng{u}{v}{f}$ is path in $G.$ In other words, it is the shortest $u-v$ path in $G$ that avoids $f$, which by definition equals $\og{u}{v}{f}.$
\end{proof}

\begin{lemma}\label{lemma4-6}
    For any edge $f\in pq$, if $\nng{u}{v}{f}$ goes through $e=(x,y)$ and $u$ is closer to $x$ than to $y$ in $G'$, then $\nng{u}{v}{f}=ux\circ e \circ yv$.
\end{lemma}

\begin{proof}
    By assumption, $\nng{u}{v}{f}$ goes through $e$ and $u$ is closer to $x$ than to $y$, so it has the form $\nng{u}{x}{f}\circ e\circ \nng{y}{v}{f}.$ Since $f\in pq$, it is not in $ux$ or $yv$, meaning that $\nng{u}{x}{f}=ux$ and $\nng{y}{v}{f}=yv.$
\end{proof}

\begin{lemma}\label{thm:not-pass-e}
    If $R$ is weak under $\pi_{G'}(u,v)$ in $G'$ and $\nng{u}{v}{R}$ does not go through $e$, then $\ndg{u}{v}{R}=\odg{u}{v}{R}$.
\end{lemma}
\begin{proof}
  Let $f$ be a weak point of $R$ under $\pi_{G'}(u,v)$. Then by Lemma \ref{lemma4-5}, $\ndg{u}{v}{R}=\nng{u}{v}{f}=\og{u}{v}{f}.$ Since $\og{u}{v}{f}=\ndg{u}{v}{R}$ avoids $R$ and $f\in R$, we know $R$ is also weak under $uv$ in graph $G$ with weak point $f$, so $\ndg{u}{v}{R}=\odg{u}{v}{R}.$
\end{proof}

In Lemma~\ref{thm:not-pass-e} we can see if $\odg{u}{v}{R}$ does not pass the transform $T_{G',R}$, $R$ is not weak or $\ndg{u}{v}{R}$ must pass $e$. So in the algorithm we first set $\ndg{u}{v}{R}$ to be $\pp{\odg{u}{v}{R}}$. Then we consider the cases that $\ndg{u}{v}{R}$ goes through $e=(x,y)$ ($u$ is closer to $x$ than to $y$):

\vspace{5pt}

\noindent\textbf{CASE 1:} $p\le a<b\le q$,

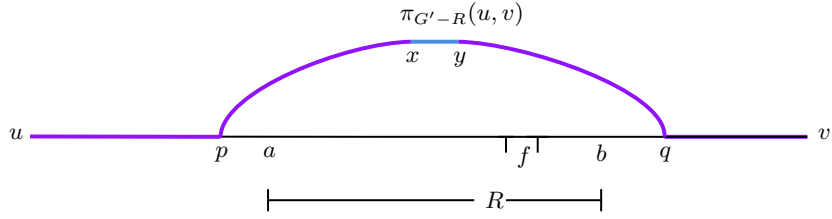
\begin{center}
    \tikzset{every picture/.style={line width=0.75pt}} 

\begin{tikzpicture}[x=0.75pt,y=0.75pt,yscale=-1,xscale=1]

\draw    (96,183.82) -- (272,184) ;
\draw [color={rgb, 255:red, 144; green, 19; blue, 254 }  ,draw opacity=0.5 ][line width=1.5]    (96,183.82) -- (192,184) ;
\draw [color={rgb, 255:red, 144; green, 19; blue, 254 }  ,draw opacity=0.5 ][line width=1.5]    (416,184) -- (488,184) ;
\draw [color={rgb, 255:red, 74; green, 144; blue, 226 }  ,draw opacity=0.5 ][line width=1.5]    (288,136) -- (312,136) ;
\draw    (272,184) -- (488,184) ;
\draw    (320,216) -- (216,216) ;
\draw [shift={(216,216)}, rotate = 360] [color={rgb, 255:red, 0; green, 0; blue, 0 }  ][line width=0.75]    (0,5.59) -- (0,-5.59)   ;
\draw    (336,216) -- (384,216) ;
\draw [shift={(384,216)}, rotate = 180] [color={rgb, 255:red, 0; green, 0; blue, 0 }  ][line width=0.75]    (0,5.59) -- (0,-5.59)   ;
\draw    (336,192) -- (336,184) ;
\draw [shift={(336,184)}, rotate = 90] [color={rgb, 255:red, 0; green, 0; blue, 0 }  ][line width=0.75]    (0,3.91) -- (0,-3.91)   ;
\draw    (352,192) -- (352,184) ;
\draw [shift={(352,184)}, rotate = 90] [color={rgb, 255:red, 0; green, 0; blue, 0 }  ][line width=0.75]    (0,5.59) -- (0,-5.59)   ;
\draw [color={rgb, 255:red, 144; green, 19; blue, 254 }  ,draw opacity=0.5 ][line width=1.5]    (192,184) .. controls (192.87,160.56) and (255.87,136.56) .. (288,136) ;
\draw [color={rgb, 255:red, 144; green, 19; blue, 254 }  ,draw opacity=0.5 ][line width=1.5]    (416,184) .. controls (416.87,160.56) and (344.37,136.56) .. (312,136) ;

\draw (492,178.4) node [anchor=north west][inner sep=0.75pt]  [font=\footnotesize]  {$v$};
\draw (84,178.4) node [anchor=north west][inner sep=0.75pt]  [font=\footnotesize]  {$u$};
\draw (284,139.4) node [anchor=north west][inner sep=0.75pt]  [font=\footnotesize]  {$x$};
\draw (188,187.4) node [anchor=north west][inner sep=0.75pt]  [font=\footnotesize]  {$p$};
\draw (412,187.4) node [anchor=north west][inner sep=0.75pt]  [font=\footnotesize]  {$q$};
\draw (308,139.4) node [anchor=north west][inner sep=0.75pt]  [font=\footnotesize]  {$y$};
\draw (212,187.4) node [anchor=north west][inner sep=0.75pt]  [font=\footnotesize]  {$a$};
\draw (380,186.4) node [anchor=north west][inner sep=0.75pt]  [font=\footnotesize]  {$b$};
\draw (324,210.4) node [anchor=north west][inner sep=0.75pt]  [font=\footnotesize]  {$R$};
\draw (340,187.4) node [anchor=north west][inner sep=0.75pt]  [font=\footnotesize]  {$f$};
\draw (281,114.4) node [anchor=north west][inner sep=0.75pt]  [font=\footnotesize]  {$\pi _{G'-R}( u,v)$};

\end{tikzpicture}
\end{center}

By Lemma~\ref{lemma4-6}, we set $\ndg{u}{v}{R}=\pp{ux\circ e\circ yv}$. 


(Note that when $\ndg{u}{v}{R}=ux\circ e\circ yv$, we do not know whether $R$ is weak, and even if it is weak we cannot find the weak point. We only know that if $R$ is weak then it is correct. That is the main reason why $\varpi$ is defined in this way.)

 
~\\
\noindent\textbf{CASE 2:} $a<p\le b\le q$,


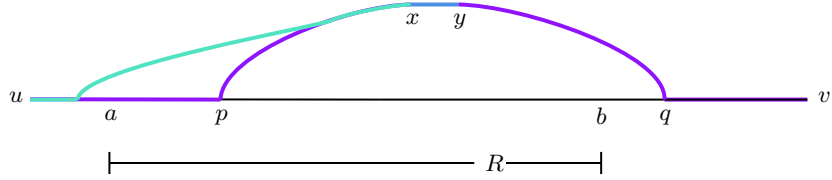
\begin{center}
    \tikzset{every picture/.style={line width=0.75pt}} 

\begin{tikzpicture}[x=0.75pt,y=0.75pt,yscale=-1,xscale=1]

\draw    (96,183.82) -- (272,184) ;
\draw [color={rgb, 255:red, 144; green, 19; blue, 254 }  ,draw opacity=0.5 ][line width=1.5]    (96,183.82) -- (192,184) ;
\draw [color={rgb, 255:red, 144; green, 19; blue, 254 }  ,draw opacity=0.5 ][line width=1.5]    (416,184) -- (488,184) ;
\draw [color={rgb, 255:red, 74; green, 144; blue, 226 }  ,draw opacity=0.5 ][line width=1.5]    (288,136) -- (312,136) ;
\draw    (272,184) -- (488,184) ;
\draw    (320,216) -- (136,216) ;
\draw [shift={(136,216)}, rotate = 360] [color={rgb, 255:red, 0; green, 0; blue, 0 }  ][line width=0.75]    (0,5.59) -- (0,-5.59)   ;
\draw    (336,216) -- (384,216) ;
\draw [shift={(384,216)}, rotate = 180] [color={rgb, 255:red, 0; green, 0; blue, 0 }  ][line width=0.75]    (0,5.59) -- (0,-5.59)   ;
\draw [color={rgb, 255:red, 144; green, 19; blue, 254 }  ,draw opacity=0.5 ][line width=1.5]    (192,184) .. controls (192.87,160.56) and (255.87,136.56) .. (288,136) ;
\draw [color={rgb, 255:red, 144; green, 19; blue, 254 }  ,draw opacity=0.5 ][line width=1.5]    (416,184) .. controls (416.87,160.56) and (344.37,136.56) .. (312,136) ;
\draw [color={rgb, 255:red, 80; green, 227; blue, 194 }  ,draw opacity=0.4 ][line width=1.5]    (120,184) .. controls (120.61,167.5) and (240.85,147.36) .. (248,144) .. controls (255.15,140.64) and (278.5,136.17) .. (288,136) ;
\draw [color={rgb, 255:red, 80; green, 227; blue, 194 }  ,draw opacity=0.4 ][line width=1.5]    (96,184) -- (120,184) ;

\draw (492,178.4) node [anchor=north west][inner sep=0.75pt]  [font=\footnotesize]  {$v$};
\draw (84,178.4) node [anchor=north west][inner sep=0.75pt]  [font=\footnotesize]  {$u$};
\draw (284,139.4) node [anchor=north west][inner sep=0.75pt]  [font=\footnotesize]  {$x$};
\draw (188,187.4) node [anchor=north west][inner sep=0.75pt]  [font=\footnotesize]  {$p$};
\draw (412,187.4) node [anchor=north west][inner sep=0.75pt]  [font=\footnotesize]  {$q$};
\draw (308,139.4) node [anchor=north west][inner sep=0.75pt]  [font=\footnotesize]  {$y$};
\draw (132,187.4) node [anchor=north west][inner sep=0.75pt]  [font=\footnotesize]  {$a$};
\draw (380,186.4) node [anchor=north west][inner sep=0.75pt]  [font=\footnotesize]  {$b$};
\draw (324,210.4) node [anchor=north west][inner sep=0.75pt]  [font=\footnotesize]  {$R$};

\end{tikzpicture}
\end{center}

If $\nng{u}{v}{R}$ goes through $e$, from the assumption that $R$ is weak, if its weak point $f\in pb$, by Lemma~\ref{lemma4-6}, $\nng{u}{v}{f}=ux\circ e\circ yv$, so it intersects with $R$, thus the weak point $f$ of $R$ must be in $ap$. Then $\nng{u}{v}{f}=\nng{u}{v}{R}$ which goes through $e$, so $\nng{u}{v}{f}=\nng{u}{x}{f}\circ e\circ \nng{y}{v}{f}=\og{u}{x}{f}\circ e\circ yv$. 
By the unique shortest path assumption, this path also does not go through $R$, thus $\og{u}{x}{f}$ does not intersect with $ap$, so $ap$ is weak under $ux$ with weak point $f$. Thus, we get:
$$\ndg{u}{v}{R}=\pp{\odg{u}{x}{ap}\circ e\circ yv}$$

Thus if $\odg{u}{x}{ap}\circ e\circ yv$ does not pass the transform $T_{G',R}$, either $\nng{u}{v}{R}$ does not go through $e$, or $R$ is not weak in $\pi_{G'}(u,v)$.


~\\
\noindent\textbf{CASE 3:} $p\le a\le q< b$,

This is symmetric to \textbf{CASE 2}. $$\ndg{u}{v}{R}=\pp{ux\circ e\circ \odg{y}{v}{qb}}.$$

~\\
\noindent\textbf{CASE 4:} $a<p\le q<b$,

We can prove that in this case if $R$ is weak then $\nng{u}{v}{R}$ cannot go through $e$. 


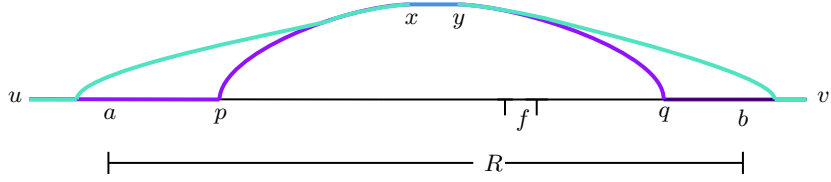
\begin{center}
    \tikzset{every picture/.style={line width=0.75pt}} 

\begin{tikzpicture}[x=0.75pt,y=0.75pt,yscale=-1,xscale=1]

\draw    (96,183.82) -- (272,184) ;
\draw [color={rgb, 255:red, 144; green, 19; blue, 254 }  ,draw opacity=0.5 ][line width=1.5]    (96,183.82) -- (192,184) ;
\draw [color={rgb, 255:red, 144; green, 19; blue, 254 }  ,draw opacity=0.5 ][line width=1.5]    (416,184) -- (488,184) ;
\draw [color={rgb, 255:red, 74; green, 144; blue, 226 }  ,draw opacity=0.5 ][line width=1.5]    (288,136) -- (312,136) ;
\draw    (272,184) -- (488,184) ;
\draw    (320,216) -- (136,216) ;
\draw [shift={(136,216)}, rotate = 360] [color={rgb, 255:red, 0; green, 0; blue, 0 }  ][line width=0.75]    (0,5.59) -- (0,-5.59)   ;
\draw    (336,216) -- (456,216) ;
\draw [shift={(456,216)}, rotate = 180] [color={rgb, 255:red, 0; green, 0; blue, 0 }  ][line width=0.75]    (0,5.59) -- (0,-5.59)   ;
\draw    (336,192) -- (336,184) ;
\draw [shift={(336,184)}, rotate = 90] [color={rgb, 255:red, 0; green, 0; blue, 0 }  ][line width=0.75]    (0,3.91) -- (0,-3.91)   ;
\draw    (352,192) -- (352,184) ;
\draw [shift={(352,184)}, rotate = 90] [color={rgb, 255:red, 0; green, 0; blue, 0 }  ][line width=0.75]    (0,5.59) -- (0,-5.59)   ;
\draw [color={rgb, 255:red, 144; green, 19; blue, 254 }  ,draw opacity=0.5 ][line width=1.5]    (192,184) .. controls (192.87,160.56) and (255.87,136.56) .. (288,136) ;
\draw [color={rgb, 255:red, 144; green, 19; blue, 254 }  ,draw opacity=0.5 ][line width=1.5]    (416,184) .. controls (416.87,160.56) and (344.37,136.56) .. (312,136) ;
\draw [color={rgb, 255:red, 80; green, 227; blue, 194 }  ,draw opacity=0.4 ][line width=1.5]    (120,184) .. controls (120.61,167.5) and (240.85,147.36) .. (248,144) .. controls (255.15,140.64) and (278.5,136.17) .. (288,136) ;
\draw [color={rgb, 255:red, 80; green, 227; blue, 194 }  ,draw opacity=0.4 ][line width=1.5]    (472,184) .. controls (472.5,170.55) and (368,146.4) .. (360,144) .. controls (352,141.6) and (320.67,135.73) .. (312,136) ;
\draw [color={rgb, 255:red, 80; green, 227; blue, 194 }  ,draw opacity=0.4 ][line width=1.5]    (96,184) -- (120,184) ;
\draw [color={rgb, 255:red, 80; green, 227; blue, 194 }  ,draw opacity=0.4 ][line width=1.5]    (472,184) -- (488,184) ;

\draw (492,178.4) node [anchor=north west][inner sep=0.75pt]  [font=\footnotesize]  {$v$};
\draw (84,178.4) node [anchor=north west][inner sep=0.75pt]  [font=\footnotesize]  {$u$};
\draw (284,139.4) node [anchor=north west][inner sep=0.75pt]  [font=\footnotesize]  {$x$};
\draw (188,187.4) node [anchor=north west][inner sep=0.75pt]  [font=\footnotesize]  {$p$};
\draw (412,186.4) node [anchor=north west][inner sep=0.75pt]  [font=\footnotesize]  {$q$};
\draw (308,139.4) node [anchor=north west][inner sep=0.75pt]  [font=\footnotesize]  {$y$};
\draw (132,187.4) node [anchor=north west][inner sep=0.75pt]  [font=\footnotesize]  {$a$};
\draw (452,187.4) node [anchor=north west][inner sep=0.75pt]  [font=\footnotesize]  {$b$};
\draw (324,210.4) node [anchor=north west][inner sep=0.75pt]  [font=\footnotesize]  {$R$};
\draw (340,187.4) node [anchor=north west][inner sep=0.75pt]  [font=\footnotesize]  {$f$};

\end{tikzpicture}
\end{center}

Prove by contradiction. Suppose $R$ is weak under $\pi_{G'}(u,v)$ with weak point $f$ and $\nng{u}{v}{R}$ goes through $e$ (first $x$ then $y$), then $\pi_{G'}(u,x)=ux$ and $\pi_{G'}(y,v)=yv$ since they cannot go through $e$. We know $\nng{u}{v}{f}$ is a 1-replacement-path in $G'$, so by Theorem~\ref{ReplacementPath}, it equals a concatenation of a shortest path, an edge, and a shortest path in $G'$. However, $a<p$ means that $\nng{u}{x}{R}$ is different from $ux$, so it is a concatenation of at least two shortest paths in $G'$. Same for $\nng{y}{v}{R}$. Therefore, $\nng{u}{v}{R}$ equals the shortest $u-x$ path avoiding $R$ plus $e$ plus the shortest $y-v$ path avoiding $R$, which will be a concatenation of at least three shortest paths in $G$, contradicting to $R$ is weak.
\\

\vspace{5pt}
\noindent\textbf{CASE 5:} $a<b\le p<q$,
$$\ndg{u}{v}{R}=\pp{\odg{u}{x}{R}\circ e\circ yv}.$$

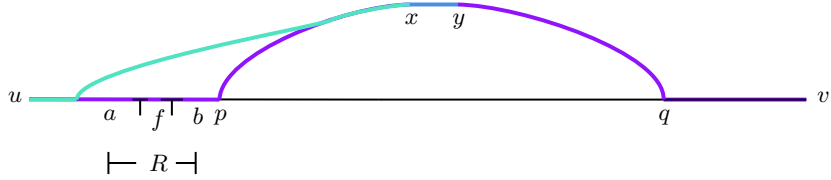
\begin{center}
    \tikzset{every picture/.style={line width=0.75pt}} 

\begin{tikzpicture}[x=0.75pt,y=0.75pt,yscale=-1,xscale=1]

\draw    (96,183.82) -- (272,184) ;
\draw [color={rgb, 255:red, 144; green, 19; blue, 254 }  ,draw opacity=0.5 ][line width=1.5]    (96,183.82) -- (192,184) ;
\draw [color={rgb, 255:red, 144; green, 19; blue, 254 }  ,draw opacity=0.5 ][line width=1.5]    (416,184) -- (488,184) ;
\draw [color={rgb, 255:red, 74; green, 144; blue, 226 }  ,draw opacity=0.5 ][line width=1.5]    (288,136) -- (312,136) ;
\draw    (272,184) -- (488,184) ;
\draw    (150,216) -- (136,216) ;
\draw [shift={(136,216)}, rotate = 360] [color={rgb, 255:red, 0; green, 0; blue, 0 }  ][line width=0.75]    (0,5.59) -- (0,-5.59)   ;
\draw    (170,216) -- (180,216) ;
\draw [shift={(180,216)}, rotate = 180] [color={rgb, 255:red, 0; green, 0; blue, 0 }  ][line width=0.75]    (0,5.59) -- (0,-5.59)   ;
\draw    (152,192) -- (152,184) ;
\draw [shift={(152,184)}, rotate = 90] [color={rgb, 255:red, 0; green, 0; blue, 0 }  ][line width=0.75]    (0,3.91) -- (0,-3.91)   ;
\draw    (168,192) -- (168,184) ;
\draw [shift={(168,184)}, rotate = 90] [color={rgb, 255:red, 0; green, 0; blue, 0 }  ][line width=0.75]    (0,5.59) -- (0,-5.59)   ;
\draw [color={rgb, 255:red, 144; green, 19; blue, 254 }  ,draw opacity=0.5 ][line width=1.5]    (192,184) .. controls (192.87,160.56) and (255.87,136.56) .. (288,136) ;
\draw [color={rgb, 255:red, 144; green, 19; blue, 254 }  ,draw opacity=0.5 ][line width=1.5]    (416,184) .. controls (416.87,160.56) and (344.37,136.56) .. (312,136) ;
\draw [color={rgb, 255:red, 80; green, 227; blue, 194 }  ,draw opacity=0.4 ][line width=1.5]    (120,184) .. controls (120.61,167.5) and (240.85,147.36) .. (248,144) .. controls (255.15,140.64) and (278.5,136.17) .. (288,136) ;
\draw [color={rgb, 255:red, 80; green, 227; blue, 194 }  ,draw opacity=0.4 ][line width=1.5]    (96,184) -- (120,184) ;

\draw (492,178.4) node [anchor=north west][inner sep=0.75pt]  [font=\footnotesize]  {$v$};
\draw (84,178.4) node [anchor=north west][inner sep=0.75pt]  [font=\footnotesize]  {$u$};
\draw (284,139.4) node [anchor=north west][inner sep=0.75pt]  [font=\footnotesize]  {$x$};
\draw (188,187.4) node [anchor=north west][inner sep=0.75pt]  [font=\footnotesize]  {$p$};
\draw (412,187.4) node [anchor=north west][inner sep=0.75pt]  [font=\footnotesize]  {$q$};
\draw (308,139.4) node [anchor=north west][inner sep=0.75pt]  [font=\footnotesize]  {$y$};
\draw (132,187.4) node [anchor=north west][inner sep=0.75pt]  [font=\footnotesize]  {$a$};
\draw (177,186.4) node [anchor=north west][inner sep=0.75pt]  [font=\footnotesize]  {$b$};
\draw (155,210.4) node [anchor=north west][inner sep=0.75pt]  [font=\footnotesize]  {$R$};
\draw (156,187.4) node [anchor=north west][inner sep=0.75pt]  [font=\footnotesize]  {$f$};

\end{tikzpicture}
\end{center}

Suppose that $R$ is weak under $uv$ with a weak point $f.$ By assumption, $\nng{u}{v}{f}$ goes through $e,$ so it equals $\nng{u}{x}{f}\circ e\circ\nng{y}{v}{f}.$ By Lemma~\ref{lemma4-5}, $\nng{y}{v}{f}=\og{y}{v}{f}=yv$ and $\nng{u}{x}{f}=\og{u}{x}{f}$. We can see that $R$ is also weak under $ux$ with weak point $f$ in $G'$, because $\nng{u}{x}{f}$ is a part of $\nng{u}{v}{f}$, which avoids $R.$ Therefore, by Lemma~\ref{thm:not-pass-e}, $\nng{u}{x}{f}=\nng{u}{x}{R}=\odg{u}{x}{R}.$

\vspace{5pt}
\noindent\textbf{CASE 6:} $p<q\le a<b$,

This is symmetric to \textbf{CASE 5}.
$$\ndg{u}{v}{R}=\pp{ux\circ e\circ\odg{y}{v}{R}}.$$

\section{One Failed Edge on Original Shortest Path}\label{sec:1-edge}

With the incremental DSO, we can proceed to solve the 3FRP problem for given vertices $s$ and $t$. Here, by Theorem~\ref{thm:offline}, we transform our incremental DSO in Section~\ref{sec-inc} to an offline fully dynamic DSO. In this section, we suppose only one failed edge, named $d_1$, is on the original shortest path $st$. If $d_2, d_3$ are both not on $\og{s}{t}{d_1}$, then the replacement path will be $\og{s}{t}{d_1}$. Therefore, we suppose that $d_2$ is on $\og{s}{t}{d_1}$, and we can restore all edges $d_3 \in \og{s}{t}{\set{d_1, d_2}}$ based on Section \ref{sec:one-2FRP}.

\begin{theorem}
There exists an algorithm that, given $G$ with two vertices $s,t \in V(G)$, for all edges $d_1 \in st$, $d_2,d_3 \not \in st$, $d_2 \in \og{s}{t}{d_1}$ and $d_3 \in \og{s}{t}{\set{d_1, d_2}}$, answers $|\og{s}{t}{\set{d_1, d_2, d_3}}|$ in $\too{n^3}$ time.
\end{theorem}

\subsection{The Graph $H$}\label{sec2-1}

As in Section~\ref{sec:one-2FRP}, we construct a graph $H$ as follows: (Here let $N$ be the sum of all edge weights in $G$.)

\begin{itemize}
    \item First let $H$ be a copy of $G-st$, that is, remove all edges on $\pi_G(s,t)$.
    \item For every edge $d=(x,y)\in st$ and $x$ is before $y$ on $st$, introduce two new vertices $d^-$ and $d^+$. 
    \begin{itemize}
        \item For every vertex $x'$ on $sx$, the vertex $d^-$ is connected with $x'$ by an edge with weight $|sx'|+N$.
        \item For every vertex $y'$ on $yt$, the vertex $d^+$ is connected with $y'$ by an edge with weight $|y't|+N$.
    \end{itemize}
\end{itemize}



Here, since $N$ is large enough, we can easily observe that for any shortest path $\pi_{H}(u,v)$ in $H$, if $u,v\in G$, it does not use the new edges; and if one or both of them are new vertices, new edges can only appear at both ends in $\pi_{H}(u,v)$.

We proceed to our lemma for this case:

\begin{lemma} \label{lemma:1edge}
    Given $G$ and one failure $d_1$ on $st$, for any $d_2, d_3$ not on $st$, in the graph $H$ we defined,
    \[|\hg{d_1^-}{d_1^+}{\set{d_2, d_3}}| = |\og{s}{t}{\set{d_1, d_2, d_3}}| + 2N.\]
\end{lemma}

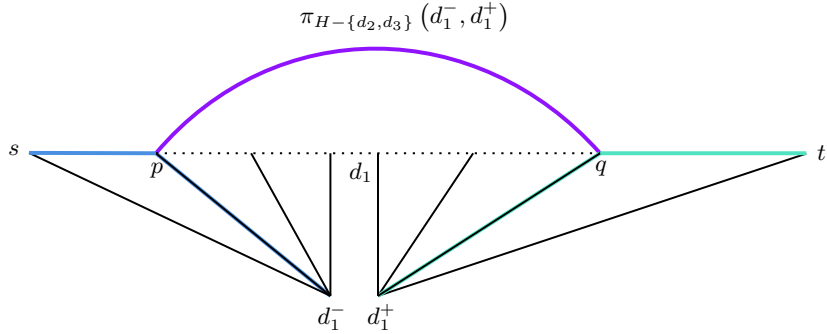
\begin{center}
    \tikzset{every picture/.style={line width=0.75pt}} 

\begin{tikzpicture}[x=0.75pt,y=0.75pt,yscale=-1,xscale=1]

\draw  [dash pattern={on 0.84pt off 2.51pt}]  (96,431.82) -- (160,432) ;
\draw [color={rgb, 255:red, 74; green, 144; blue, 226 }  ,draw opacity=0.5 ][line width=1.5]    (248,504) -- (160,432) ;
\draw [color={rgb, 255:red, 144; green, 19; blue, 254 }  ,draw opacity=0.5 ][line width=1.5]    (160,432) .. controls (220.2,360.8) and (322.2,362.4) .. (384,432) ;
\draw [color={rgb, 255:red, 0; green, 0; blue, 0 }  ,draw opacity=0.2 ]   (96,431.82) -- (248,504) ;
\draw [color={rgb, 255:red, 0; green, 0; blue, 0 }  ,draw opacity=0.2 ]   (208,432) -- (248,504) ;
\draw [color={rgb, 255:red, 0; green, 0; blue, 0 }  ,draw opacity=0.2 ]   (248,432) -- (248,504) ;
\draw [color={rgb, 255:red, 0; green, 0; blue, 0 }  ,draw opacity=0.2 ]   (320,432) -- (272,504) ;
\draw [color={rgb, 255:red, 0; green, 0; blue, 0 }  ,draw opacity=0.2 ]   (488,432.18) -- (272,504) ;
\draw [color={rgb, 255:red, 80; green, 227; blue, 194 }  ,draw opacity=0.5 ][line width=1.5]    (272,504) -- (384,432) ;
\draw [color={rgb, 255:red, 0; green, 0; blue, 0 }  ,draw opacity=0.2 ]   (272,432) -- (272,504) ;
\draw  [dash pattern={on 0.84pt off 2.51pt}]  (384,432) -- (488,432.18) ;
\draw  [dash pattern={on 0.84pt off 2.51pt}]  (160,432) -- (384,432) ;
\draw [color={rgb, 255:red, 80; green, 227; blue, 194 }  ,draw opacity=0.5 ][line width=1.5]    (384,432) -- (488,432) ;
\draw [color={rgb, 255:red, 0; green, 0; blue, 0 }  ,draw opacity=0.2 ]   (384,432) -- (272,504) ;
\draw [color={rgb, 255:red, 0; green, 0; blue, 0 }  ,draw opacity=0.2 ]   (160,432) -- (248,504) ;
\draw [color={rgb, 255:red, 74; green, 144; blue, 226 }  ,draw opacity=0.5 ][line width=1.5]    (160,432) -- (96,431.82) ;

\draw (84,426.4) node [anchor=north west][inner sep=0.75pt]  [font=\footnotesize]  {$s$};
\draw (231,354.4) node [anchor=north west][inner sep=0.75pt]  [font=\footnotesize]  {$\pi _{H-\{d_{2} ,d_{3}\}}\left( d_{1}^{-} ,d_{1}^{+}\right)$};
\draw (492,426.4) node [anchor=north west][inner sep=0.75pt]  [font=\footnotesize]  {$t$};
\draw (240,506.4) node [anchor=north west][inner sep=0.75pt]  [font=\footnotesize]  {$d_{1}^{-}$};
\draw (265,506.4) node [anchor=north west][inner sep=0.75pt]  [font=\footnotesize]  {$d_{1}^{+}$};
\draw (256,435.4) node [anchor=north west][inner sep=0.75pt]  [font=\footnotesize]  {$d_{1}$};
\draw (156,435.4) node [anchor=north west][inner sep=0.75pt]  [font=\footnotesize]  {$p$};
\draw (380,434.4) node [anchor=north west][inner sep=0.75pt]  [font=\footnotesize]  {$q$};

\end{tikzpicture}
\end{center}

\begin{proof}
    Consider the path $\og{s}{t}{\set{d_1, d_2, d_3}}$. Suppose in $G$, $\og{s}{t}{\set{d_1, d_2, d_3}}$ diverges at point $p$ before $d_1$ and converges at point $q$ after $d_1$. (Recall that diverging means the path first gets into an edge out of $st$ and converging means the opposite.) In $H-\set{d_2, d_3}$ we can use the edge $(d_1^-,p)$ to replace the part $sp$ and the edge $(q,d_1^+)$ to replace the part $qt$. 
    
    Since $w(d_1^-,p) = |sp| + N$ and $w(q,d_1^+) = |qt| + N$, we know the path from $d_1^-$ to $d_1^+$ in $H-d_2-d_3$ is with weight $|\og{s}{t}{\set{d_1, d_2, d_3}}|+2N$ in $G$. Therefore, 
    $$|\hg{d_1^-}{d_1^+}{\set{d_2, d_3}}| \leq |\og{s}{t}{\set{d_1, d_2, d_3}}| + 2N.$$

    Conversely, take the shortest path from $d_1^-$ to $d_1^+$ in $H-\set{d_2, d_3}$, we can use the original shortest path $sp$ to replace the first edge $(d_1^-,p)$ and $qt$ to replace the last edge $(q,d_1^+)$. Since $N$ is large enough, the other edges will fall in $G$. Since $w(d_1^-,p) = N + |sp|$ and $w(q,d_1^+) = N + |qt|$, we know the path from $s$ to $t$ in $G-\set{d_1,d_2,d_3}$ is with weight $|\hg{d_1^-}{d_1^+}{\set{d_2, d_3}}|-2N$ in $H-\set{d_2, d_3}$.Therefore, 
    $$|\hg{d_1^-}{d_1^+}{\set{d_2, d_3}}| \geq |\og{s}{t}{\set{d_1, d_2, d_3}}| + 2N.$$

    Therefore, the two values are equal.
\end{proof}

\subsection{The Algorithm}

We consider every $d_2$ that is on $\og{s}{t}{d_1}$ for some $d_1$ on $st$. By Lemma \ref{thm:n-edges}, the number of possible $d_2$ is $O(n)$. We start from the graph $H$ and consider every $d_2$ in any predetermined order and use the offline dynamic DSO from Section~\ref{sec-inc} and Theorem~\ref{thm:offline}. 

For every $d_2$ we considered, we first delete $d_2$ from the graph to get $H-d_2$. Then for every $d_1$ on the original shortest path $st$, and every $d_3$ on $\og{s}{t}{\set{d_1, d_2}}$, we calculate the value $|\hg{d_1^-}{d_1^+}{\set{d_2, d_3}}|-2N$ in this graph $H-d_2$ using our oracle. Finally, we add $d_2$ back to proceed.

By Lemma~\ref{lemma:1edge}, the value $|\hg{d_1^-}{d_1^+}{\set{d_2, d_3}}|-2N$ we computed equals $|\og{s}{t}{\set{d_1, d_2, d_3}}|$. Therefore, we have got the correct replacement path lengths for ${d_1,d_2,d_3}$.

We have deleted and added $O(n)$ edges in total and queried the oracle once for every triple $(d_1,d_2,d_3)$. By Section~\ref{sec-inc} and Theorem~\ref{thm:offline}, the total update time is in $\too{n^3}$, and each query can be performed in $\tilde{O}(1)$ time. Therefore, the total time we use is $\too{n^3}$. 





\section{Two Failed Edges on Original Shortest Path}\label{sec:2-edge}

In this section, we consider the case that exactly two failures $d_1,d_2$ are on the original shortest path $st$. Our approach follows the framework and adapts some techniques from \cite{duan2009dual}.

In this section, without loss of generality, we assume there are $2^k$ edges on $st$ for some integer $k=O(\log n)$. Otherwise, for any $\| st \|$ let $k=\lceil \log \| st \| \rceil = O(\log n)$, we can add an extra path from $s'$ to $s$ with $2^k - \| st \|$ edges with weight $0$ and consider the 3FRP problem for $(s',t)$ to solve the original 3FRP problem for $(s,t)$. With this assumption, we have $t=s \oplus 2^k$ on the original shortest path $st$. 

\begin{theorem}
There exists an algorithm that, given $G$ with two points $s,t \in V(G)$, for all edges $d_1,d_2 \in st$, $d_3 \not \in st$ and $d_3 \in \og{s}{t}{\set{d_1, d_2}}$, answers $|\og{s}{t}{\set{d_1, d_2, d_3}}|$ in $\too{n^3}$ time.
\end{theorem}

\subsection{The Graph $H$}

We use the same graph $H$ that is obtained by adding new vertices $d^-$ and $d^+$ and their edges to $G-st$, as defined in Section~\ref{sec2-1}. Also, we use the interval notation $[x,y]$ for $x,y\in st$ and $x$ is closer to $s$, to denote the subpath $xy$ on $st$ in the original graph $G$. 

Consider two edge failures $d_1=(s \oplus (l-1), s \oplus l),d_2=(s \oplus r, s \oplus (r+1))$ on $st$, $0 < l \leq r < 2^k$. Define $R=[s \oplus l, s\oplus r]$, that is, the interval between $d_1$ and $d_2$, but not including them. Remind that in $H$, edges from $d_1^-$ may replace the subpaths in the interval $[s, s \oplus (l-1)]$, and edges from $d_2^+$ may replace the subpaths in the interval $[s \oplus (r+1), t]$, but the detour can still go through edges in $R$. Before introducing our algorithm, we first prove some properties on $H$.

Consider any failure set $D$ that is in $(G-st) \cup R$. Similarly, we may use the shortest path between $d_1^-$ and $d_2^+$ that avoids $D$ in $H \cup R$ to represent $\og{s}{t}{(\set{d_1,d_2} \cup D)}$ in the original graph $G$. 

Moreover, consider edges in $[s,s\oplus(l-1)]$ and $[s\oplus(r+1),t]$. Adding an arbitrary subset of them to the graph $H \cup R$ will not influence the distance between $d_1^-$ and $d_2^+$, as the edges from $d_1^-$ to any vertex in $[s,s\oplus(l-1)]$ and (similarly) from $d_2^+$ to any vertex in $[s\oplus(r+1),t]$ will not be shortcut by adding these edges. We demonstrate the results in the following lemma:

\begin{lemma} \label{lemma:2edge}
    Given $G$ and two failures $d_1,d_2$ on $st$, consider arbitrary edge subsets $E_1 \subseteq [s,s\oplus(l-1)]$ and $E_2 \subseteq [s\oplus(r+1),t]$. Define $H' = (H \cup R \cup E_1 \cup E_2)$ to be the graph obtained by adding the whole $R=[s \oplus l, s \oplus r]$ and $E_1,E_2$ to $H$. Then for any edge set $D \subseteq (G-st) \cup R$,
    \[|\pi_{(H'-D)}(d_1^-, d_2^+)| = |\og{s}{t}{(\set{d_1, d_2}\cup D)}|+2N.\]
    
\end{lemma}

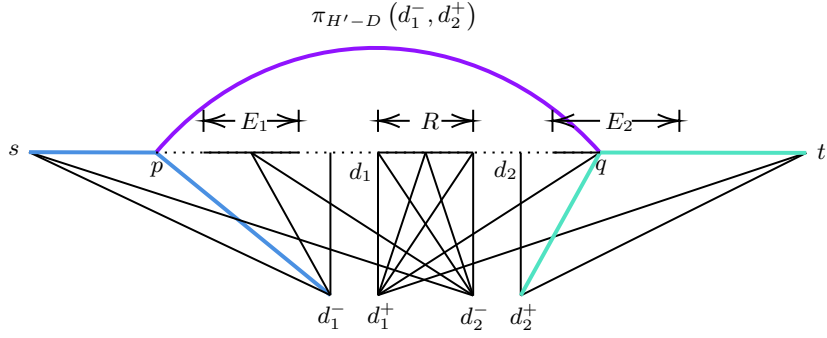
\begin{center}
    \tikzset{every picture/.style={line width=0.75pt}} 

\begin{tikzpicture}[x=0.75pt,y=0.75pt,yscale=-1,xscale=1]

\draw  [dash pattern={on 0.84pt off 2.51pt}]  (96,431.82) -- (160,432) ;
\draw [color={rgb, 255:red, 74; green, 144; blue, 226 }  ,draw opacity=0.5 ][line width=1.5]    (248,504) -- (160,432) ;
\draw [color={rgb, 255:red, 144; green, 19; blue, 254 }  ,draw opacity=0.5 ][line width=1.5]    (160,432) .. controls (220.2,360.8) and (322.2,362.4) .. (384,432) ;
\draw [color={rgb, 255:red, 0; green, 0; blue, 0 }  ,draw opacity=0.2 ]   (96,431.82) -- (248,504) ;
\draw [color={rgb, 255:red, 0; green, 0; blue, 0 }  ,draw opacity=0.2 ]   (208,432) -- (248,504) ;
\draw [color={rgb, 255:red, 0; green, 0; blue, 0 }  ,draw opacity=0.2 ]   (320,432) -- (272,504) ;
\draw [color={rgb, 255:red, 0; green, 0; blue, 0 }  ,draw opacity=0.2 ]   (488,432.18) -- (272,504) ;
\draw [color={rgb, 255:red, 0; green, 0; blue, 0 }  ,draw opacity=0.2 ]   (272,432) -- (272,504) ;
\draw [color={rgb, 255:red, 0; green, 0; blue, 0 }  ,draw opacity=0.2 ]   (272,432) -- (320,504.18) ;
\draw [color={rgb, 255:red, 0; green, 0; blue, 0 }  ,draw opacity=0.2 ]   (296,432) -- (320,504.18) ;
\draw [color={rgb, 255:red, 0; green, 0; blue, 0 }  ,draw opacity=0.2 ]   (320,432.18) -- (320,504.18) ;
\draw [color={rgb, 255:red, 0; green, 0; blue, 0 }  ,draw opacity=0.2 ]   (488,432.18) -- (344,504.18) ;
\draw [color={rgb, 255:red, 0; green, 0; blue, 0 }  ,draw opacity=0.2 ]   (344,432.18) -- (344,504.18) ;
\draw [color={rgb, 255:red, 0; green, 0; blue, 0 }  ,draw opacity=0.2 ]   (296,432) -- (272,504.18) ;
\draw [color={rgb, 255:red, 0; green, 0; blue, 0 }  ,draw opacity=0.2 ]   (208,432) -- (320,504.18) ;
\draw [color={rgb, 255:red, 0; green, 0; blue, 0 }  ,draw opacity=0.2 ]   (96,431.82) -- (320,504.18) ;
\draw [color={rgb, 255:red, 0; green, 0; blue, 0 }  ,draw opacity=0.2 ]   (384,432) -- (272,504) ;
\draw [color={rgb, 255:red, 0; green, 0; blue, 0 }  ,draw opacity=0.3 ]   (248,432) -- (248,504) ;
\draw  [dash pattern={on 0.84pt off 2.51pt}]  (384,432) -- (488,432.18) ;
\draw  [dash pattern={on 0.84pt off 2.51pt}]  (160,432) -- (384,432) ;
\draw [color={rgb, 255:red, 80; green, 227; blue, 194 }  ,draw opacity=0.5 ][line width=1.5]    (344,504.18) -- (384,432) ;
\draw    (272,416) -- (288,416) ;
\draw [shift={(272,416)}, rotate = 0] [color={rgb, 255:red, 0; green, 0; blue, 0 }  ][line width=0.75]    (0,5.59) -- (0,-5.59)(10.93,-3.29) .. controls (6.95,-1.4) and (3.31,-0.3) .. (0,0) .. controls (3.31,0.3) and (6.95,1.4) .. (10.93,3.29)   ;
\draw    (304,416) -- (320,416) ;
\draw [shift={(320,416)}, rotate = 180] [color={rgb, 255:red, 0; green, 0; blue, 0 }  ][line width=0.75]    (0,5.59) -- (0,-5.59)(10.93,-3.29) .. controls (6.95,-1.4) and (3.31,-0.3) .. (0,0) .. controls (3.31,0.3) and (6.95,1.4) .. (10.93,3.29)   ;
\draw    (184,416) -- (200,416) ;
\draw [shift={(184,416)}, rotate = 0] [color={rgb, 255:red, 0; green, 0; blue, 0 }  ][line width=0.75]    (0,5.59) -- (0,-5.59)(10.93,-3.29) .. controls (6.95,-1.4) and (3.31,-0.3) .. (0,0) .. controls (3.31,0.3) and (6.95,1.4) .. (10.93,3.29)   ;
\draw    (216,416) -- (232,416) ;
\draw [shift={(232,416)}, rotate = 180] [color={rgb, 255:red, 0; green, 0; blue, 0 }  ][line width=0.75]    (0,5.59) -- (0,-5.59)(10.93,-3.29) .. controls (6.95,-1.4) and (3.31,-0.3) .. (0,0) .. controls (3.31,0.3) and (6.95,1.4) .. (10.93,3.29)   ;
\draw    (360,416) -- (384,416) ;
\draw [shift={(360,416)}, rotate = 0] [color={rgb, 255:red, 0; green, 0; blue, 0 }  ][line width=0.75]    (0,5.59) -- (0,-5.59)(10.93,-3.29) .. controls (6.95,-1.4) and (3.31,-0.3) .. (0,0) .. controls (3.31,0.3) and (6.95,1.4) .. (10.93,3.29)   ;
\draw    (400,416) -- (424,416) ;
\draw [shift={(424,416)}, rotate = 180] [color={rgb, 255:red, 0; green, 0; blue, 0 }  ][line width=0.75]    (0,5.59) -- (0,-5.59)(10.93,-3.29) .. controls (6.95,-1.4) and (3.31,-0.3) .. (0,0) .. controls (3.31,0.3) and (6.95,1.4) .. (10.93,3.29)   ;
\draw    (184,432) -- (232,432) ;
\draw    (360,432) -- (424,432) ;
\draw    (272,432) -- (320,432) ;
\draw [color={rgb, 255:red, 74; green, 144; blue, 226 }  ,draw opacity=0.5 ][line width=1.5]    (160,432) -- (96,431.82) ;
\draw [color={rgb, 255:red, 80; green, 227; blue, 194 }  ,draw opacity=0.5 ][line width=1.5]    (384,432) -- (488,432.18) ;

\draw (84,426.4) node [anchor=north west][inner sep=0.75pt]  [font=\footnotesize]  {$s$};
\draw (237,354.4) node [anchor=north west][inner sep=0.75pt]  [font=\footnotesize]  {$\pi _{H'-D}\left( d_{1}^{-} ,d_{2}^{+}\right)$};
\draw (492,426.4) node [anchor=north west][inner sep=0.75pt]  [font=\footnotesize]  {$t$};
\draw (240,506.4) node [anchor=north west][inner sep=0.75pt]  [font=\footnotesize]  {$d_{1}^{-}$};
\draw (265,506.4) node [anchor=north west][inner sep=0.75pt]  [font=\footnotesize]  {$d_{1}^{+}$};
\draw (256,435.4) node [anchor=north west][inner sep=0.75pt]  [font=\footnotesize]  {$d_{1}$};
\draw (156,435.4) node [anchor=north west][inner sep=0.75pt]  [font=\footnotesize]  {$p$};
\draw (380,434.4) node [anchor=north west][inner sep=0.75pt]  [font=\footnotesize]  {$q$};
\draw (329,434.4) node [anchor=north west][inner sep=0.75pt]  [font=\footnotesize]  {$d_{2}$};
\draw (312,506.58) node [anchor=north west][inner sep=0.75pt]  [font=\footnotesize]  {$d_{2}^{-}$};
\draw (337,506.58) node [anchor=north west][inner sep=0.75pt]  [font=\footnotesize]  {$d_{2}^{+}$};
\draw (292,410.4) node [anchor=north west][inner sep=0.75pt]  [font=\footnotesize]  {$R$};
\draw (201,410.4) node [anchor=north west][inner sep=0.75pt]  [font=\footnotesize]  {$E_{1}$};
\draw (385,410.4) node [anchor=north west][inner sep=0.75pt]  [font=\footnotesize]  {$E_{2}$};

\end{tikzpicture}
\end{center}

\begin{proof}
    The proof is similar to Lemma~\ref{lemma:1edge}. Consider the path $\og{s}{t}{(\set{d_1, d_2}\cup D)}$. Suppose on the original shortest path $st$, $\og{s}{t}{(\set{d_1, d_2}\cup D)}$ diverges at point $p$ before $d_1$ and converges at point $q$ after $d_2$. In $H \cup R$ (and so in $H'$), we can use the edge $(d_1^-,p)$ to replace the part $[s,p]$ and the edge $(q,d_2^+)$ to replace the part $[q,t]$. Since $w(d_1^-,p) = N + |sp|$ and $w(q,d_1^+) = N + |qt|$, we know the path from $d_1^-$ to $d_2^+$ in $H'-D$ is with weight $|\og{s}{t}{(\set{d_1, d_2}\cup D)}|+2N$. Therefore, 
    \[|\pi_{(H'-D)}(d_1^-, d_2^+)| \leq |\og{s}{t}{(\set{d_1, d_2}\cup D)}|+2N.\]

    Conversely, take the shortest path from $d_1^-$ to $d_2^+$ avoiding $d_1,d_2$ and $D$ in $H'-D$, we can use the original shortest path $sx$ to replace the first edge $(d_1^-,x)$ and $yt$ to replace the last edge $(y,d_1^+)$. The path we constructed avoids $d_1$, $d_2$, and $D$, and is fully in $G$ because $N$ is large enough. Since $w(d_1^-,p) = N + |sp|$ and $w(q,d_1^+) = N + |qt|$, we know the new path we constructed is with weight $|\pi_{H'-D}(d_1^-, d_2^+)|-2N$. Therefore, 
    \[|\pi_{(H'-D)}(d_1^-, d_2^+)| \geq |\og{s}{t}{(\set{d_1, d_2}\cup D)}|+2N.\]

    So the two values are equal.
\end{proof}


\subsection{The Binary Partition Structure}
\label{sec:binary}

We define a binary partition structure similar to $\cite{duan2009dual}$ as follows. Recall that we assume there are $2^k$ edges on $st$, so $t=s \oplus 2^k$. Consider a set $\set{m_{i,j}|0 \leq i \leq k, 0 \leq j \leq 2^i}$ to be:

\begin{itemize}
    \item $m_{0,0} = s, m_{0,1}=t$.
    \item For any $1 \leq i \leq k$, $0 \leq j \leq 2^i$, if $j$ is even, we define $m_{i,j}=m_{i-1,j/2}$; if $j$ is odd, we define $m_{i,j}=m_{i-1,(j-1)/2} \oplus 2^{k-i}$.
\end{itemize}

Each vertex may be represented by multiple elements in $\set{m_{i,j}}$. Moreover, we define intervals $Q_{i,j} = [m_{i,j},m_{i,j+1}]$ for all $0 \leq i \leq k$ and $0 \leq j < 2^i$. Thus, in level $i$, we will have $2^i$ edge-disjoint ranges
$Q_{i,0}, Q_{i,1}, ..., Q_{i,2^i-1}$, each with $2^{k-i}$ edges, and their union is the whole original shortest path $st$.

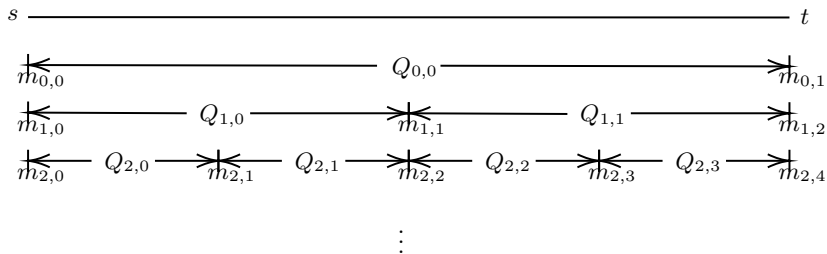
\begin{center}
    \tikzset{every picture/.style={line width=0.75pt}} 

\begin{tikzpicture}[x=0.75pt,y=0.75pt,yscale=-1,xscale=1]

\draw    (96,423.47) -- (480,423.65) ;
\draw    (96,447.65) -- (272,448) ;
\draw [shift={(96,447.65)}, rotate = 0.12] [color={rgb, 255:red, 0; green, 0; blue, 0 }  ][line width=0.75]    (0,5.59) -- (0,-5.59)(10.93,-3.29) .. controls (6.95,-1.4) and (3.31,-0.3) .. (0,0) .. controls (3.31,0.3) and (6.95,1.4) .. (10.93,3.29)   ;
\draw    (304,448) -- (480,448.35) ;
\draw [shift={(480,448.35)}, rotate = 180.12] [color={rgb, 255:red, 0; green, 0; blue, 0 }  ][line width=0.75]    (0,5.59) -- (0,-5.59)(10.93,-3.29) .. controls (6.95,-1.4) and (3.31,-0.3) .. (0,0) .. controls (3.31,0.3) and (6.95,1.4) .. (10.93,3.29)   ;
\draw    (96,471.65) -- (176,472) ;
\draw [shift={(96,471.65)}, rotate = 0.25] [color={rgb, 255:red, 0; green, 0; blue, 0 }  ][line width=0.75]    (0,5.59) -- (0,-5.59)(10.93,-3.29) .. controls (6.95,-1.4) and (3.31,-0.3) .. (0,0) .. controls (3.31,0.3) and (6.95,1.4) .. (10.93,3.29)   ;
\draw    (208,472) -- (288,472) ;
\draw [shift={(288,472)}, rotate = 180] [color={rgb, 255:red, 0; green, 0; blue, 0 }  ][line width=0.75]    (0,5.59) -- (0,-5.59)(10.93,-3.29) .. controls (6.95,-1.4) and (3.31,-0.3) .. (0,0) .. controls (3.31,0.3) and (6.95,1.4) .. (10.93,3.29)   ;
\draw    (288,472) -- (368,472.35) ;
\draw [shift={(288,472)}, rotate = 0.25] [color={rgb, 255:red, 0; green, 0; blue, 0 }  ][line width=0.75]    (0,5.59) -- (0,-5.59)(10.93,-3.29) .. controls (6.95,-1.4) and (3.31,-0.3) .. (0,0) .. controls (3.31,0.3) and (6.95,1.4) .. (10.93,3.29)   ;
\draw    (400,472) -- (480,472) ;
\draw [shift={(480,472)}, rotate = 180] [color={rgb, 255:red, 0; green, 0; blue, 0 }  ][line width=0.75]    (0,5.59) -- (0,-5.59)(10.93,-3.29) .. controls (6.95,-1.4) and (3.31,-0.3) .. (0,0) .. controls (3.31,0.3) and (6.95,1.4) .. (10.93,3.29)   ;
\draw    (96,496) -- (128,496) ;
\draw [shift={(96,496)}, rotate = 0] [color={rgb, 255:red, 0; green, 0; blue, 0 }  ][line width=0.75]    (0,5.59) -- (0,-5.59)(10.93,-3.29) .. controls (6.95,-1.4) and (3.31,-0.3) .. (0,0) .. controls (3.31,0.3) and (6.95,1.4) .. (10.93,3.29)   ;
\draw    (160,496) -- (192,496) ;
\draw [shift={(192,496)}, rotate = 180] [color={rgb, 255:red, 0; green, 0; blue, 0 }  ][line width=0.75]    (0,5.59) -- (0,-5.59)(10.93,-3.29) .. controls (6.95,-1.4) and (3.31,-0.3) .. (0,0) .. controls (3.31,0.3) and (6.95,1.4) .. (10.93,3.29)   ;
\draw    (192,496) -- (224,496) ;
\draw [shift={(192,496)}, rotate = 0] [color={rgb, 255:red, 0; green, 0; blue, 0 }  ][line width=0.75]    (0,5.59) -- (0,-5.59)(10.93,-3.29) .. controls (6.95,-1.4) and (3.31,-0.3) .. (0,0) .. controls (3.31,0.3) and (6.95,1.4) .. (10.93,3.29)   ;
\draw    (256,496) -- (288,496) ;
\draw [shift={(288,496)}, rotate = 180] [color={rgb, 255:red, 0; green, 0; blue, 0 }  ][line width=0.75]    (0,5.59) -- (0,-5.59)(10.93,-3.29) .. controls (6.95,-1.4) and (3.31,-0.3) .. (0,0) .. controls (3.31,0.3) and (6.95,1.4) .. (10.93,3.29)   ;
\draw    (288,496) -- (320,496) ;
\draw [shift={(288,496)}, rotate = 0] [color={rgb, 255:red, 0; green, 0; blue, 0 }  ][line width=0.75]    (0,5.59) -- (0,-5.59)(10.93,-3.29) .. controls (6.95,-1.4) and (3.31,-0.3) .. (0,0) .. controls (3.31,0.3) and (6.95,1.4) .. (10.93,3.29)   ;
\draw    (352,496) -- (384,496) ;
\draw [shift={(384,496)}, rotate = 180] [color={rgb, 255:red, 0; green, 0; blue, 0 }  ][line width=0.75]    (0,5.59) -- (0,-5.59)(10.93,-3.29) .. controls (6.95,-1.4) and (3.31,-0.3) .. (0,0) .. controls (3.31,0.3) and (6.95,1.4) .. (10.93,3.29)   ;
\draw    (384,496) -- (416,496) ;
\draw [shift={(384,496)}, rotate = 0] [color={rgb, 255:red, 0; green, 0; blue, 0 }  ][line width=0.75]    (0,5.59) -- (0,-5.59)(10.93,-3.29) .. controls (6.95,-1.4) and (3.31,-0.3) .. (0,0) .. controls (3.31,0.3) and (6.95,1.4) .. (10.93,3.29)   ;
\draw    (448,496) -- (480,496) ;
\draw [shift={(480,496)}, rotate = 180] [color={rgb, 255:red, 0; green, 0; blue, 0 }  ][line width=0.75]    (0,5.59) -- (0,-5.59)(10.93,-3.29) .. controls (6.95,-1.4) and (3.31,-0.3) .. (0,0) .. controls (3.31,0.3) and (6.95,1.4) .. (10.93,3.29)   ;

\draw (84,418.4) node [anchor=north west][inner sep=0.75pt]  [font=\footnotesize]  {$s$};
\draw (484,418.05) node [anchor=north west][inner sep=0.75pt]  [font=\footnotesize]  {$t$};
\draw (89,450.05) node [anchor=north west][inner sep=0.75pt]  [font=\footnotesize]  {$m_{0,0}$};
\draw (473,450.05) node [anchor=north west][inner sep=0.75pt]  [font=\footnotesize]  {$m_{0,1}$};
\draw (89,474.4) node [anchor=north west][inner sep=0.75pt]  [font=\footnotesize]  {$m_{1,0}$};
\draw (473,474.4) node [anchor=north west][inner sep=0.75pt]  [font=\footnotesize]  {$m_{1,2}$};
\draw (281,474.4) node [anchor=north west][inner sep=0.75pt]  [font=\footnotesize]  {$m_{1,1}$};
\draw (89,498.4) node [anchor=north west][inner sep=0.75pt]  [font=\footnotesize]  {$m_{2,0}$};
\draw (473,498.4) node [anchor=north west][inner sep=0.75pt]  [font=\footnotesize]  {$m_{2,4}$};
\draw (281,498.4) node [anchor=north west][inner sep=0.75pt]  [font=\footnotesize]  {$m_{2,2}$};
\draw (185,498.4) node [anchor=north west][inner sep=0.75pt]  [font=\footnotesize]  {$m_{2,1}$};
\draw (377,498.4) node [anchor=north west][inner sep=0.75pt]  [font=\footnotesize]  {$m_{2,3}$};
\draw (278,442.4) node [anchor=north west][inner sep=0.75pt]  [font=\footnotesize]  {$Q_{0,0}$};
\draw (181,466.4) node [anchor=north west][inner sep=0.75pt]  [font=\footnotesize]  {$Q_{1,0}$};
\draw (373,466.4) node [anchor=north west][inner sep=0.75pt]  [font=\footnotesize]  {$Q_{1,1}$};
\draw (133,490.4) node [anchor=north west][inner sep=0.75pt]  [font=\footnotesize]  {$Q_{2,0}$};
\draw (229,490.4) node [anchor=north west][inner sep=0.75pt]  [font=\footnotesize]  {$Q_{2,1}$};
\draw (325,490.4) node [anchor=north west][inner sep=0.75pt]  [font=\footnotesize]  {$Q_{2,2}$};
\draw (421,490.4) node [anchor=north west][inner sep=0.75pt]  [font=\footnotesize]  {$Q_{2,3}$};
\draw (281,522.4) node [anchor=north west][inner sep=0.75pt]  [font=\footnotesize]  {$\vdots $};

\end{tikzpicture}
\end{center}

We further define $2k$ graphs as follows. For all $1 \leq i \leq k$, let $H_{i,0}=H\cup_{\text{even }j} Q_{i,j}$ and let $H_{i,1}=H \cup_{\text{odd }j} Q_{i,j}$. These graphs are obtained from $H$, with $H_{i,0}$ adding the ranges with even orders in level $i$, and $H_{i,1}$ adding the ranges with odd orders. (Note that in $H$ the original shortest path $st$ is removed, and we now add half of the ranges back to get $H_{i,0}$ and $H_{i,1}$.) A demonstration of the graphs $H_{i,0}$ and $H_{i,1}$ is as below. 

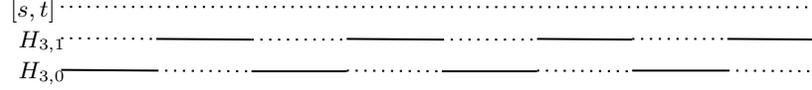
\begin{center}
    \tikzset{every picture/.style={line width=0.75pt}} 

\begin{tikzpicture}[x=0.75pt,y=0.75pt,yscale=-1,xscale=1]

\draw    (96,623.65) -- (144,624) ;
\draw    (192,624.35) -- (240,624.35) ;
\draw  [dash pattern={on 0.84pt off 2.51pt}]  (144,624) -- (192,624.35) ;
\draw  [dash pattern={on 0.84pt off 2.51pt}]  (240,624) -- (288,624.35) ;
\draw  [dash pattern={on 0.84pt off 2.51pt}]  (336,624) -- (384,624.35) ;
\draw  [dash pattern={on 0.84pt off 2.51pt}]  (432,624) -- (480,624.35) ;
\draw    (288,624.35) -- (336,624.35) ;
\draw    (384,624) -- (432,624) ;
\draw  [dash pattern={on 0.84pt off 2.51pt}]  (96,607.65) -- (144,608) ;
\draw  [dash pattern={on 0.84pt off 2.51pt}]  (192,608.35) -- (240,608.35) ;
\draw    (144,608) -- (192,608.35) ;
\draw    (240,608) -- (288,608.35) ;
\draw    (336,608) -- (384,608.35) ;
\draw    (432,608) -- (480,608.35) ;
\draw  [dash pattern={on 0.84pt off 2.51pt}]  (288,608.35) -- (336,608.35) ;
\draw  [dash pattern={on 0.84pt off 2.51pt}]  (384,608) -- (432,608) ;
\draw  [dash pattern={on 0.84pt off 2.51pt}]  (96,591.82) -- (480,592) ;

\draw (73,618.4) node [anchor=north west][inner sep=0.75pt]  [font=\footnotesize]  {$H_{3,0}$};
\draw (73,602.4) node [anchor=north west][inner sep=0.75pt]  [font=\footnotesize]  {$H_{3,1}$};
\draw (69,586.4) node [anchor=north west][inner sep=0.75pt]  [font=\footnotesize]  {$[ s,t]$};

\end{tikzpicture}
\end{center}


Consider two failures $d_1=(s \oplus (l-1), s \oplus l), d_2=(s \oplus r, s \oplus (r+1))$ on $st$, with $0 < l \leq r < 2^k$. In the partition structure, we can find a vertex $m_{i,j}$ (where $i$ is minimized) separating them, so $d_1$ is on $Q_{i,j-1}$ and $d_2$ is on $Q_{i,j}$. We know $j$ should be odd since otherwise the vertex $m_{i,j}$ will also appear as $m_{i-1,j/2}$. Let the whole interval between $d_1$ and $d_2$ be $R=[s \oplus l, s \oplus r]$, the left part of $R$ be $R_l=[s \oplus l, m_{i,j}]$ and the right part be $R_r=[m_{i,j},s \oplus r]$. We can see that $H_{i,0}$ contains $R_l$ but not $R_r$, and $H_{i,1}$ contains $R_r$ but not $R_l$.

\subsection{The Algorithm}

Let's consider what we need to compute $|\og{s}{t}{\set{d_1,d_2,d_3}}|$. Here, we claim the following theorem for this case with two edge failures on $st$:

\begin{theorem}\label{theorem:2edge}
The length of the replacement path $|\og{s}{t}{\set{d_1, d_2, d_3}}|$ equals the minimum one in the following four values:
\begin{itemize}
    \item $|\hg{d_1^-}{d_2^+}{d_3}|-2N$ (as in Lemma \ref{2edgecase1});
    \item $|\pi_{H_{i,0}-d_1-d_3}(d_1^-,d_2^+)|-2N$ (as in Lemma \ref{2edgecase2});
    \item $|\pi_{H_{i,1}-d_2-d_3}(d_1^-,d_2^+)|-2N$ (as in Lemma \ref{2edgecase3});
    \item $\min\{|\pi_{H_{i,0}-d_1-d_3}(d_1^-,m_{i,j})|,|\pi_{H_{i,1}-d_2-d_3}(d_1^-,m_{i,j})|\}$
        
         $+ \min\{|\pi_{H_{i,0}-d_1-d_3}(m_{i,j},d_2^+)|,|\pi_{H_{i,1}-d_2-d_3}(m_{i,j},d_2^+)|\}-2N.$ (as in Lemma \ref{lemma:2edgescase4-2})
\end{itemize}

\end{theorem}

The proof will be shown in Section \ref{section:2edge}. 

Now let's consider what do we need to compute these values. For $|\hg{d_1^-}{d_2^+}{d_3}|-2N$, we can obtain it by one DSO query with failure $d_3$ in $H$. For the other values, we notice that all these shortest path values are either shortest paths in $\pi_{H_{i,0}-d_1-d_3}$, or shortest paths in $\pi_{H_{i,1}-d_2-d_3}$. These can be obtained by one DSO query with failure $d_3$ in either $H_{i,0}-d_1$ or $H_{i,1}-d_2$.

Now note that there are only $O(\log n)$ graphs in $\set{H_{i,j}}$. By the offline dynamic DSO from Section~\ref{sec-inc} and Theorem~\ref{thm:offline}, we can compute all DSOs for graphs like $\set{H_{i,j}}-d$ for every $d \in st$ by deleting and adding back $d$ one by one. There are $O(n)$ operations done for each graph $\set{H_{i,j}}$, therefore, the total update time is $\Tilde{O}(n^3)$.

When we obtained a DSO for $\set{H_{i,j}}-d$ at some time in the offline dynamic oracle, we can operate all queries related to it. There are $O(n^3)$ queries in total (constant number of queries need to be made for any set $\set{d_1,d_2,d_3}$), and each query can be done in time $\Tilde{O}(1)$. Therefore, the total query time is also in $\Tilde{O}(n^3)$. In conclusion, we can run the algorithm for all sets $\set{d_1,d_2,d_3}$ satisfying only $d_1,d_2$ are on $st$, in $\Tilde{O}(n^3)$ time.

\subsection{Resolving Possible Path Cases}\label{section:2edge}

We observe that the replacement shortest path $\og{s}{t}{\set{d_1, d_2, d_3}}$ has the following structure:

\begin{itemize}
    \item It starts from $s$, goes along $st$, and diverges before $d_1$. 
    \item It may converges and diverges in $st$ between $d_1$ and $d_2$ (possibly in either direction). 
    \item It converges in $st$ after $d_2$, goes along $st$, and ends at $t$.
\end{itemize}

Now consider the part $R=[s \oplus l, s \oplus r]$ between $d_1$ and $d_2$ on $st$. For the replacement shortest path $\og{s}{t}{\set{d_1, d_2, d_3}}$, regarding the vertex $m_{i,j}$, there are four cases we need to consider:

\begin{itemize}
    \item There is no intersection with $R=[s \oplus l, s \oplus r]$.
    \item The edge intersection is fully in $R_l=[s \oplus l, m_{i,j}]$.
    \item The edge intersection is fully in $R_r=[m_{i,j},s \oplus r]$.
    \item The edge intersection is an interval containing the vertex $m_{i,j}$.
\end{itemize}

First, consider the case that the replacement path has no intersection with $R=[s \oplus l, s \oplus r]$. Here, we can assume that the whole interval $R$ is also avoided and we consider simulating it using the graph $H$ (with everything in $R$ deleted). We get the following lemma:

\begin{lemma}\label{2edgecase1}
    Starting from the graph $H$, we have:
    \[|\hg{d_1^-}{d_2^+}{d_3}| = |\og{s}{t}{(\set{d_1, d_2, d_3} \cup R)}| + 2N \geq |\og{s}{t}{\set{d_1, d_2, d_3}}| + 2N.\]
    Moreover, if $\og{s}{t}{\set{d_1, d_2, d_3}}$ has no edge intersection with $R=[s \oplus l, s \oplus r]$, then:
    \[|\hg{d_1^-}{d_2^+}{d_3}| = |\og{s}{t}{\set{d_1, d_2, d_3}}| + 2N.\]
\end{lemma}

\begin{proof}
    Consider Lemma \ref{lemma:2edge}. Here, we let $D = \set{d_3} \cup R$ and $E_1=E_2=\emptyset$ and we get $H'-D = (H \cup R \cup E_1 \cup E_2)-D = H-d_3$ and $G-(\set{d_1, d_2}\cup D) = G- (\set{d_1, d_2, d_3} \cup R)$. Therefore, $|\hg{d_1^-}{d_2^+}{d_3}| = |\og{s}{t}{(\set{d_1, d_2, d_3} \cup R)}| + 2N$ holds. Moreover, since $|\og{s}{t}{\set{d_1, d_2, d_3}}| \leq |\og{s}{t}{(\set{d_1, d_2, d_3} \cup R)}|$, and when there is no edge intersection for $\og{s}{t}{\set{d_1, d_2, d_3}}$ and $R$, we can get $\og{s}{t}{\set{d_1, d_2, d_3}} = \og{s}{t}{(\set{d_1, d_2, d_3} \cup R)}$. Therefore, the lemma holds.
\end{proof}

Now suppose the replacement path has an intersection with $R$ that is only on a subinterval of $R_l$. Here we can assume that the interval $R_r$ is also avoided and we consider simulating it using the graph $H_{i,0}$ (with everything in $R_r$ deleted but everything in $R_l$ still included). We get the following lemma:

\begin{lemma}\label{2edgecase2}
    Starting from the graph $H_{i,0} - d_1$, we have:
    \[|\pi_{H_{i,0}-d_1-d_3}(d_1^-,d_2^+)| = |\og{s}{t}{(\set{d_1, d_2, d_3}\cup R_r)}| + 2N \geq |\og{s}{t}{\set{d_1, d_2, d_3}}| + 2N.\]
    Moreover, if $\og{s}{t}{\set{d_1, d_2, d_3}}$ intersects $R = [s \oplus l, s \oplus r]$ only on $R_l = [s \oplus l, m_{i,j}]$, then we have:
    \[|\pi_{H_{i,0}-d_1-d_3}(d_1^-,d_2^+)| = |\og{s}{t}{\set{d_1, d_2, d_3}}| + 2N.\]
    
    (Note that $H_{i,0}$ includes $\set{d_1} \cup R_l$ and does not include $R_r \cup \set{d_2}$.)
\end{lemma}

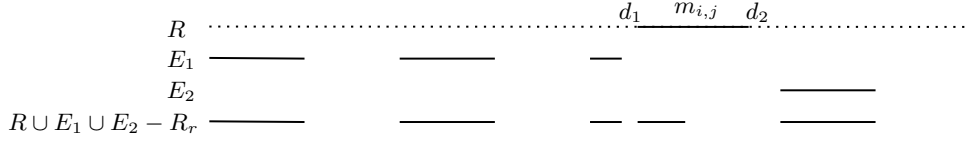
\begin{center}
    \tikzset{every picture/.style={line width=0.75pt}} 

\begin{tikzpicture}[x=0.75pt,y=0.75pt,yscale=-1,xscale=1]

\draw  [dash pattern={on 0.84pt off 2.51pt}]  (176,127.82) -- (560,128) ;
\draw    (392,128) -- (448,128) ;
\draw    (176,143.65) -- (224,144) ;
\draw    (272,144) -- (320,144) ;
\draw    (368,144) -- (384,144) ;
\draw    (464,160) -- (512,160) ;
\draw    (392,176) -- (416,176) ;
\draw    (176,175.65) -- (224,176) ;
\draw    (272,176) -- (320,176) ;
\draw    (368,176) -- (384,176) ;
\draw    (464,176) -- (512,176) ;

\draw (381,114.4) node [anchor=north west][inner sep=0.75pt]  [font=\footnotesize]  {$d_{1}$};
\draw (445,114.4) node [anchor=north west][inner sep=0.75pt]  [font=\footnotesize]  {$d_{2}$};
\draw (153,123.4) node [anchor=north west][inner sep=0.75pt]  [font=\footnotesize]  {$R$};
\draw (153,138.4) node [anchor=north west][inner sep=0.75pt]  [font=\footnotesize]  {$E_{1}$};
\draw (153,154.4) node [anchor=north west][inner sep=0.75pt]  [font=\footnotesize]  {$E_{2}$};
\draw (409,114.4) node [anchor=north west][inner sep=0.75pt]  [font=\footnotesize]  {$m_{i,j}$};
\draw (73,170.4) node [anchor=north west][inner sep=0.75pt]  [font=\footnotesize]  {$R\cup E_{1} \cup E_{2} -R_{r}$};

\end{tikzpicture}
\end{center}

\begin{proof}
     Consider Lemma \ref{lemma:2edge}. Here, we let $D = \set{d_3} \cup R_r$ and $E_1=H_{i,0} \cap [s,s\oplus(l-1)]$, $E_2=H_{i,0} \cap [s\oplus(r+1),t]$ in Lemma \ref{lemma:2edge}, we can get $H'-D = (H \cup R \cup E_1 \cup E_2)-D = H_{i,0}-d_1-d_3$ and $G-(\set{d_1, d_2}\cup D) = G- (\set{d_1, d_2, d_3} \cup R_r)$. Therefore, $|\pi_{H_{i,0}-d_1-d_3}(d_1^-,d_2^+)| = |\og{s}{t}{(\set{d_1, d_2, d_3}\cup R_r)}| + 2N$.
     
     Moreover, we can observe that $|\og{s}{t}{\set{d_1, d_2, d_3}}| \leq |\og{s}{t}{(\set{d_1, d_2, d_3} \cup R_r)}|$, and when there is no edge intersection for $\og{s}{t}{\set{d_1, d_2, d_3}}$ and $R_r$, we can get $\og{s}{t}{\set{d_1, d_2, d_3}} = \og{s}{t}{(\set{d_1, d_2, d_3} \cup R_r)}$, so the lemma holds.
\end{proof}

Now suppose the replacement path has an intersection with $R$ that is only on a subinterval of $R_r$. This is exactly the symmetric case and we demonstrate it in the following lemma:

\begin{lemma}\label{2edgecase3}
    Starting from the graph $H_{i,1} - d_2$, we have:
    \[|\pi_{H_{i,1}-d_2-d_3}(d_1^-,d_2^+)| = |\og{s}{t}{(\set{d_1, d_2, d_3}\cup R_l)}| + 2N \geq |\og{s}{t}{\set{d_1, d_2, d_3}}| + 2N.\]
    If $\og{s}{t}{\set{d_1, d_2, d_3}}$ intersects $R = [s \oplus l, s \oplus r]$ only on $R_r = [m_{i,j},s \oplus r]$, then we have:
    \[|\pi_{H_{i,1}-d_2-d_3}(d_1^-,d_2^+)| = |\og{s}{t}{\set{d_1, d_2, d_3}}| + 2N.\]
    
    (Note that $H_{i,1}$ includes $R_r \cup \set{d_2}$ and does not include $\set{d_1} \cup R_l$.)
\end{lemma}

\begin{proof}
    The lemma can be proved symmetrically. Here, we let $D = \set{d_3} \cup R_l$ and $E_1=H_{i,1} \cap [s,s\oplus(l-1)]$, $E_2=H_{i,1} \cap [s\oplus(r+1),t]$. In Lemma \ref{lemma:2edge}, we can get $H'-D = (H \cup R \cup E_1 \cup E_2)-D = H_{i,1}-d_2-d_3$ and $G-(\set{d_1, d_2}\cup D) = G- (\set{d_1, d_2, d_3} \cup R_l)$. Therefore, $|\pi_{H_{i,1}-d_2-d_3}(d_1^-,d_2^+)| = |\og{s}{t}{(\set{d_1, d_2, d_3}\cup R_l)}| + 2N$.
    
    Moreover, we can observe that $|\og{s}{t}{\set{d_1, d_2, d_3}}| \leq |\og{s}{t}{(\set{d_1, d_2, d_3} \cup R_l)}|$, and when there is no edge intersection for $\og{s}{t}{\set{d_1, d_2, d_3}}$ and $R_l$, we can get $\og{s}{t}{\set{d_1, d_2, d_3}} = \og{s}{t}{(\set{d_1, d_2, d_3} \cup R_l)}$, so the lemma holds.
\end{proof}

Finally we consider the case that $\og{s}{t}{\set{d_1, d_2, d_3}}$ goes through the vertex $m_{i,j}$. Here, we consider breaking the whole replacement shortest path into two parts: One from $s$ to $m_{i,j}$, and one from $m_{i,j}$ to $t$, in the graph $G-\set{d_1, d_2, d_3}$. We can immediately obtain the following lemma:

\begin{lemma}\label{2edgecase4}
    We have:
    \[|\og{s}{t}{\set{d_1, d_2, d_3}}| \leq |\og{s}{m_{i,j}}{\set{d_1, d_2, d_3}}| + |\og{m_{i,j}}{t}{\set{d_1, d_2, d_3}}|.\]
    If $\og{s}{t}{\set{d_1, d_2, d_3}}$ contains the vertex $m_{i,j}$, then: 
    \[|\og{s}{t}{\set{d_1, d_2, d_3}}| = |\og{s}{m_{i,j}}{\set{d_1, d_2, d_3}}| + |\og{m_{i,j}}{t}{\set{d_1, d_2, d_3}}|.\]
\end{lemma}

\begin{proof}
    The inequality is just by triangle inequality in $G-\set{d_1, d_2, d_3}$. And if $\og{s}{t}{\set{d_1, d_2, d_3}}$ contains the vertex $m_{i,j}$, we can break the path in two parts: one from $s$ to $m_{i,j}$, and another one from $m_{i,j}$ to $t$, and the lemma follows.
\end{proof}

Finally, we consider how to compute these two parts when $\og{s}{t}{\set{d_1, d_2, d_3}}$ contains the vertex $m_{i,j}$ as in Lemma~\ref{2edgecase4}. Here, we will show in the following lemmas the approach to calculate $\og{s}{m_{i,j}}{\set{d_1, d_2, d_3}}$, and after that, we can calculate $\og{m_{i,j}}{t}{\set{d_1, d_2, d_3}}$ symmetrically.

Consider the path $\og{s}{m_{i,j}}{\set{d_1, d_2, d_3}}$. We observe that the path cannot intersect the interval $R$ on edges both in $R_l$ and $R_r$. If so, then there must be a shorter path from $s$ to $m_{i,j}$ by taking the shortest path after it first intersects $R$. 

Moreover, if $|\og{s}{t}{\set{d_1, d_2, d_3}}|$ goes through $m_{i,j}$, then we show that the first part $\og{s}{m_{i,j}}{\set{d_1, d_2, d_3}}$ cannot go through any vertex in $[s \oplus (r+1),t]$ as well. Otherwise, if it goes through some vertex $q \in [s \oplus (r+1),t]$ before $m_{i,j}$, however, $\og{s}{q}{\set{d_1, d_2, d_3}} \circ qt$ will be a shorter path that does not go through $m_{i,j}$, which comes to contradiction.

In the following lemmas, we make use of the graphs $H_{i,0}-d_1$ and $H_{i,1}-d_2$ to capture both cases that $\og{s}{m_{i,j}}{\set{d_1, d_2, d_3}}$ intersects $R$ only on $R_l$ or only on $R_r$. We obtain the following lemmas:

\begin{lemma}\label{lemma:smcase1}
    Consider the graph $H_{i,0}-d_1$, we have:
    \[|\pi_{H_{i,0}-d_1-d_3}(d_1^-,m_{i,j})| \geq |\og{s}{m_{i,j}}{(\set{d_1, d_2, d_3} \cup R_r)}| + N \geq |\og{s}{m_{i,j}}{\set{d_1, d_2, d_3}}| + N.\]
    Moreover, if $\og{s}{m_{i,j}}{\set{d_1, d_2, d_3}}$ does not go through any vertex in $[s \oplus (r+1),t]$ and intersects $R$ only on $R_l$, then:
    \[|\pi_{H_{i,0}-d_1-d_3}(d_1^-,m_{i,j})| = |\og{s}{m_{i,j}}{(\set{d_1, d_2, d_3} \cup R_r)}| + N = |\og{s}{m_{i,j}}{\set{d_1, d_2, d_3}}| + N.\]
\end{lemma}

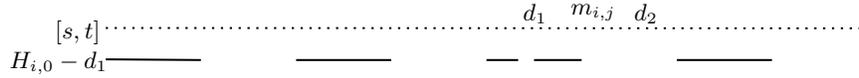
\begin{center}
    \tikzset{every picture/.style={line width=0.75pt}} 

\begin{tikzpicture}[x=0.75pt,y=0.75pt,yscale=-1,xscale=1]

\draw  [dash pattern={on 0.84pt off 2.51pt}]  (96,175.82) -- (480,176) ;
\draw    (312,192) -- (336,192) ;
\draw    (96,191.65) -- (144,192) ;
\draw    (192,192) -- (240,192) ;
\draw    (288,192) -- (304,192) ;
\draw    (384,192) -- (432,192) ;

\draw (69,170.4) node [anchor=north west][inner sep=0.75pt]  [font=\footnotesize]  {$[ s,t]$};
\draw (305,162.4) node [anchor=north west][inner sep=0.75pt]  [font=\footnotesize]  {$d_{1}$};
\draw (361,162.4) node [anchor=north west][inner sep=0.75pt]  [font=\footnotesize]  {$d_{2}$};
\draw (329,162.4) node [anchor=north west][inner sep=0.75pt]  [font=\footnotesize]  {$m_{i,j}$};
\draw (46,186.4) node [anchor=north west][inner sep=0.75pt]  [font=\footnotesize]  {$H_{i,0} -d_{1}$};

\end{tikzpicture}
\end{center}


\begin{proof}
    We consider a proof that is similar to Lemma \ref{lemma:2edge}. Note that the graph $H_{i,0}-d_1$ contains the part $R_l$ but not $R_r$ or $d_2$. First, suppose in $\pi_{H_{i,0}-d_1-d_3}(d_1^-,m_{i,j})$ the first edge is $(d_1^-,p)$, as in Lemma \ref{lemma:2edge}, we can use $[s,p]$ to replace it and get a path with length $|\pi_{H_{i,0}-d_1-d_3}(d_1^-,m_{i,j})| - N$ in $G-(\set{d_1, d_2, d_3} \cup R_r)$. So here we have the inequality $|\pi_{H_{i,0}-d_1-d_3}(d_1^-,m_{i,j})| \geq |\og{s}{m_{i,j}}{(\set{d_1, d_2, d_3} \cup R_r)}| + N$.
    
    Conversely, we consider the shortest replacement path $\og{s}{m_{i,j}}{(\set{d_1, d_2, d_3} \cup R_r)}$, given the condition that it {does not go through $[s \oplus (r+1),t]$ and} intersects $R$ only on $R_l$. Suppose it diverges $st$ at point $p$ before $d_1$, then we can use $(d_1^-,p)$ to replace the part $[s,p]$ to get a path with length $|\og{s}{m_{i,j}}{(\set{d_1, d_2, d_3} \cup R_r)}| + N$ in $H_{i,0}-d_1-d_3$.

    Moreover, if $\og{s}{m_{i,j}}{\set{d_1, d_2, d_3}}$ intersects $R$ only on $R_l$, then we know it also avoids $R_r$ and thus $\og{s}{m_{i,j}}{\set{d_1, d_2, d_3}} = \og{s}{m_{i,j}}{(\set{d_1, d_2, d_3} \cup R_r)}$. Therefore, the equation holds.
\end{proof}

Following we state the symmetric case that $\og{s}{m_{i,j}}{\set{d_1, d_2, d_3}}$ intersects $R$ only on $R_r$. Essentially the same proof works for this lemma.

\begin{lemma}\label{lemma:smcase2}
    Consider the graph $H_{i,1}-d_2$, we have:
    \[|\pi_{H_{i,1}-d_2-d_3}(d_1^-,m_{i,j})| \geq |\og{s}{m_{i,j}}{(\set{d_1, d_2, d_3} \cup R_l)}| + N \geq |\og{s}{m_{i,j}}{\set{d_1, d_2, d_3}}| + N.\]
    Moreover, if $\og{s}{m_{i,j}}{\set{d_1, d_2, d_3}}$ {does not go through any vertex in $[s \oplus (r+1),t]$ and} intersects $R$ only on $R_r$, then:
    \[|\pi_{H_{i,1}-d_2-d_3}(d_1^-,m_{i,j})| = |\og{s}{m_{i,j}}{(\set{d_1, d_2, d_3} \cup R_l)}| + N = |\og{s}{m_{i,j}}{\set{d_1, d_2, d_3}}| + N.\]
\end{lemma}

\begin{proof}
    We still use the same proof that corresponds to the edge $(d_1^-,p)$ and the subpath $[s,p]$, and note that $H_{i,1}$ does not intersect $R_l$ or $d_1$. The same arguments will also work for this lemma here.
\end{proof}

\begin{lemma}\label{lemma:sm}
    We can bound the length of $\og{s}{m_{i,j}}{\set{d_1, d_2, d_3}}$ by:
    \[|\og{s}{m_{i,j}}{\set{d_1, d_2, d_3}}| \leq \min \{ |\pi_{H_{i,0}-d_1-d_3}(d_1^-,m_{i,j})|,|\pi_{H_{i,1}-d_2-d_3}(d_1^-,m_{i,j})| \}-N.\]
    Moreover, if $\og{s}{m_{i,j}}{\set{d_1, d_2, d_3}}$ does not go through any node in $[s \oplus (r+1),t]$, then it will become equation:
    \[|\og{s}{m_{i,j}}{\set{d_1, d_2, d_3}}| = \min \{ |\pi_{H_{i,0}-d_1-d_3}(d_1^-,m_{i,j})|,|\pi_{H_{i,1}-d_2-d_3}(d_1^-,m_{i,j})| \} - N.\]
\end{lemma}

\begin{proof}
    The inequality is immediately obtained by considering both inequalities in Lemma \ref{lemma:smcase1} and Lemma \ref{lemma:smcase2}.
    
    Moreover, if we suppose $\og{s}{m_{i,j}}{\set{d_1, d_2, d_3}}$ does not go through any node in $[s \oplus (r+1),t]$. Now it can intersect $R$ only on one of $R_l$ or $R_r$ since it will take the shortest path using edges in $R$ after first intersecting $R$. Therefore, if it intersects $R$ only on $R_l$, from Lemma \ref{lemma:smcase1} we can obtain that the equation $|\og{s}{m_{i,j}}{\set{d_1, d_2, d_3}}| = |\pi_{H_{i,0}-d_1-d_3}(d_1^-,m_{i,j})|-N$ holds. Otherwise, it must intersect $R$ only on $R_r$, and therefore we know from Lemma \ref{lemma:smcase2} that $|\og{s}{m_{i,j}}{\set{d_1, d_2, d_3}}| = |\pi_{H_{i,1}-d_2-d_3}(d_1^-,m_{i,j})|-N$ holds. These will conclude the result.
\end{proof}

Symmetrically, we can obtain the length of the second part of the shortest path $|\og{m_{i,j}}{t}{\set{d_1, d_2, d_3}}|$. We state the following lemma here and the proof is necessarily the same as Lemma \ref{lemma:sm}.

\begin{lemma}\label{lemma:mt}
    We can bound the length of $\og{m_{i,j}}{t}{\set{d_1, d_2, d_3}}$ by:
    \[|\og{m_{i,j}}{t}{\set{d_1, d_2, d_3}}| {\leq} \min \{ |\pi_{H_{i,0}-d_1-d_3}(m_{i,j},d_2^+)|,|\pi_{H_{i,1}-d_2-d_3}(m_{i,j},d_2^+)| \} -N.\]
    Moreover, if $\og{s}{m_{i,j}}{\set{d_1, d_2, d_3}}$ {does not go through any node in $[s, s \oplus (l-1)]$}, then it will become equation:
    \[|\og{m_{i,j}}{t}{\set{d_1, d_2, d_3}}| = \min \{ |\pi_{H_{i,0}-d_1-d_3}(m_{i,j},d_2^+)|,|\pi_{H_{i,1}-d_2-d_3}(m_{i,j},d_2^+)| \} -N.\]
\end{lemma}

Now consider the case that $\og{s}{t}{\set{d_1, d_2, d_3}}$ goes through $m_{i,j}$. We claim below that we have an equation for the last value as in Lemma \ref{2edgecase4}.

\begin{lemma}\label{lemma:2edgescase4-2}
    If $\og{s}{t}{\set{d_1, d_2, d_3}}$ goes through $m_{i,j}$, then:
    \begin{itemize}
    \item $|\og{s}{m_{i,j}}{\set{d_1, d_2, d_3}}| = \min \{ |\pi_{H_{i,0}-d_1-d_3}(d_1^-,m_{i,j})|,|\pi_{H_{i,1}-d_2-d_3}(d_1^-,m_{i,j})| \} -N$;
    \item $|\og{m_{i,j}}{t}{\set{d_1, d_2, d_3}}| = \min \{ |\pi_{H_{i,0}-d_1-d_3}(m_{i,j},d_2^+)|,|\pi_{H_{i,1}-d_2-d_3}(m_{i,j},d_2^+)| \} -N$.
\end{itemize}
\end{lemma}

\begin{proof}
If $|\og{s}{t}{\set{d_1, d_2, d_3}}|$ goes through $m_{i,j}$, and $\og{s}{m_{i,j}}{\set{d_1, d_2, d_3}}$ intersects $[s\oplus (r+1),t]$ at any vertex $q$, then we consider the path $\og{s}{q}{\set{d_1, d_2, d_3}} \circ qt$, which does not go through $m_{i,j}$ and is shorter than $\og{s}{t}{\set{d_1, d_2, d_3}}$.This is impossible, so $\og{s}{m_{i,j}}{\set{d_1, d_2, d_3}}$ will not intersect $[s\oplus (r+1),t]$ at any node. Similarly, it is also impossible that $\og{s}{m_{i,j}}{\set{d_1, d_2, d_3}}$ intersects $[s, s \oplus (l-1)]$ at any node in this case. Therefore, from both equations in Lemma \ref{lemma:sm} and Lemma \ref{lemma:mt}, we can obtain this lemma.
\end{proof}


Now all cases are evaluated in the lemmas and we begin our final analysis for $\og{s}{t}{\set{d_1, d_2, d_3}}$. We can get it by a minimization among all the four cases above in this section.

\begin{theorem}(cf. Theorem \ref{theorem:2edge})
The length of the replacement path $|\og{s}{t}{\set{d_1, d_2, d_3}}|$ equals the minimum one in the following four values:
\begin{itemize}
    \item $|\hg{d_1^-}{d_2^+}{d_3}|-2N$ (as in Lemma \ref{2edgecase1});
    \item $|\pi_{H_{i,0}-d_1-d_3}(d_1^-,d_2^+)|-2N$ (as in Lemma \ref{2edgecase2});
    \item $|\pi_{H_{i,1}-d_2-d_3}(d_1^-,d_2^+)|-2N$ (as in Lemma \ref{2edgecase3});
    \item $\min \{ |\pi_{H_{i,0}-d_1-d_3}(d_1^-,m_{i,j})|,|\pi_{H_{i,1}-d_2-d_3}(d_1^-,m_{i,j})| \}$ 
    
    $+ \min \{ |\pi_{H_{i,0}-d_1-d_3}(m_{i,j},d_2^+)|,|\pi_{H_{i,1}-d_2-d_3}(m_{i,j},d_2^+)| \} -2N$ (adding two equations in Lemma \ref{lemma:2edgescase4-2}).
\end{itemize}
\end{theorem}

\begin{proof}
    First, we obtain that these four values are all upper bounds of $|\og{s}{t}{\set{d_1, d_2, d_3}}|$. From the inequalities in Lemmas \ref{2edgecase1}, \ref{2edgecase2}, \ref{2edgecase3}, above, the first three values are at least $|\og{s}{t}{\set{d_1, d_2, d_3}}|$. Now for the last value, from Lemma \ref{lemma:sm} and \ref{lemma:mt} we know 
    $$\min \{ |\pi_{H_{i,1}-d_1-d_3}(m_{i,j},d_2^+)|, |\pi_{H_{i,1}-d_2-d_3}(m_{i,j},d_2^+)| \} - N \geq \og{s}{m_{i,j}}{\set{d_1, d_2, d_3}},$$ and $$\min \{ |\pi_{H_{i,0}-d_1-d_3}(m_{i,j},d_2^+)|, |\pi_{H_{i,1}-d_2-d_3}(m_{i,j},d_2^+)| \} - N \geq \og{m_{i,j}}{t}{\set{d_1, d_2, d_3}}.$$ So by adding them, the last value is at least $\og{s}{t}{\set{d_1, d_2, d_3}}$ by triangle inequality in Lemma \ref{2edgecase4}.
    
    Moreover, we consider the intersection of $\og{s}{t}{\set{d_1, d_2, d_3}}$ and $R$. 
    If it is empty, then from Lemma \ref{2edgecase1}, $|\hg{d_1^-}{d_2^+}{d_3}|-2N = |\og{s}{t}{\set{d_1, d_2, d_3}}|$; 
    If it is completely in $R_l$, then from Lemma \ref{2edgecase2}, $|\pi_{H_{i,0}-d_1-d_3}(d_1^-,d_2^+)|-2N = |\og{s}{t}{\set{d_1, d_2, d_3}}|$; 
    If it is completely in $R_r$, then from Lemma \ref{2edgecase3}, $|\pi_{H_{i,1}-d_2-d_3}(d_1^-,d_2^+)|-2N=|\og{s}{t}{\set{d_1, d_2, d_3}}|$; 
    If it contains $m_{i,j}$, then from Lemmas \ref{2edgecase4} and \ref{lemma:2edgescase4-2} the last value will be $|\og{s}{t}{\set{d_1, d_2, d_3}}|$. 
    
    Since one of the above four cases must be true for $\og{s}{t}{\set{d_1, d_2, d_3}}$, we know one upper bound must be exactly $|\og{s}{t}{\set{d_1, d_2, d_3}}|$. Therefore, the minimum of the four values must be $|\og{s}{t}{\set{d_1, d_2, d_3}}|$.
    
\end{proof}

This ends our proof of Theorem \ref{theorem:2edge} and finishes our analysis of the case.



\section{Three Failed Edges on Original Shortest Path}\label{sec4}

In this section, we consider the last case that all three failures $d_1,d_2,d_3$ are on the original shortest path $st$ from left (i.e. close to $s$) to right (i.e. close to $t$), which cut $st$ into four intervals $D_1,D_2,D_3,D_4$, also from left to right. We prove the following theorem:

\begin{theorem}\label{thm:3edges}
There exists an algorithm that, given undirected $G$ with two vertices $s,t \in V(G)$, for all edges $d_1,d_2,d_3 \in st$, answers $|\og{s}{t}{\set{d_1, d_2, d_3}}|$ in $\too{n^3}$ time.
\end{theorem}

In this section we still assume that there are $2^k$ edges on $st$ for some integer $k=O(\log n)$. Let the replacement path be $\pi'=\pi_{G-\{d_1,d_2,d_3\}}(s,t)$. It is straightforward to observe the followings:

\begin{observation}[informal] The replacement path $\pi'$ satisfies:
\begin{itemize}
    \item It diverges at $D_1$ at first, and converges at $D_4$ at last.
    \item It may converges and diverges at $D_2$, $D_3$, or both of them (in any order) in the middle.
    \item Between a convergence point and next divergence point on $st$, it will take the shortest path on $st$.
    \item Between a divergence point and next convergence point on $st$, it will take the shortest path in $G-st$.
\label{obs:3failure}
\end{itemize}    
\end{observation}

With this observation, we can first compute all-pairs shortest paths in $G-st$, which makes it possible to simplify our analysis to purely interval queries for the nodes on the original shortest path $st$. 

Therefore, in this section, we use a \textbf {range tree structure} from Section \ref{sec:binary}, Namely: \begin{itemize}
    \item The ranges contained in the range tree are all $\{Q_{i,j}\}$ in the binary partition structure.
    \item For any range $\{Q_{i,j}\}$ with $i<k$, let its childs to be $Q_{i+1,2j}$ and $Q_{i+1,2j+1}$, which divide it into two subranges.
\end{itemize}

Note that any interval on $st$ can be represented as a union of at most $O(\log n)$ ranges in the range tree. In this section, for any interval $R$ on $st$, we use $R^l$ to denote the leftmost vertex of $R$, and use $R^r$ to denote the rightmost vertex of $R$.

Below, we build an oracle based on the range tree structure. Moreover, we also show how to use these oracles for 3-fault queries.

\subsection{Oracle Build-up}\label{sec:3-fault-oracle}


\begin{theorem}\label{thm:oracle-A}
   For a graph $G$ and vertices $s,t$, we can construct an oracle $A$ in $\too{n^3}$ preprocessing time that answers the following type of query in $\too{1}$ time: given two vertex-disjoint intervals $R_1$ and $R_2$ of consecutive edges on the original shortest path $st$ ($R_1, R_2$ can be single vertices), answer the shortest path that starts from the leftmost (rightmost) vertex of $R_1$, diverges from $R_1$, converges in $R_2$, and finally ends at the leftmost (rightmost) vertex of $R_2$. Denote these by: $A^{l,l}(R_1,R_2) = \min_{x\in R_1, y\in R_2}\{\zdd{R_1}x\circ \pi_{G-st}(x, y) \circ y\zdd{R_2}\}$ and $A^{l,r},A^{r,l},A^{r,r}$ are defined accordingly.   
\end{theorem}

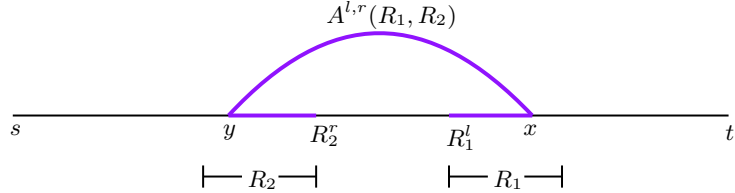
\begin{center}
    \tikzset{every picture/.style={line width=0.75pt}} 

\begin{tikzpicture}[x=0.75pt,y=0.75pt,yscale=-1,xscale=1]

\draw    (40,241.55) -- (401,241.55) ;
\draw [color={rgb, 255:red, 144; green, 19; blue, 254 }  ,draw opacity=0.5 ][line width=1.5]    (260,241.55) -- (302,241.55) ;
\draw    (279.6,272.35) -- (260,272.35) ;
\draw [shift={(260,272.35)}, rotate = 360] [color={rgb, 255:red, 0; green, 0; blue, 0 }  ][line width=0.75]    (0,5.59) -- (0,-5.59)   ;
\draw    (298,272.35) -- (317,272.35) ;
\draw [shift={(317,272.35)}, rotate = 180] [color={rgb, 255:red, 0; green, 0; blue, 0 }  ][line width=0.75]    (0,5.59) -- (0,-5.59)   ;
\draw [color={rgb, 255:red, 144; green, 19; blue, 254 }  ,draw opacity=0.5 ][line width=1.5]    (148.8,241.55) .. controls (200,186.35) and (249.2,185.95) .. (302,241.55) ;
\draw    (155.6,272.35) -- (136,272.35) ;
\draw [shift={(136,272.35)}, rotate = 360] [color={rgb, 255:red, 0; green, 0; blue, 0 }  ][line width=0.75]    (0,5.59) -- (0,-5.59)   ;
\draw    (174,272.35) -- (193,272.35) ;
\draw [shift={(193,272.35)}, rotate = 180] [color={rgb, 255:red, 0; green, 0; blue, 0 }  ][line width=0.75]    (0,5.59) -- (0,-5.59)   ;
\draw [color={rgb, 255:red, 144; green, 19; blue, 254 }  ,draw opacity=0.5 ][line width=1.5]    (193,241.55) -- (148.8,241.55) ;

\draw (398,244.95) node [anchor=north west][inner sep=0.75pt]  [font=\footnotesize]  {$t$};
\draw (37,244.95) node [anchor=north west][inner sep=0.75pt]  [font=\footnotesize]  {$s$};
\draw (257,244.95) node [anchor=north west][inner sep=0.75pt]  [font=\footnotesize]  {$R_{1}^{l}$};
\draw (296,244.95) node [anchor=north west][inner sep=0.75pt]  [font=\footnotesize]  {$x$};
\draw (281,267.75) node [anchor=north west][inner sep=0.75pt]  [font=\footnotesize]  {$R_{1}$};
\draw (144.6,244.95) node [anchor=north west][inner sep=0.75pt]  [font=\footnotesize]  {$y$};
\draw (157,267.75) node [anchor=north west][inner sep=0.75pt]  [font=\footnotesize]  {$R_{2}$};
\draw (188.2,244.95) node [anchor=north west][inner sep=0.75pt]  [font=\footnotesize]  {$R_{2}^{r}$};
\draw (196.67,182.45) node [anchor=north west][inner sep=0.75pt]  [font=\footnotesize]  {$A^{l,r}( R_{1} ,R_{2})$};

\end{tikzpicture}
\end{center}

Note that we can store the path by storing the positions of $x$ and $y$. First, we prove the following lemma.

\begin{lemma}\label{lemma:merging-A}
    If the interval $R_1$ on $st$ is the union of intervals $R_1'$ and $R_1''$ ($\ydd{R_1'}=\zdd{R_1''}$ or there is an edge connecting $\ydd{R_1'}$ and $\zdd{R_1''}$), and we already have $A^{l,l}(R_1',R_2)$ and $A^{l,l}(R_1'',R_2)$, then we can obtain $A^{l,l}(R_1,R_2)$ in $O(1)$ time. Symmetrically, we can get $A^{l,l}(R_1,R_2)$ from $A^{l,l}(R_1,R_2')$ and $A^{l,l}(R_1,R_2'')$ if $R_2$ is the union of $R_2'$ and $R_2''$. This holds for all other $A^{l,r},A^{r,l}$ and $A^{r,r}$ cases as well.
\end{lemma}

\begin{proof}
    W.l.o.g., assume $R_1'$ is on the left of $R_1''$. To get $A^{l,l}(R_1,R_2)= \min_{x\in R_1, y\in R_2}\{\zdd{R_1}x\circ \pi_{G-st}(x, y) \circ y\zdd{R_2}\}$, consider the position of $x$. If $x$ is in $R_1'$, then the path simply avoids $R_1''$, so it equals $A^{l,l}(R_1',R_2)$. If $x$ is in $R_1''$, then the path first goes through $R_1'$ to $\zdd{R_1''}$, then travels through $A^{l,l}(R_1'',R_2)$. Therefore,
    $$A^{l,l}(R_1,R_2)=\min\left\{A^{l,l}(R_1',R_2),\zdd{R_1'}\zdd{R_1''}\circ A^{l,l}(R_1'',R_2)\right\}$$
    Other cases can be proved similarly.
    \begin{center}
        \tikzset{every picture/.style={line width=0.75pt}} 

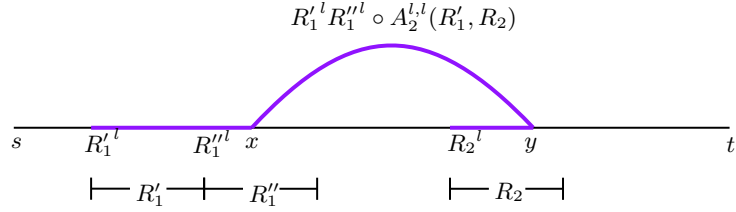
\begin{tikzpicture}[x=0.75pt,y=0.75pt,yscale=-1,xscale=1]

\draw    (40,91.55) -- (401,91.55) ;
\draw [color={rgb, 255:red, 144; green, 19; blue, 254 }  ,draw opacity=0.5 ][line width=1.5]    (260,91.55) -- (302,91.55) ;
\draw    (279.6,122.35) -- (260,122.35) ;
\draw [shift={(260,122.35)}, rotate = 360] [color={rgb, 255:red, 0; green, 0; blue, 0 }  ][line width=0.75]    (0,5.59) -- (0,-5.59)   ;
\draw    (298,122.35) -- (317,122.35) ;
\draw [shift={(317,122.35)}, rotate = 180] [color={rgb, 255:red, 0; green, 0; blue, 0 }  ][line width=0.75]    (0,5.59) -- (0,-5.59)   ;
\draw [color={rgb, 255:red, 144; green, 19; blue, 254 }  ,draw opacity=0.5 ][line width=1.5]    (160,91.6) .. controls (211.2,36.4) and (249.2,35.95) .. (302,91.55) ;
\draw    (155.6,122.35) -- (136,122.35) ;
\draw [shift={(136,122.35)}, rotate = 360] [color={rgb, 255:red, 0; green, 0; blue, 0 }  ][line width=0.75]    (0,5.59) -- (0,-5.59)   ;
\draw    (174,122.35) -- (193,122.35) ;
\draw [shift={(193,122.35)}, rotate = 180] [color={rgb, 255:red, 0; green, 0; blue, 0 }  ][line width=0.75]    (0,5.59) -- (0,-5.59)   ;
\draw [color={rgb, 255:red, 144; green, 19; blue, 254 }  ,draw opacity=0.5 ][line width=1.5]    (160,91.55) -- (79,91.55) ;
\draw    (98.6,122.3) -- (79,122.3) ;
\draw [shift={(79,122.3)}, rotate = 360] [color={rgb, 255:red, 0; green, 0; blue, 0 }  ][line width=0.75]    (0,5.59) -- (0,-5.59)   ;
\draw    (117,122.3) -- (136,122.3) ;
\draw [shift={(136,122.3)}, rotate = 180] [color={rgb, 255:red, 0; green, 0; blue, 0 }  ][line width=0.75]    (0,5.59) -- (0,-5.59)   ;

\draw (398,94.95) node [anchor=north west][inner sep=0.75pt]  [font=\footnotesize]  {$t$};
\draw (37,94.95) node [anchor=north west][inner sep=0.75pt]  [font=\footnotesize]  {$s$};
\draw (257,90.95) node [anchor=north west][inner sep=0.75pt]  [font=\footnotesize]  {$\zdd{R_{2}}$};
\draw (296,94.95) node [anchor=north west][inner sep=0.75pt]  [font=\footnotesize]  {$y$};
\draw (281,117.75) node [anchor=north west][inner sep=0.75pt]  [font=\footnotesize]  {$R_{2}$};
\draw (74,90.28) node [anchor=north west][inner sep=0.75pt]  [font=\footnotesize]  {$\zdd{R'_{1}}$};
\draw (157,117.75) node [anchor=north west][inner sep=0.75pt]  [font=\footnotesize]  {$R''_{1}$};
\draw (155.2,94.95) node [anchor=north west][inner sep=0.75pt]  [font=\footnotesize]  {$x$};
\draw (100,117.7) node [anchor=north west][inner sep=0.75pt]  [font=\footnotesize]  {$R'_{1}$};
\draw (178,26.45) node [anchor=north west][inner sep=0.75pt]  [font=\footnotesize]  {$\zdd{R'_{1}}\zdd{R''_{1}} \circ A_{2}^{l,l}( R'_{1} ,R_{2})$};
\draw (128.67,90.95) node [anchor=north west][inner sep=0.75pt]  [font=\footnotesize]  {$\zdd{R''_{1}}$};

\end{tikzpicture}
    \end{center}
\end{proof}

Thus, the data structure of \cref{thm:oracle-A} stores the paths for ranges $R_1$ and $R_2$ in the range tree structure or single vertices. The preprocessing and query algorithms for $A^{l,l}$ are as follows, and $A^{l,r},A^{r,l},A^{r,r}$ are similar.
\begin{itemize}
    \item Preprocessing: First compute $A^{l,l}(a,b)$ for all vertices $a,b$ on $st$, which can be done by APSP in $G-st$ in $O(n^3)$ time. Then by \cref{lemma:merging-A}, we can compute $A^{l,l}(R_1,R_2)$ for all vertex-disjoint ranges $R_1$ and $R_2$ in the range tree structure in $O(n^2)$ time.
    \item Query: Since every interval is the union of $O(\log n)$ ranges in the range tree structure (or a single vertex), by \cref{lemma:merging-A}, we can get $A^{l,l}(R_1,R_2)$ for any $R_1$ and $R_2$ in $\too{1}$ time.
\end{itemize}

This completes the proof of \cref{thm:oracle-A}.

\begin{theorem}\label{thm:oracle-B}
   For a graph $G$ and vertices $s,t$, we can construct an oracle $B$ in $\too{n^3}$ preprocessing time, and answer the following type of query in $\too{1}$ time: given an edge $d_1=(a,b)$ on $st$ and two vertex-disjoint intervals $R_1$ and $R_2$ of consecutive edges on $st$ after $d_1$ ($R_1$ can be on the left or right of $R_2$), answer the shortest path that starts from $s$, diverges before $d_1$, converges in $R_1$, and diverges from $R_1$, then converges in $R_2$, and finally ends at the leftmost (rightmost) point of $R_2$. Denote them by 
    $$B^l(d_1,R_1,R_2) = \min_{x\leq a, y_1, y_2\in {R_1}, z\in {R_2}} \{sx\circ \pi_{G-st}(x, y_1)\circ y_1y_2\circ \pi_{G-st}(y_2, z)\circ z\zdd{R_2}\}$$ 
    $$B^r(d_1,R_1,R_2) = \min_{x\leq a, y_1, y_2\in {R_1}, z\in {R_2}} \{sx\circ \pi_{G-st}(x, y_1)\circ y_1y_2\circ \pi_{G-st}(y_2, z)\circ z\ydd{R_2}\}$$ 
\end{theorem}

\begin{center}
    \tikzset{every picture/.style={line width=0.75pt}} 

\begin{tikzpicture}[x=0.75pt,y=0.75pt,yscale=-1,xscale=1]

\draw    (40,402.8) -- (102,402.8) ;
\draw    (115,402.8) -- (401,402.8) ;
\draw [color={rgb, 255:red, 144; green, 19; blue, 254 }  ,draw opacity=0.5 ][line width=1.5]    (70,402.8) .. controls (130.2,331.6) and (209.2,333.2) .. (271,402.8) ;
\draw [color={rgb, 255:red, 144; green, 19; blue, 254 }  ,draw opacity=0.5 ][line width=1.5]    (40,402.8) -- (71,402.8) ;
\draw [color={rgb, 255:red, 144; green, 19; blue, 254 }  ,draw opacity=0.5 ][line width=1.5]    (271,402.8) -- (302,402.8) ;
\draw    (279.6,433.6) -- (260,433.6) ;
\draw [shift={(260,433.6)}, rotate = 360] [color={rgb, 255:red, 0; green, 0; blue, 0 }  ][line width=0.75]    (0,5.59) -- (0,-5.59)   ;
\draw    (298,433.6) -- (317,433.6) ;
\draw [shift={(317,433.6)}, rotate = 180] [color={rgb, 255:red, 0; green, 0; blue, 0 }  ][line width=0.75]    (0,5.59) -- (0,-5.59)   ;
\draw [color={rgb, 255:red, 144; green, 19; blue, 254 }  ,draw opacity=0.5 ][line width=1.5]    (148.8,402.8) .. controls (200,347.6) and (249.2,347.2) .. (302,402.8) ;
\draw    (155.6,433.6) -- (136,433.6) ;
\draw [shift={(136,433.6)}, rotate = 360] [color={rgb, 255:red, 0; green, 0; blue, 0 }  ][line width=0.75]    (0,5.59) -- (0,-5.59)   ;
\draw    (174,433.6) -- (193,433.6) ;
\draw [shift={(193,433.6)}, rotate = 180] [color={rgb, 255:red, 0; green, 0; blue, 0 }  ][line width=0.75]    (0,5.59) -- (0,-5.59)   ;
\draw [color={rgb, 255:red, 144; green, 19; blue, 254 }  ,draw opacity=0.5 ][line width=1.5]    (193,402.8) -- (148.8,402.8) ;

\draw (102.8,399.6) node [anchor=north west][inner sep=0.75pt]  [font=\footnotesize]  {$d_{1}$};
\draw (398,406.2) node [anchor=north west][inner sep=0.75pt]  [font=\footnotesize]  {$t$};
\draw (37,406.2) node [anchor=north west][inner sep=0.75pt]  [font=\footnotesize]  {$s$};
\draw (267,406.2) node [anchor=north west][inner sep=0.75pt]  [font=\footnotesize]  {$y_{1}$};
\draw (296,406.2) node [anchor=north west][inner sep=0.75pt]  [font=\footnotesize]  {$y_{2}$};
\draw (281,429) node [anchor=north west][inner sep=0.75pt]  [font=\footnotesize]  {$R_{1}$};
\draw (144.6,406.2) node [anchor=north west][inner sep=0.75pt]  [font=\footnotesize]  {$z$};
\draw (157,429) node [anchor=north west][inner sep=0.75pt]  [font=\footnotesize]  {$R_{2}$};
\draw (187.53,405) node [anchor=north west][inner sep=0.75pt]  [font=\footnotesize]  {$\ydd{R_{2}}$};
\draw (66,406.73) node [anchor=north west][inner sep=0.75pt]  [font=\footnotesize]  {$x$};
\draw (165.33,330.45) node [anchor=north west][inner sep=0.75pt]  [font=\footnotesize]  {$B^{r}( d_{1} ,R_{1} ,R_{2})$};

\end{tikzpicture}
\end{center}

Also we can store the path by storing the positions of $x,y_1,y_2,z$. We also apply the idea of \cref{lemma:merging-A} to this theorem.

\begin{lemma}\label{lemma:merging-B}
    If the interval $R_1$ on $st$ is the union of intervals $R_1'$ and $R_1''$ ($\ydd{R_1'}=\zdd{R_1''}$ or there is an edge connecting $\ydd{R_1'}$ and $\zdd{R_1''}$), and we already have $B^l(d_1,R_1',R_2)$, $B^l(d_1,R_1'',R_2)$ and oracle $A$, then we can get $B^l(d_1,R_1,R_2)$ in $\too{1}$ time. Also, we can get $B^l(d_1,R_1,R_2)$ from $B^l(d_1,R_1,R_2')$ and $B^l(d_1,R_1,R_2'')$ in $O(1)$ time if $R_2$ is the union of $R_2'$ and $R_2''$. This also holds for $B^r$.
\end{lemma}
\begin{proof}
    The case where $R_2$ is the union of $R_2'$ and $R_2''$ is similar to \cref{lemma:merging-A}, so here we only consider the case where $R_1$ is the union of $R_1'$ and $R_1''$. W.l.o.g., assume $R_1'$ is on the left of $R_1''$.

    The path of $B^l(d_1,R_1,R_2)$ is the minimum of these cases: 
    \begin{enumerate}
        \item It converges and diverges in interval $R_1'$
        \item It converges and diverges in interval $R_1''$
        \item It converges in $R_1'$ and diverges from $R_1''$
        \item It converges in $R_1''$ and diverges from $R_1'$        
    \end{enumerate}

    The first two cases are equal to $B^l(d_1,R_1',R_2)$ and $B^l(d_1,R_1'',R_2)$, respectively, and the paths in the last two cases must go through the vertex (or edge) between $R_1'$ and $R_1''$, and we can use oracle $A$ to deal with these cases. Let $R_3$ be the interval from $s$ to the left vertex $a$ of $d_1$, then the path $B^l(d_1,R_1,R_2)$ is:

     $$\min\left\{B^l(d_1,R_1',R_2),\; B^l(d_1,R_1'',R_2),\; A^{l,r}(R_3,R_1')\circ A^{l,l}(R_1'',R_2),\; A^{l,l}(R_3,R_1'')\circ A^{r,l}(R_1',R_2)\right\}$$

    Also $B^r$ can be computed similarly.

    \begin{center}
        \tikzset{every picture/.style={line width=0.75pt}} 

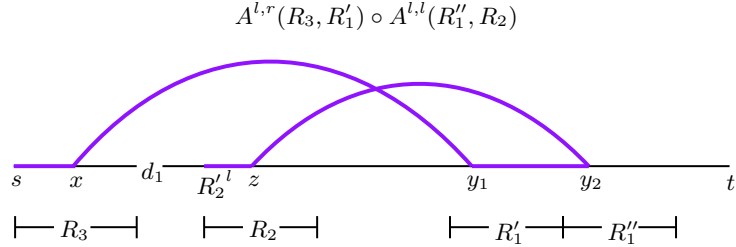
\begin{tikzpicture}[x=0.75pt,y=0.75pt,yscale=-1,xscale=1]

\draw    (40,136.5) -- (102,136.5) ;
\draw    (115,136.5) -- (401,136.5) ;
\draw [color={rgb, 255:red, 144; green, 19; blue, 254 }  ,draw opacity=0.5 ][line width=1.5]    (70,136.5) .. controls (130.2,65.3) and (209.2,66.9) .. (271,136.5) ;
\draw [color={rgb, 255:red, 144; green, 19; blue, 254 }  ,draw opacity=0.5 ][line width=1.5]    (40,136.5) -- (71,136.5) ;
\draw [color={rgb, 255:red, 144; green, 19; blue, 254 }  ,draw opacity=0.5 ][line width=1.5]    (271,136.5) -- (330,136.5) ;
\draw    (279.6,167.3) -- (260,167.3) ;
\draw [shift={(260,167.3)}, rotate = 360] [color={rgb, 255:red, 0; green, 0; blue, 0 }  ][line width=0.75]    (0,5.59) -- (0,-5.59)   ;
\draw    (298,167.3) -- (317,167.3) ;
\draw [shift={(317,167.3)}, rotate = 180] [color={rgb, 255:red, 0; green, 0; blue, 0 }  ][line width=0.75]    (0,5.59) -- (0,-5.59)   ;
\draw [color={rgb, 255:red, 144; green, 19; blue, 254 }  ,draw opacity=0.5 ][line width=1.5]    (160,136.5) .. controls (211.2,81.3) and (277.2,80.9) .. (330,136.5) ;
\draw    (155.6,167.3) -- (136,167.3) ;
\draw [shift={(136,167.3)}, rotate = 360] [color={rgb, 255:red, 0; green, 0; blue, 0 }  ][line width=0.75]    (0,5.59) -- (0,-5.59)   ;
\draw    (174,167.3) -- (193,167.3) ;
\draw [shift={(193,167.3)}, rotate = 180] [color={rgb, 255:red, 0; green, 0; blue, 0 }  ][line width=0.75]    (0,5.59) -- (0,-5.59)   ;
\draw [color={rgb, 255:red, 144; green, 19; blue, 254 }  ,draw opacity=0.5 ][line width=1.5]    (136,136.5) -- (146,136.5) -- (160,136.5) ;
\draw    (336.6,167.3) -- (317,167.3) ;
\draw [shift={(317,167.3)}, rotate = 360] [color={rgb, 255:red, 0; green, 0; blue, 0 }  ][line width=0.75]    (0,5.59) -- (0,-5.59)   ;
\draw    (355,167.3) -- (374,167.3) ;
\draw [shift={(374,167.3)}, rotate = 180] [color={rgb, 255:red, 0; green, 0; blue, 0 }  ][line width=0.75]    (0,5.59) -- (0,-5.59)   ;
\draw    (60.27,167.3) -- (40.67,167.3) ;
\draw [shift={(40.67,167.3)}, rotate = 360] [color={rgb, 255:red, 0; green, 0; blue, 0 }  ][line width=0.75]    (0,5.59) -- (0,-5.59)   ;
\draw    (78.67,167.3) -- (102,167.3) ;
\draw [shift={(102,167.3)}, rotate = 180] [color={rgb, 255:red, 0; green, 0; blue, 0 }  ][line width=0.75]    (0,5.59) -- (0,-5.59)   ;

\draw (102.8,133.3) node [anchor=north west][inner sep=0.75pt]  [font=\footnotesize]  {$d_{1}$};
\draw (398,139.9) node [anchor=north west][inner sep=0.75pt]  [font=\footnotesize]  {$t$};
\draw (37,139.9) node [anchor=north west][inner sep=0.75pt]  [font=\footnotesize]  {$s$};
\draw (267,139.9) node [anchor=north west][inner sep=0.75pt]  [font=\footnotesize]  {$y_{1}$};
\draw (324,139.9) node [anchor=north west][inner sep=0.75pt]  [font=\footnotesize]  {$y_{2}$};
\draw (281,162.7) node [anchor=north west][inner sep=0.75pt]  [font=\footnotesize]  {$R'_{1}$};
\draw (130.6,137) node [anchor=north west][inner sep=0.75pt]  [font=\footnotesize]  {$\zdd{R'_{2}}$};
\draw (157,162.7) node [anchor=north west][inner sep=0.75pt]  [font=\footnotesize]  {$R_{2}$};
\draw (156.53,139.9) node [anchor=north west][inner sep=0.75pt]  [font=\footnotesize]  {$z$};
\draw (338,162.7) node [anchor=north west][inner sep=0.75pt]  [font=\footnotesize]  {$R''_{1}$};
\draw (150,52.4) node [anchor=north west][inner sep=0.75pt]  [font=\footnotesize]  {$A^{l,r}( R_{3} ,R'_{1}) \circ A^{l,l}( R''_{1} ,R_{2})$};
\draw (66.53,140.56) node [anchor=north west][inner sep=0.75pt]  [font=\footnotesize]  {$x$};
\draw (61.67,162.7) node [anchor=north west][inner sep=0.75pt]  [font=\footnotesize]  {$R_{3}$};

\end{tikzpicture}
    \end{center}

\end{proof}

So the data structure of \cref{thm:oracle-B} stores the paths for all $d_1$ on $s,t$ and ranges $R_1$ and $R_2$ in the range tree structure or single vertices, so the total space is $O(n^3)$. The preprocessing and query algorithm for $B^l$ is as follows, and $B^r$ is similar.
\begin{itemize}
    \item Preprocessing: For all $d_1$ on $s,t$ and $R_2$ in the range tree structure or single vertex, consider $R_1$ in the range tree structure which is disjoint with $R_2$.
    \begin{itemize}
        \item First compute $B^l(d_1,x,R_2)$ for all single vertex $x$ after $d_1$ and not in $R_2$, which we can get from oracle $A$ in $\too{n^3}$ time. Let $R_3$ be the interval from $s$ to the left vertex of $d_1$, then $B^l(d_1,x,R_2)=A^{l,l}(R_3,x)\circ A^{l,l}(x,R_2)$.
        \item Then by \cref{lemma:merging-B}, we can compute $B^l(d_1,R_1,R_2)$ for all vertex-disjoint ranges $R_1$ and $R_2$ in the range tree structure in $\too{n^3}$ time.
    \end{itemize}
    \item Query: Since every interval on $st$ is the union of $O(\log n)$ ranges in the range tree structure (or a single vertex), by \cref{lemma:merging-B}, we can get $B^l(d_1,R_1,R_2)$ for any $d_1,R_1,R_2$ in $\too{1}$ time.
\end{itemize}

This completes the proof of \cref{thm:oracle-B}. Note that \cref{thm:oracle-B} can be extended to the symmetric case that the shortest path that starts at the leftmost (rightmost) point of $R_1$, diverges from $R_1$, converges in $R_2$, and diverges from $R_2$, then converges after $d_3$, and finally ends at $t$, for $d_3$ on $s,t$ and vertex-disjoint intervals $R_1$ and $R_2$ on $st$ before $d_3$.

\subsection{Answering 3-Fault Queries}
Now we are ready to answer queries for the replacement path $\pi'=\pi_{G-\{d_1,d_2,d_3\}}(s,t)$ given three failed edges $d_1,d_2,d_3$ on $st$. We suppose the original shortest path $st$ is cut into four intervals $D_1,D_2,D_3,D_4$. Here we first state a more detailed version of Observation \ref{obs:3failure}:

\begin{observation}[cf. \ref{obs:3failure}, formal]
The 3-fault replacement path $\pi'$ falls in one of the cases:
\begin{itemize}
    \item[-] \textbf{Type 1:} Diverges at $D_1$, converges at $D_4.$
    \item[-] \textbf{Type 2:} Diverges at $D_1$, converges and diverges at $D_2$, and finally converges at $D_4.$
    \item[-] \textbf{Type 3:} Diverges at $D_1$, converges and diverges at $D_3$, and finally converges at $D_4.$
    \item[-] \textbf{Type 4:} Diverges at $D_1$, converges and diverges at $D_2$, then converges and diverges at $D_3$, and finally converges at $D_4.$
    \item[-] \textbf{Type 5:} Diverges at $D_1$, converges and diverges at $D_3$, then converges and diverges at $D_2$, and finally converges at $D_4.$
\end{itemize}    
\end{observation}

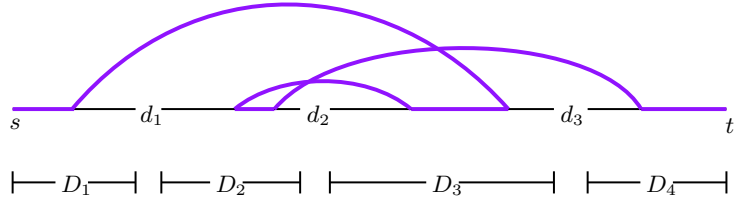
\begin{center}
    \tikzset{every picture/.style={line width=0.75pt}} 

\begin{tikzpicture}[x=0.75pt,y=0.75pt,yscale=-1,xscale=1]

\draw    (60,102.8) -- (122,102.8) ;
\draw    (220,102.8) -- (333,102.8) ;
\draw [color={rgb, 255:red, 144; green, 19; blue, 254 }  ,draw opacity=0.5 ][line width=1.5]    (90,102.8) .. controls (150.2,31.6) and (248.2,33.2) .. (310,102.8) ;
\draw [color={rgb, 255:red, 144; green, 19; blue, 254 }  ,draw opacity=0.5 ][line width=1.5]    (60,102.8) -- (91,102.8) ;
\draw    (135,102.8) -- (205,102.8) ;
\draw    (350,102.8) -- (420,102.8) ;
\draw [color={rgb, 255:red, 144; green, 19; blue, 254 }  ,draw opacity=0.5 ][line width=1.5]    (172,102.8) .. controls (194.5,84.25) and (238,83.75) .. (261,102.8) ;
\draw [color={rgb, 255:red, 144; green, 19; blue, 254 }  ,draw opacity=0.5 ][line width=1.5]    (192,102.8) .. controls (227.5,60.75) and (353,62.75) .. (377,102.8) ;
\draw [color={rgb, 255:red, 144; green, 19; blue, 254 }  ,draw opacity=0.5 ][line width=1.5]    (172,102.8) -- (192,102.8) ;
\draw [color={rgb, 255:red, 144; green, 19; blue, 254 }  ,draw opacity=0.5 ][line width=1.5]    (261,102.8) -- (310,102.8) ;
\draw [color={rgb, 255:red, 144; green, 19; blue, 254 }  ,draw opacity=0.5 ][line width=1.5]    (377,102.8) -- (420,102.8) ;
\draw    (81.6,139.8) -- (60,139.8) ;
\draw [shift={(60,139.8)}, rotate = 360] [color={rgb, 255:red, 0; green, 0; blue, 0 }  ][line width=0.75]    (0,5.59) -- (0,-5.59)   ;
\draw    (98,139.8) -- (122,139.8) ;
\draw [shift={(122,139.8)}, rotate = 180] [color={rgb, 255:red, 0; green, 0; blue, 0 }  ][line width=0.75]    (0,5.59) -- (0,-5.59)   ;
\draw    (159.6,139.8) -- (135,139.8) ;
\draw [shift={(135,139.8)}, rotate = 360] [color={rgb, 255:red, 0; green, 0; blue, 0 }  ][line width=0.75]    (0,5.59) -- (0,-5.59)   ;
\draw    (176,139.8) -- (205,139.8) ;
\draw [shift={(205,139.8)}, rotate = 180] [color={rgb, 255:red, 0; green, 0; blue, 0 }  ][line width=0.75]    (0,5.59) -- (0,-5.59)   ;
\draw    (268.6,139.8) -- (220,139.8) ;
\draw [shift={(220,139.8)}, rotate = 360] [color={rgb, 255:red, 0; green, 0; blue, 0 }  ][line width=0.75]    (0,5.59) -- (0,-5.59)   ;
\draw    (285,139.8) -- (333,139.8) ;
\draw [shift={(333,139.8)}, rotate = 180] [color={rgb, 255:red, 0; green, 0; blue, 0 }  ][line width=0.75]    (0,5.59) -- (0,-5.59)   ;
\draw    (376.6,139.8) -- (350,139.8) ;
\draw [shift={(350,139.8)}, rotate = 360] [color={rgb, 255:red, 0; green, 0; blue, 0 }  ][line width=0.75]    (0,5.59) -- (0,-5.59)   ;
\draw    (393,139.8) -- (420,139.8) ;
\draw [shift={(420,139.8)}, rotate = 180] [color={rgb, 255:red, 0; green, 0; blue, 0 }  ][line width=0.75]    (0,5.59) -- (0,-5.59)   ;

\draw (122.8,99.6) node [anchor=north west][inner sep=0.75pt]  [font=\footnotesize]  {$d_{1}$};
\draw (418,106.2) node [anchor=north west][inner sep=0.75pt]  [font=\footnotesize]  {$t$};
\draw (57,106.2) node [anchor=north west][inner sep=0.75pt]  [font=\footnotesize]  {$s$};
\draw (206.8,99.6) node [anchor=north west][inner sep=0.75pt]  [font=\footnotesize]  {$d_{2}$};
\draw (334.8,99.6) node [anchor=north west][inner sep=0.75pt]  [font=\footnotesize]  {$d_{3}$};
\draw (83,135.2) node [anchor=north west][inner sep=0.75pt]  [font=\footnotesize]  {$D_{1}$};
\draw (161,135.2) node [anchor=north west][inner sep=0.75pt]  [font=\footnotesize]  {$D_{2}$};
\draw (270,135.2) node [anchor=north west][inner sep=0.75pt]  [font=\footnotesize]  {$D_{3}$};
\draw (378,135.2) node [anchor=north west][inner sep=0.75pt]  [font=\footnotesize]  {$D_{4}$};

\end{tikzpicture}
\end{center}

\begin{proof}
The key observation is that in an undirected graph, if $\pi'$ diverges from some interval $D_i$, it never converges on the same interval. 
Then we can enumerate all the possible cases as shown. 

\end{proof}

{
Therefore, in the following part we consider the shortest candidates in all 5 types, and a minimization of them will be the correct replacement path that we want.
}

~\\
\noindent\textbf{When shortest path $\pi'$ is of type 1,} we can solve this case by straightly calling $A^{l,r}(D_1,D_4)$ from Theorem \ref{thm:oracle-A}, which captures type 1 paths.



~\\
\noindent\textbf{When shortest path $\pi'$ is of type 2 or type 3,} similarly, by Theorem \ref{thm:oracle-B}, we can call $B^r(d_1,D_2,D_4)$ and $B^r(d_1,D_3,D_4)$ to capture type 2 and type 3 paths, respectively.


~\\
\noindent\textbf{When shortest path $\pi'$ is of type 4 or type 5,} the case is harder. For example, suppose we want to compute the path with type 5 below (and similarly for type 4). Here, as illustrated in the above figure, the path can be formally expressed as:
$$\min_{s_0\in D_1, x_1, x_2\in D_3, y_1, y_2\in D_2, t_0\in D_4} \left\{s s_0 \circ \pi_{G-st}(s_0, x_1) \circ x_1x_2 \circ \pi_{G-st}(x_2, y_1) \circ y_1y_2 \circ \pi_{G-st}(y_2, t_0) \circ t_0t \right\}. $$

\begin{lemma}\label{lemma:snake3}
    Let $d_1,d_2,d_3$ be three edges on $\pi(s,t)$, and $x$ be a vertex (or an edge) on $\pi(s,t)$, which is after $d_1$ and before $d_3.$ Calculating a shortest path of type 4 or 5 from $s$ to $t$ in $G-\{d_1,d_2,d_3\}$ that goes through $x$ requires time $\too{1}.$ 
\end{lemma}

\begin{proof}
    W.l.o.g., we suppose $x$ is a vertex in $D_2$ and the path is with type 5 (since all other cases can be dealt with similarly). Suppose $x$ divides $D_2$ into two intervals $R_2$ and $R_3$. ($R_2$ is on the left of $R_3$.)
    
    
    We call the oracles in \cref{sec:3-fault-oracle} and the answer is the minimum of $B^l(d_1,D_3,R_3)\circ {A}^{r,r}(R_2,D_4)$ and $B^r(d_1,D_3,R_2)\circ {A}^{l,r}(R_3,D_4)$. 
    
\end{proof}

Below we generalize type 4 and 5 paths as  \textbf{snake paths} when there are $w=\Tilde{O}(1)$ edges removed.

\begin{definition}

Suppose $w=\Tilde{O}(1)$, and there are $w$ edges $d_1,d_2,\cdots,d_w$ on $st$ removed. Suppose they cut $st$ into $w+1$ intervals $D_1, \dots , D_{w+1}$, a \textbf{snake path} is defined as a path that only converges and then diverges in exactly two intervals in $D_2, \dots, D_w$ in the middle. 

\end{definition}

With this definition, we similarly extend the lemma as follows:


\begin{lemma}\label{lemma:snake}
    If $x$ is a vertex (or an edge) after $d_1$ and before $d_w$, then calculating a shortest snake path that goes through $x$ avoiding $\{d_1,d_2,\cdots,d_w\}$ still requires time $\too{1}.$ Moreover, we can compute a shortest snake path in $G-\{d_1,d_2,\cdots,d_w\}$ in time $\too{n}.$
\end{lemma}



\begin{proof}
    First, for the shortest snake path passing $x$, we can enumerate which $D_i$ is the first interval that the snake path intersects with $st$ after diverging from $D_1$ and which $D_j$ is the second one (one of them contains $x$). After fixing $D_i$ and $D_j$, we perform the same oracle calls as in Lemma \ref{lemma:snake3}. Since the total number of possible pairs of $D_i,D_j$ is still $\too{1}$, we can still compute in $\too{1}$ time.

    Moreover, if we do not know such an $x$ in advance, we can just enumerate all possible $x$ and find the shortest one among the snake paths passing them. This will cost time $\too{n}$.
    
\end{proof}




We are ready to close this section by proving the following lemma.
\begin{lemma}\label{thm1-5}
    Let $d_1,d_3$ be two edges on $\pi(s,t).$ Calculating a shortest snake path from $s$ to $t$ in $G-\{d_1,d_2,d_3\}$ for all possible $d_2$ between $d_1$ and $d_3$ requires time $\too{n}.$
\end{lemma}

We can easily see that, if Lemma \ref{thm1-5} is true, then we can enumerate all possible $d_1$ and $d_3$ and then get answers for all $d_2$. Therefore, we can answer the replacement path for all triples $\{d_1,d_2,d_3\}$, and the total time is in $\too{n^3}$. Therefore, to prove Theorem \ref{thm:3edges}, it remains for us to prove this lemma. 
 
\begin{proof}
~\\~\\
\noindent\textbf{Algorithm.}

We begin with introducing the algorithm. Let $D$ be the interval between $d_1$ and $d_3$. 

\begin{itemize}

    \item In Stage 1, let $e_1$ be the middle edge of $D$. We calculate the shortest snake path $\pi_1$ that avoids $\{d_1,d_3,e_1\}$ in $\too{n}$ time by \cref{lemma:snake}. So its intersection with $D$ is a union of two intervals. Let $E_1^{(1)}$ be the one with more edges, and $E_2^{(1)}$ be the other one. For $d_2$ not in $E_1^{(1)}\cup E_2^{(1)}$, the shortest snake path avoiding $\{d_1,d_2,d_3\}$ either goes through $e_1$, which can be computed by \cref{lemma:snake3}, or avoids $e_1$, which is just $\pi_1$. So next we only need to consider $d_2\in E_1^{(1)}\cup E_2^{(1)}$.

    \item In Stage $i$ ($i\geq 2$):
    \begin{itemize}

        \item[-] Let $e_i$ be the middle edge of $E_1^{(i-1)}$. Calculate the shortest snake path that avoids $\{d_1,d_3,e_1,\cdots,e_i\}$ and call it $\pi_i$.
    
        \item[-] Let $F=\pi_i \cap D$. Since $\pi_i$ is a snake path, $F$ is a union of at most two intervals. Consider $F \cap \left(E_1^{(i-1)}\cup E_2^{(i-1)}\right)$. Since there is at least an edge between $E_1^{(i-1)}$ and $E_2^{(i-1)}$ that is not in $F$, each interval of $F$ can not intersect with both $E_1^{(i-1)}$ and $E_2^{(i-1)}$, so $F \cap \left(E_1^{(i-1)}\cup E_2^{(i-1)}\right)$ consists of at most two intervals. Call the one with more edges $E_1^{(i)}$, and call the other one $E_2^{(i)}$ (both can be empty).
      \end{itemize}
    \item When both of $E_1^{(i)}$ and $E_2^{(i)}$ become empty or single vertex, end the loop. Otherwise, go to stage $i+1$.
        

    
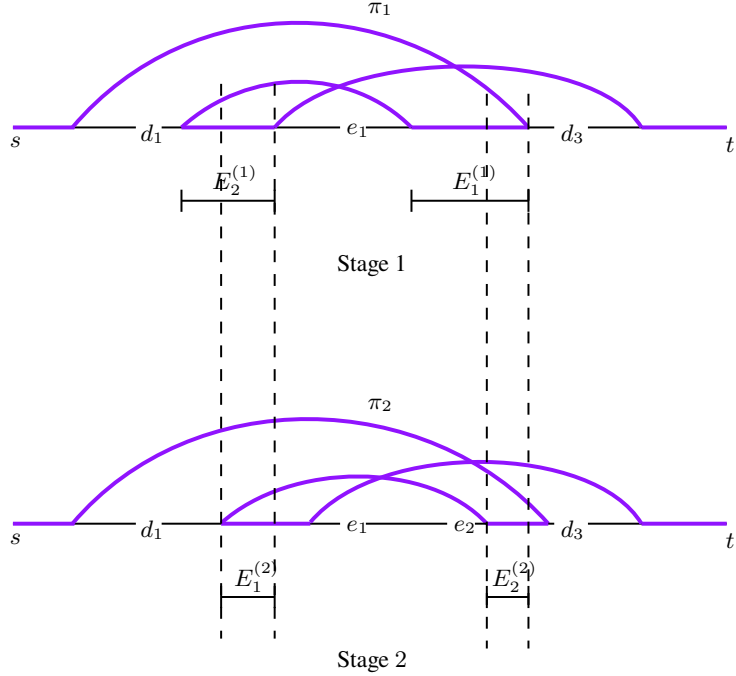
\begin{center}
    \tikzset{every picture/.style={line width=0.75pt}} 

\begin{tikzpicture}[x=0.75pt,y=0.75pt,yscale=-1,xscale=1]

\draw    (60,102) -- (122,102) ;
\draw    (238,102) -- (333,102) ;
\draw [color={rgb, 255:red, 144; green, 19; blue, 254 }  ,draw opacity=0.5 ][line width=1.5]    (90,102) .. controls (150.2,30.8) and (258.2,32.4) .. (320,102) ;
\draw [color={rgb, 255:red, 144; green, 19; blue, 254 }  ,draw opacity=0.5 ][line width=1.5]    (60,102) -- (91,102) ;
\draw    (135,102) -- (225,102) ;
\draw    (350,102) -- (420,102) ;
\draw [color={rgb, 255:red, 144; green, 19; blue, 254 }  ,draw opacity=0.5 ][line width=1.5]    (145,102) .. controls (179,71.25) and (230.5,71.25) .. (261,102) ;
\draw [color={rgb, 255:red, 144; green, 19; blue, 254 }  ,draw opacity=0.5 ][line width=1.5]    (192,102) .. controls (227.5,59.95) and (353,61.95) .. (377,102) ;
\draw [color={rgb, 255:red, 144; green, 19; blue, 254 }  ,draw opacity=0.5 ][line width=1.5]    (145,102) -- (192,102) ;
\draw [color={rgb, 255:red, 144; green, 19; blue, 254 }  ,draw opacity=0.5 ][line width=1.5]    (261,102) -- (320,102) ;
\draw [color={rgb, 255:red, 144; green, 19; blue, 254 }  ,draw opacity=0.5 ][line width=1.5]    (377,102) -- (420,102) ;
\draw    (192,139) -- (145,139) ;
\draw [shift={(145,139)}, rotate = 360] [color={rgb, 255:red, 0; green, 0; blue, 0 }  ][line width=0.75]    (0,5.59) -- (0,-5.59)   ;
\draw [shift={(192,139)}, rotate = 360] [color={rgb, 255:red, 0; green, 0; blue, 0 }  ][line width=0.75]    (0,5.59) -- (0,-5.59)   ;
\draw    (320,139) -- (261,139) ;
\draw [shift={(261,139)}, rotate = 360] [color={rgb, 255:red, 0; green, 0; blue, 0 }  ][line width=0.75]    (0,5.59) -- (0,-5.59)   ;
\draw [shift={(320,139)}, rotate = 360] [color={rgb, 255:red, 0; green, 0; blue, 0 }  ][line width=0.75]    (0,5.59) -- (0,-5.59)   ;
\draw    (60,302) -- (122,302) ;
\draw    (238,302) -- (279,302) ;
\draw [color={rgb, 255:red, 144; green, 19; blue, 254 }  ,draw opacity=0.5 ][line width=1.5]    (90,302) .. controls (150.2,230.8) and (268.2,232.4) .. (330,302) ;
\draw [color={rgb, 255:red, 144; green, 19; blue, 254 }  ,draw opacity=0.5 ][line width=1.5]    (60,302) -- (91,302) ;
\draw    (135,302) -- (225,302) ;
\draw    (350,302) -- (420,302) ;
\draw [color={rgb, 255:red, 144; green, 19; blue, 254 }  ,draw opacity=0.5 ][line width=1.5]    (165,301.75) .. controls (200,270.25) and (270,270.25) .. (299,302) ;
\draw [color={rgb, 255:red, 144; green, 19; blue, 254 }  ,draw opacity=0.5 ][line width=1.5]    (210,301.25) .. controls (245.5,259.2) and (353,261.95) .. (377,302) ;
\draw [color={rgb, 255:red, 144; green, 19; blue, 254 }  ,draw opacity=0.5 ][line width=1.5]    (165,302) -- (210,302) ;
\draw [color={rgb, 255:red, 144; green, 19; blue, 254 }  ,draw opacity=0.5 ][line width=1.5]    (377,302) -- (420,302) ;
\draw    (192,339) -- (165,339) ;
\draw [shift={(165,339)}, rotate = 360] [color={rgb, 255:red, 0; green, 0; blue, 0 }  ][line width=0.75]    (0,5.59) -- (0,-5.59)   ;
\draw [shift={(192,339)}, rotate = 360] [color={rgb, 255:red, 0; green, 0; blue, 0 }  ][line width=0.75]    (0,5.59) -- (0,-5.59)   ;
\draw    (333,302) -- (293,302) ;
\draw    (320,339) -- (299,339) ;
\draw [shift={(299,339)}, rotate = 360] [color={rgb, 255:red, 0; green, 0; blue, 0 }  ][line width=0.75]    (0,5.59) -- (0,-5.59)   ;
\draw [shift={(320,339)}, rotate = 360] [color={rgb, 255:red, 0; green, 0; blue, 0 }  ][line width=0.75]    (0,5.59) -- (0,-5.59)   ;
\draw [color={rgb, 255:red, 144; green, 19; blue, 254 }  ,draw opacity=0.5 ][line width=1.5]    (299,302) -- (330,302) ;
\draw [color={rgb, 255:red, 0; green, 0; blue, 0 }  ,draw opacity=0.3 ] [dash pattern={on 4.5pt off 4.5pt}]  (165,80) -- (165,360) ;
\draw [color={rgb, 255:red, 0; green, 0; blue, 0 }  ,draw opacity=0.3 ] [dash pattern={on 4.5pt off 4.5pt}]  (192,80) -- (192,360) ;
\draw [color={rgb, 255:red, 0; green, 0; blue, 0 }  ,draw opacity=0.3 ] [dash pattern={on 4.5pt off 4.5pt}]  (299,85) -- (299,365) ;
\draw [color={rgb, 255:red, 0; green, 0; blue, 0 }  ,draw opacity=0.3 ] [dash pattern={on 4.5pt off 4.5pt}]  (320,85) -- (320,365) ;

\draw (122.8,98.8) node [anchor=north west][inner sep=0.75pt]  [font=\footnotesize]  {$d_{1}$};
\draw (418,105.4) node [anchor=north west][inner sep=0.75pt]  [font=\footnotesize]  {$t$};
\draw (57,105.4) node [anchor=north west][inner sep=0.75pt]  [font=\footnotesize]  {$s$};
\draw (226.8,98.8) node [anchor=north west][inner sep=0.75pt]  [font=\footnotesize]  {$e_{1}$};
\draw (334.8,98.8) node [anchor=north west][inner sep=0.75pt]  [font=\footnotesize]  {$d_{3}$};
\draw (159,119.6) node [anchor=north west][inner sep=0.75pt]  [font=\footnotesize]  {$E_{2}^{( 1)}$};
\draw (280,119.6) node [anchor=north west][inner sep=0.75pt]  [font=\footnotesize]  {$E_{1}^{( 1)}$};
\draw (237.5,37.4) node [anchor=north west][inner sep=0.75pt]  [font=\footnotesize]  {$\pi _{1}$};
\draw (122.8,298.8) node [anchor=north west][inner sep=0.75pt]  [font=\footnotesize]  {$d_{1}$};
\draw (418,305.4) node [anchor=north west][inner sep=0.75pt]  [font=\footnotesize]  {$t$};
\draw (57,305.4) node [anchor=north west][inner sep=0.75pt]  [font=\footnotesize]  {$s$};
\draw (226.8,298.8) node [anchor=north west][inner sep=0.75pt]  [font=\footnotesize]  {$e_{1}$};
\draw (334.8,298.8) node [anchor=north west][inner sep=0.75pt]  [font=\footnotesize]  {$d_{3}$};
\draw (170,319.6) node [anchor=north west][inner sep=0.75pt]  [font=\footnotesize]  {$E_{1}^{( 2)}$};
\draw (237.5,237.4) node [anchor=north west][inner sep=0.75pt]  [font=\footnotesize]  {$\pi _{2}$};
\draw (280.8,298.8) node [anchor=north west][inner sep=0.75pt]  [font=\footnotesize]  {$e_{2}$};
\draw (300,319.6) node [anchor=north west][inner sep=0.75pt]  [font=\footnotesize]  {$E_{2}^{( 2)}$};
\draw (222,164.9) node [anchor=north west][inner sep=0.75pt]  [font=\footnotesize]  {$\text{Stage\ 1}$};
\draw (222,364.9) node [anchor=north west][inner sep=0.75pt]  [font=\footnotesize]  {$\text{Stage\ } 2$};

\end{tikzpicture}
\end{center}

\end{itemize}

We observe that $E_1^{(i)} \cup E_2^{(i)}=D\cap (\pi_1 \cap \pi_2 \cap \dots \cap \pi_i)$. Therefore, for each $d_2$ between $d_1,d_3$, there exists some $\pi_i$ avoiding $d_2$, and we can find the smallest $k$ such that $\pi_k$ avoids $d_2$. Then we take the shortest snake path avoiding $\{d_1,d_2,d_3\}$ as the shortest one among $\pi_k$ and all snake paths going through $e_i$ that avoids $\{d_1,d_2,d_3\}$ for all $i\leq k$.


The correctness proof of this algorithm is straightforward. Consider any 
edge $d_2$ between $d_1$ and $d_3$, suppose $\pi_k$ is the first one in $\pi_1, \pi_2, \dots$ that avoids $d_2$. When $\og{s}{t}{\set{d_1, d_2, d_3}}$ is a snake path:

\begin{itemize}
    \item If it avoids $e_1,e_2, \dots, e_{k-1},e_k$, then it is $\pi_k$ by definition. 
    \item Otherwise, it will go through one of $e_1,e_2, \dots, e_{k-1},e_k$, so we can still capture it.
\end{itemize}

Therefore, our algorithm is correct.
\\

\noindent\textbf{Time analysis.}

By Lemma \ref{lemma:snake3} and \ref{lemma:snake}, the time of each stage is $\too{n}$, and for each $d_2$, the time needed to find shortest snake path avoiding $\set{d_1, d_2, d_3}$ is $\too{1}$. We will show that there are at most $O(\log{n})$ stages. Let $S^{(i)}= {\left\| E_1^{(i)}\right\|}^2+{\left\| E_2^{(i)}\right\|}^2$, and we prove that $S^{(i+1)}\le\frac{5}{8}S^{(i)}.$

In stage $i+1$, the larger interval $E_1^{(i)}$ is cut into two almost equal-size intervals which we call $E_3^{(i)}$ and $E_4^{(i)}.$ (Note that the number of edges in either of $E_3^{(i)}$ and $E_4^{(i)}$ is at most half in $E_1^{(i)}$.) Then $E_1^{(i+1)}$ and $E_2^{(i+1)}$ is contained in two intervals among $E_2^{(i)}$, $E_3^{(i)}$ and $E_4^{(i)}.$ If they are contained in $E_3^{(i)}$ and $E_4^{(i)}$,
\begin{equation*}
\begin{split}
    S^{(i+1)}&={\left\| E_1^{(i+1)}\right\|}^2+{\left\| E_2^{(i+1)}\right\|}^2
    \le{\left\| E_3^{(i)}\right\|}^2+{\left\| E_4^{(i)}\right\|}^2
    \le\frac{1}{4}{\left\| E_1^{(i)}\right\|}^2+\frac{1}{4}{\left\| E_1^{(i)}\right\|}^2
    \le\frac{1}{2}S^{(i)}.\\
\end{split}
\end{equation*}

If at least one of them is contained in $E_2^{(i)}$, we know
\begin{equation*}
\begin{split}
    S^{(i+1)}&={\left\| E_1^{(i+1)}\right\|}^2+{\left\| E_2^{(i+1)}\right\|}^2
    \le{\left\| E_2^{(i)}\right\|}^2+{\left\| E_3^{(i)}\right\|}^2\\
    &\le{\left\| E_2^{(i)}\right\|}^2+\frac{1}{4}{\left\| E_1^{(i)}\right\|}^2
    \le \frac{5}{8}{\left\| E_2^{(i)}\right\|}^2+\frac{5}{8}{\left\| E_1^{(i)}\right\|}^2=\frac{5}{8}S^{(i)}.\\
    \end{split}
\end{equation*}

\end{proof}
\section{Hardness of 2FRP}\label{sec:hardness}

In this section we show that the 2FRP problem in undirected graphs is hard, assuming the APSP conjecture. We note that in directed graphs, 1FRP is as hard as APSP~\cite{WW18}. However, in undirected graphs, 1FRP can be solved in almost-optimal $\tilde{O}(m)$ time~\cite{NPW01}. 

\begin{theorem}\label{theorem:2FRP}
    If 2FRP in undirected weighted graphs can be solved in $O(n^{3-\epsilon})$ time for any $\epsilon>0$, then in undirected weighted graphs, the all-pairs shortest path problem can also be solved in $O(n^{3-\epsilon})$ time.
\end{theorem}

We remark that by \cite{WW18}, the hardness of the all-pairs shortest path problem in undirected weighted graphs is equivalent to the hardness of APSP. Therefore, any algorithm in $O(n^{3-\epsilon})$ time for 2FRP will refuse the APSP conjecture. So we immediately get the corollary:

\begin{corollary}
    Assuming the APSP conjecture that APSP cannot be solved in $O(n^{3-\epsilon})$ time for any $\epsilon>0$, 2FRP problem in undirected weighted graphs cannot be solved in $O(n^{3-\epsilon})$ time for any $\epsilon>0$.
\end{corollary}

Suppose toward contradiction that we can solve 2FRP in $O(n^{3-\epsilon})$ time. Consider any undirected weighted all-pairs shortest path instance $G$ with $n$ nodes $v_1, v_2, ..., v_n$, we will construct an undirected weighted 2FRP instance $H$, solving which will result in solving all-pairs shortest path in $G$.

Let $N$ be a number that is larger than the sum of all weights in $G$. Let $S$ be a path with $(n+2)$ nodes {$(s_0,s_1,s_2,...,s_n,s_{n+1}=s)$} with weight 0 edges $(s_i,s_{i+1})$ for all $0 \leq i \leq n$. Similarly, let $T$ be a path with $(n+2)$ nodes $(t_0,t_1,t_2,...,t_n,t_{n+1}=t)$ with weight 0 edges $(t_i,t_{i+1})$ for all $0 \leq i \leq n$. Moreover, we construct two matchings $E_1 = \set{(s_i,v_i)}$ with weight $w(s_i,v_i)=iN$ for all $1 \leq i \leq n$, and $E_2 = \set{(v_i,t_i)}$ with weight $w(v_i,t_i)=iN$ for all $1 \leq i \leq n$.

Let {$H=S \cup E_1 \cup G \cup E_2 \cup T$}. By inspection, the $s$ to $t$ shortest path in $H$ is precisely the path goes from $s$ to $s_1$ in $S$, from $s_1$ to $v_1$ to $t1$, and then from $t_1$ to $t$ in $T$. Consider any pair of failed edges, one on $\pi_H(s,s_1)$ and the other on $\pi_H(t_1,t)$. We will show the following relationship for 2FRP in $H$ and all-pairs shortest paths in $G$:

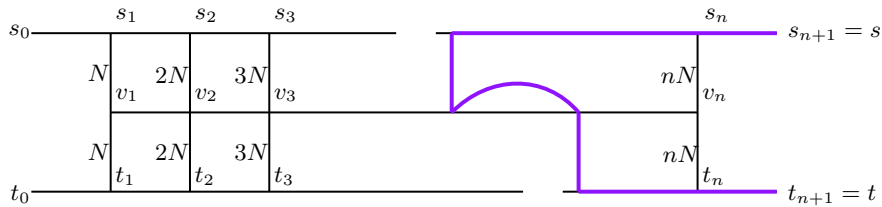
\begin{center}
    \tikzset{every picture/.style={line width=0.75pt}} 

\begin{tikzpicture}[x=0.75pt,y=0.75pt,yscale=-1,xscale=1]

\draw [color={rgb, 255:red, 0; green, 0; blue, 0 }  ,draw opacity=0.2 ]   (80,160) -- (328,160) ;
\draw    (120,120) -- (416,120) ;
\draw [color={rgb, 255:red, 0; green, 0; blue, 0 }  ,draw opacity=0.2 ]   (80,80) -- (264,80) ;
\draw    (120,80) -- (120,160) ;
\draw    (160,80) -- (160,160) ;
\draw    (200,80) -- (200,160) ;
\draw    (416,80) -- (416,160) ;
\draw [color={rgb, 255:red, 0; green, 0; blue, 0 }  ,draw opacity=0.2 ]   (284,80) -- (456,80) ;
\draw [color={rgb, 255:red, 0; green, 0; blue, 0 }  ,draw opacity=0.2 ]   (348,160) -- (456,160) ;
\draw [color={rgb, 255:red, 144; green, 19; blue, 254 }  ,draw opacity=0.5 ][line width=1.5]    (292,80) -- (456,80) ;
\draw [color={rgb, 255:red, 144; green, 19; blue, 254 }  ,draw opacity=0.5 ][line width=1.5]    (292,80) -- (292,120) ;
\draw [color={rgb, 255:red, 144; green, 19; blue, 254 }  ,draw opacity=0.5 ][line width=1.5]    (356,120) -- (356,160) ;
\draw [color={rgb, 255:red, 144; green, 19; blue, 254 }  ,draw opacity=0.5 ][line width=1.5]    (292,120) .. controls (311,100.92) and (339,100.26) .. (356,120) ;
\draw [color={rgb, 255:red, 144; green, 19; blue, 254 }  ,draw opacity=0.5 ][line width=1.5]    (356,160) -- (456,160) ;

\draw (460,155.4) node [anchor=north west][inner sep=0.75pt]  [font=\footnotesize]  {$t_{n+1} =t$};
\draw (68,155.4) node [anchor=north west][inner sep=0.75pt]  [font=\footnotesize]  {$t_{0}$};
\draw (417,106.4) node [anchor=north west][inner sep=0.75pt]  [font=\footnotesize]  {$v_{n}$};
\draw (121,106.4) node [anchor=north west][inner sep=0.75pt]  [font=\footnotesize]  {$v_{1}$};
\draw (460,75.4) node [anchor=north west][inner sep=0.75pt]  [font=\footnotesize]  {$s_{n+1} =s$};
\draw (67,74.4) node [anchor=north west][inner sep=0.75pt]  [font=\footnotesize]  {$s_{0}$};
\draw (418,66.4) node [anchor=north west][inner sep=0.75pt]  [font=\footnotesize]  {$s_{n}$};
\draw (417,146.4) node [anchor=north west][inner sep=0.75pt]  [font=\footnotesize]  {$t_{n}$};
\draw (121,146.4) node [anchor=north west][inner sep=0.75pt]  [font=\footnotesize]  {$t_{1}$};
\draw (161,146.4) node [anchor=north west][inner sep=0.75pt]  [font=\footnotesize]  {$t_{2}$};
\draw (201,146.4) node [anchor=north west][inner sep=0.75pt]  [font=\footnotesize]  {$t_{3}$};
\draw (161,106.4) node [anchor=north west][inner sep=0.75pt]  [font=\footnotesize]  {$v_{2}$};
\draw (201,106.4) node [anchor=north west][inner sep=0.75pt]  [font=\footnotesize]  {$v_{3}$};
\draw (121,66.4) node [anchor=north west][inner sep=0.75pt]  [font=\footnotesize]  {$s_{1}$};
\draw (161,66.4) node [anchor=north west][inner sep=0.75pt]  [font=\footnotesize]  {$s_{2}$};
\draw (201,66.4) node [anchor=north west][inner sep=0.75pt]  [font=\footnotesize]  {$s_{3}$};
\draw (107,94.4) node [anchor=north west][inner sep=0.75pt]  [font=\footnotesize]  {$N$};
\draw (107,134.4) node [anchor=north west][inner sep=0.75pt]  [font=\footnotesize]  {$N$};
\draw (141,95.4) node [anchor=north west][inner sep=0.75pt]  [font=\footnotesize]  {$2N$};
\draw (141,134.4) node [anchor=north west][inner sep=0.75pt]  [font=\footnotesize]  {$2N$};
\draw (181,95.4) node [anchor=north west][inner sep=0.75pt]  [font=\footnotesize]  {$3N$};
\draw (181,134.4) node [anchor=north west][inner sep=0.75pt]  [font=\footnotesize]  {$3N$};
\draw (396,95.4) node [anchor=north west][inner sep=0.75pt]  [font=\footnotesize]  {$nN$};
\draw (397,135.4) node [anchor=north west][inner sep=0.75pt]  [font=\footnotesize]  {$nN$};

\end{tikzpicture}
\end{center}

\begin{lemma}\label{lemma:2FRP}
    
For any $1 \leq i,j \leq n$, let $D=\set{(s_{i-1},s_{i}),(t_{j-1},t_{j})}$ be a set of two edge failures with one in $S$ and one in $T$. Then regarding the failure set $D$:
    \[ |\hg{s}{t}{D}| = |\pi_G(v_i,v_j)| + (i+j)N.\]
    
\end{lemma}

\begin{proof}

    For any $s$-$t$ path in $H-D$, we observe that there must be an edge $(s_{a},v_{a})$ for some $i \leq a \leq n$ and an edge $(v_{b},t_{b})$ for some $j \leq b \leq n$ that is taken. Consider the case that $a>i$ or $b>j$, then $(a+b) > (i+j)$ and the path will be of length at least $(a+b)N > |\pi_G(v_i,v_j)| + (i+j)N$, where $|\pi_G(v_i,v_j)| + (i+j)N$ is the length of the $s$-$t$ path $P = \pi_H(s,s_{i}) \circ \set{(s_i,v_i)} \circ \pi_G(v_i,v_j) \circ \set{(v_j,t_j)} \circ \pi_H(t_j,t)$ in $H$. Therefore, we know the shortest path in $H-D$ must take two edges $(s_{i},v_{i})$ and $(v_{j},t_{j})$, and thus $|\hg{s}{t}{D}| = |\hg{v_i}{v_j}{D}| + (i+j)N$.

    Now consider $\hg{v_i}{v_j}{D}$. If it takes any edge in the matching $E_1 \cup E_2$, then $|\hg{v_i}{v_j}{D}| \geq N > \pi_G(v_i,v_j)$. Therefore, $\hg{v_i}{v_j}{D}$ can only take edges in $G$ and thus we know $|\hg{v_i}{v_j}{D}| = |\pi_G(v_i,v_j)|$. Therefore, we obtain $|\hg{s}{t}{D}| = |\pi_G(v_i,v_j)| + (i+j)N$.
    
\end{proof}

With this lemma, we can prove Theorem \ref{theorem:2FRP}. Considering any undirected weighted all-pairs shortest path instance $G$, we construct an undirected weighted graph $H$ as above. Now we check all failure set $D=\set{(s_{i-1},s_{i}),(t_{j-1},t_{j})}$ for all pairs $(i,j)$, $1 \leq i,j \leq n$. By Lemma \ref{lemma:2FRP}, we can obtain $|\pi_G(v_i,v_j)|$ from $|\hg{s}{t}{D}|$, which means by solving 2FRP in $H$ we can solve all-pairs shortest paths in $G$.

Since we know $n' = 2(n+2)+n=O(n)$ nodes in $H$, if we have a truly subcubic time algorithm that solving 2FRP in time $O((n')^{3-\epsilon}) = O(n^{3-\epsilon})$ for $H$, then using this algorithm we can solve all-pairs shortest paths in $O(n^{3-\epsilon})$ time for $G$. Therefore, Theorem \ref{theorem:2FRP} is true.
\section{Almost Optimal Single-Source Replacement Path under Multiple Edge Failures}\label{sec:2ssrp}

For a graph $G$ and a vertex $s,$ the 2-fault single-source replacement path problem is that for every tuple of two edges $d_1,d_2$ and a vertex $t,$ answer the distance between $s$ and $t$ in $G$ after removing edges $d_1$ and $d_2.$ 
We show how to solve it in almost optimal time $\too{n^3}$ for undirected weighted graphs, via the incremental DSO.

Let $T$ be the shortest path tree of $s$ in $G,$ with edges $e_1,e_2,\cdots,e_{n-1}$. Our algorithm goes as follows: we maintain a dynamic DSO on $G.$ For each $1\le i<n,$ we remove $e_i$ in $G,$ maintain the dynamic DSO. For every vertex $t_j$ in the subtree of $e_i$ in $T$ and every edge $d_k$ on the replacement path $\og{s}{t_j}{e_i},$ save the distance $\pi_{(G-e_i)-d_k}(s,t_j)$ in $G-{e_i}$ via querying the dynamic DSO. At the end of each step, add $e_i$ back into $G.$

\paragraph{Correctness.}

For each tuple $(d_1,d_2,t),$ if none of $d_1,d_2$ is in $T,$ we know the shortest path $st$ is the same after removing $d_1,d_2.$ W.l.o.g., suppose $d_1\in T,$ and $t$ is in the subtree of $d_1$ in $T.$ If $d_2$ is not in $\og{s}{t}{d_1},$ we know that $\og{s}{t}{d_1-d_2}=\og{s}{t}{d_1},$ and this value can be queried in the static DSO. If $d_2\in \og{s}{t}{d_1},$ we know the value equals $\pi_{(G-d_1)-d_2}(s,t)$ from the dynamic DSO in graph $G-d_1,$ which is obtained in our algorithm.

\paragraph{Time Analysis.}

The incremental DSO in Section \ref{sec-inc} can be viewed as an offline fully dynamic DSO by Theorem \ref{thm1-2}. Therefore, the update time is $\too{n^2}$ for each edge $e_i.$ The number of edges on a shortest path tree is $O(n),$ so the total update time is $\too{n^3}.$ The number of distances we query is $O(n^3),$ each with query time $\too{1}$, so the total query time is also $\too{n^3}.$ We also retrieve all 1-fault single-source replacement paths $\og{s}{t_j}{e_i}$ in the static DSO, and the time is also bounded by $\too{n^3}.$

\paragraph{Extend to $f$-fault SSRP for $f\geq 2$.}

As the $f$FRP reduction in~\cite{WWX22}, when we have an $f$-fault SSRP algorithm in $O(n^s)$ time, we can construct an $(f+1)$-fault SSRP algorithm in $O(n^{s+1})$ time, by computing the $f$-fault SSRP in graph $G-e$ for every $e$ in the shortest path tree from $s$. Thus, the undirected 2-fault SSRP algorithm in $\tilde{O}(n^3)$ time can be extended to undirected $f$-fault SSRP algorithm in $\tilde{O}(n^{f+1})$ time for all $f\geq 2$, which is almost optimal.

\vspace{5mm}
\bibliographystyle{plain}
\bibliography{ref}

\end{document}